%% file: main.tex
\pdfoutput=1
\documentclass[a4paper,USenglish,cleveref,autoref]{lipics-v2021}

\AtBeginDocument{\hypersetup{pdfsubject={Licensed under Creative Commons License CC-BY 4.0}}}
\hideLIPIcs  
\makeatletter\let\@oddfoot\@empty\makeatother 

\title{Space Complexity Dichotomies for Subgraph Finding Problems in the Streaming Model} 

\titlerunning{Streaming Space Complexity Dichotomies for Subgraph Finding Problems} 

\author{Yu-Sheng Shih}{National Taiwan University, Taiwan}{r13922048@csie.ntu.edu.tw}{https://orcid.org/0009-0005-4826-0383}{}
\author{Meng-Tsung Tsai}{Academia Sinica, Taiwan}{mttsai@iis.sinica.edu.tw}{https://orcid.org/0000-0002-2243-8666}{This research was supported in part by the National Science and Technology Council, Taiwan, under grants NSTC 114-2221-E-001-023 and NSTC 115-2221-E-001-017-MY2.}
\author{Yen-Chu Tsai}{National Taiwan University, Taiwan}{b11902003@ntu.edu.tw}{}{}
\author{Ying-Sian Wu}{National Taiwan University, Taiwan}{b11201018@ntu.edu.tw}{https://orcid.org/0000-0002-1856-4918}{}

\authorrunning{Y.-S. Shih, M.-T. Tsai, Y.-C. Tsai, and Y.-S. Wu} 

\Copyright{Yu-Sheng Shih, Meng-Tsung Tsai, Yen-Chu Tsai, and Ying-Sian Wu} 



\category{} 

\relatedversion{} 

\acknowledgements{We used ChatGPT to improve the exposition and simplify proofs in this paper. ChatGPT suggested derandomizing an earlier version of \cref{novel}. The authors critically reviewed and verified the resulting arguments, including the derandomization, and retained full editorial control. The authors take full responsibility for the content.}

\nolinenumbers 

\EventEditors{}
\EventNoEds{0}
\EventLongTitle{}
\EventShortTitle{}
\EventAcronym{}
\EventYear{}
\EventDate{}
\EventLocation{}
\EventLogo{}
\SeriesVolume{}
\ArticleNo{1}

\input{packages}
\input{commands}

\begin{document}

\maketitle
\begin{abstract}
We study the space complexity of four variants of the standard subgraph finding problem in the streaming model. Specifically, given an $n$-vertex input graph and a fixed pattern graph on at least three vertices, we consider two settings: undirected simple graphs, denoted by $G$ and $H$, and oriented graphs, denoted by $\vec{G}$ and $\vec{H}$. Depending on the setting, the task is to decide whether $G$ contains $H$ as a subgraph or as an induced subgraph, or whether $\vec{G}$ contains $\vec{H}$ as a subgraph or as an induced subgraph. Let $\textsc{Sub}(H)$, $\textsc{IndSub}(H)$, $\textsc{Sub}(\vec{H})$, and $\textsc{IndSub}(\vec{H})$ denote these four variants, respectively. 

An oriented graph is well-oriented if it admits a bipartition in which every arc is oriented from the same part to the other, and a vertex is internal if both its in-degree and out-degree are non-zero. For each variant, we obtain a complete dichotomy theorem, briefly summarized as follows.
\begin{itemize}
\item $\textsc{Sub}(H)$ can be solved by an $\tilde{O}(1)$-pass $n^{2-\Omega(1)}$-space algorithm if and only if $H$ is bipartite.
\item $\textsc{IndSub}(H)$ can be solved by an $\tilde{O}(1)$-pass $n^{2-\Omega(1)}$-space algorithm if and only if $H$ is $P_3$, $P_4$, or $\mbox{co-}P_3$.
\item $\textsc{Sub}(\vec{H})$ can be solved by a deterministic single-pass $n^{2-\Omega(1)}$-space algorithm if and only if every connected component of $\vec H$ is either a well-oriented bipartite graph or a tree containing at most one internal vertex.
\item $\textsc{IndSub}(\vec{H})$ can be solved by an $\tilde{O}(1)$-pass $n^{2-\Omega(1)}$-space algorithm if and only if the underlying undirected simple graph $H$ is a $\mbox{co-}P_3$.
\end{itemize}
All our upper bounds are attained by deterministic algorithms that find a copy, whereas all our lower bounds hold even for randomized algorithms that only detect one.
\end{abstract}

\clearpage
\begin{nolinenumbers}
\begingroup
\small
\setcounter{tocdepth}{2}
\section*{Contents}
\makeatletter
\@starttoc{toc}
\makeatother
\endgroup
\end{nolinenumbers}

\thispagestyle{empty}
\clearpage
\setcounter{page}{1}

\input{intro}

\input{prel}

\input{section3}

\input{oriented}

\input{induced-main}

\input{openproblem}

\bibliography{ref}

\clearpage

\appendix

\input{AppendixA}
\input{AppendixB}

\input{AppendixD}

\end{document}

%% file: packages.tex
\usepackage{svg}
\usepackage[table]{xcolor}
\usepackage{todonotes}
\usepackage{amsthm}
\usepackage{thmtools, thm-restate}
\usepackage{paralist}
\usepackage{float}
\usepackage{algorithm}
\usepackage{algpseudocode}
\newcommand{\algcom}[1]{\hfill{\footnotesize\textit{#1}}}

\usepackage{booktabs,tabularx}
\usepackage{mdframed}

\hypersetup{colorlinks=true,
 linkcolor=red,
citecolor=blue,
urlcolor=magenta
}

\usepackage{enumitem}

\usepackage{multicol}

%% file: commands.tex
\definecolor{purple}{HTML}{1A1AB3}

\DeclareEmphSequence{\color{purple}\itshape}

\AtBeginEnvironment{algorithm}{%
  \DeclareEmphSequence{\itshape}
}
\AtEndEnvironment{algorithm}{%
  \DeclareEmphSequence{\color{purple}\itshape}
}

\AtBeginEnvironment{thebibliography}{%
  \DeclareEmphSequence{\itshape}
}
\AtEndEnvironment{thebibliography}{%
  \DeclareEmphSequence{\color{purple}\itshape}
}

\theoremstyle{definition}
\newenvironment{streamingproblem}[1][]{\begin{definition}[Problem: #1]}{\end{definition}}
\newtheorem{Question}{Question}
\crefname{Question}{Question}{Questions}
\Crefname{Question}{Question}{Questions}
\crefname{question}{Question}{Questions}
\Crefname{question}{Question}{Questions}

\usepackage[most]{tcolorbox}

\newtcolorbox{graytheorembox}{
    enhanced,
    breakable,
    colback=gray!10,
    colframe=gray!50,
    boxrule=0.5pt,
    arc=2pt,
    left=6pt,
    right=6pt,
    top=6pt,
    bottom=6pt
}

%% file: intro.tex
\section{Introduction}\label{sec:intro}

Many real-world applications require finding prescribed subgraphs or induced subgraphs, called \emph{patterns}, in streams of graph edges or digraph arcs. Here, $C_\ell$ and $P_\ell$ denote the cycle and path on $\ell$ vertices, respectively, $K_s$ denotes the complete graph on $s$ vertices, and $K_{s,t}$ denotes the complete bipartite graph with parts of sizes $s$ and $t$. For undirected subgraphs, copies of $C_4=K_{2,2}$, called butterflies, capture shared behavior in user–item networks and support recommendation and fraud analysis~\cite{Cai23}. For directed subgraphs, directed paths model lateral movement in cyberattacks~\cite{Choudhury15}, while cyclic orientations of $C_3$ and $C_5$ can indicate fraudulent transactions~\cite{Qiu18}. An orientation of $C_4$ consisting of two directed two-arc paths with common endpoints also provides a signal for social recommendations~\cite{Gupta14}. For undirected induced subgraphs, counts of small graphlets, such as $P_3$, $C_4$, and $K_4-e$ (the complete graph on four vertices with one edge removed), provide structural features for anomaly detection~\cite{Yang24}. 

Despite these applications, the space complexity of finding patterns in a sequence of graph edges or digraph arcs, called a \emph{data stream}, remains incompletely understood. A \emph{pass} is a sequential scan of the stream from beginning to end. Even for undirected subgraphs in arbitrary-order insertion-only streams, different patterns exhibit substantially different space requirements. Detecting $C_3$ requires $\Omega(n^2)$ bits for randomized algorithms with constant error probability and a constant number of passes, making storage of the entire adjacency matrix asymptotically optimal~\cite{BravermanOV13}. Similar quadratic lower bounds for a constant number of passes hold for every fixed clique $K_r$ with $r\ge 3$ and every fixed odd cycle~\cite{BeraC17}. In contrast, a \emph{na\"ive algorithm} for finding a fixed undirected subgraph $H$ stores the first $\text{ex}(n,H)+1$ edges, or the entire stream if it is shorter, and searches the stored graph for a copy of $H$. Here, the \emph{Tur\'an number} $\text{ex}(n,H)$ is the maximum number of edges in an $n$-vertex simple graph containing no copy of $H$ as a subgraph. This algorithm uses $O((\text{ex}(n,H)+1)\log n)$ bits and one pass, since exceeding the threshold guarantees a copy of $H$. In particular, $\text{ex}(n,C_4)=\Theta(n^{3/2})$ gives an $O(n^{3/2}\log n)$-bit algorithm, while $\text{ex}(n,C_{2\ell})=O(n^{1+1/\ell})$~\cite{BondySimonovits} gives an $O(n^{1+1/\ell}\log n)$-bit algorithm for every fixed $\ell\ge 2$. These examples connect pattern-finding complexity to extremal graph theory and raise the question of when the na\"ive algorithm is optimal. Combined with~\cite{BeraC17,BravermanOV13}, our \cref{thmUNB,thmUB} establish that, for every fixed $\ell\ge 3$, detecting $C_\ell$ requires $\Omega(\text{ex}(n,C_\ell))$ bits for algorithms using a constant number of passes, matching the na\"ive upper bound up to a logarithmic factor. Throughout this paper, \emph{detecting} a pattern means deciding whether a copy exists, whereas \emph{finding} a pattern means returning one copy if one exists and reporting its absence otherwise. Our algorithms solve the finding problem, while our lower bounds hold even for detection.

These results suggest a dichotomy: some fixed patterns require quadratic space, making storage of the entire graph asymptotically optimal, whereas others can be found using subquadratic space. Where does the boundary lie? More intriguingly, the examples above suggest that the Turán number may govern the space complexity of finding a pattern $H$, potentially yielding $\widetilde{\Theta}(\text{ex}(n,H))$ bits for algorithms making a constant number of passes.\footnote{The notation $\widetilde{O}$, $\widetilde{\Omega}$, and $\widetilde{\Theta}$ suppresses factors polynomial in $\log n$.} This motivates searching for algorithms that improve on the Tur\'{a}n-number benchmark. We call an algorithm that finds $H$ using $o(\text{ex}(n,H))$ bits a \emph{clever algorithm}, in contrast to the na\"ive algorithm described above. We therefore ask both which patterns admit subquadratic-space algorithms and which admit algorithms that beat the Tur\'{a}n-number benchmark, as formalized below.
\begin{graytheorembox}
 \begin{question}
      Which fixed patterns require $\Omega(n^2)$ bits for detection, and which admit algorithms that find a copy using $o(n^2)$ bits?
 \end{question}
\begin{question}
    For every fixed connected undirected graph $H$ and every orientation $\vec H$ of $H$, does detection require $\Omega(\text{ex}(n,H))$ and $\Omega(\text{ex}(n,\vec H))$ space, respectively, or can a clever algorithm find a copy using asymptotically less space?
\end{question}
Both questions concern arbitrary-order insertion-only streams and a constant number of passes. Question~2 focuses on connected patterns because color-coding allows disconnected subgraph patterns to be handled componentwise.
\end{graytheorembox}

In this paper, we address both questions. For Question~1, we establish space-complexity dichotomies for undirected subgraphs with a constant number of passes, oriented subgraphs with one pass, and induced patterns with polylogarithmically many passes. For Question~2, we identify oriented patterns that admit clever single-pass algorithms. Before formally stating our main results, we define the four variants and the model of computation.

\begin{streamingproblem}[Pattern Finding]
\label{prob:main}
~
\begin{itemize}[label={},leftmargin=0pt]
    \item \textbf{Input:} An $n$-vertex host graph and a fixed-size pattern graph, either both undirected, denoted by $G$ and $H$, or both \emph{oriented},\footnote{An oriented graph is a digraph in which each pair of vertices is connected by at most one arc.} denoted by $\vec{G}$ and $\vec{H}$. The pattern has at least three vertices and may be disconnected.\footnote{For ordinary subgraph finding, color coding reduces finding a disconnected pattern to finding its connected components in disjoint color classes. For induced subgraph finding, this reduction does not suffice: one must also ensure that no edges or arcs join the selected copies of distinct components.}
    \item \textbf{Output:} A copy of the pattern in the host graph, subject to one of the following requirements, or a report that no such copy exists:
    \begin{itemize}
        \item $\textsc{Sub}(H)$: find $H$ as a subgraph of $G$.
        \item $\textsc{Sub}(\vec{H})$: find $\vec{H}$ as a subgraph of $\vec{G}$.
        \item $\textsc{IndSub}(H)$: find $H$ as an induced subgraph of $G$.
        \item $\textsc{IndSub}(\vec{H})$: find $\vec{H}$ as an induced subgraph of $\vec{G}$.
    \end{itemize}
    A subgraph is specified by an injective mapping from the pattern vertices to the host vertices that preserves adjacency and, in the oriented variants, arc directions. For the induced variants, the mapping also preserves nonadjacency.
\end{itemize}
\end{streamingproblem}

\begin{definition}[Insertion-Only Graph Streaming Model~\cite{McGregor14}]
The edges of the input graph $G$, or the arcs of the input oriented graph $\vec{G}$, arrive one at a time in an arbitrary order, forming the input data stream. Each edge or arc appears exactly once. The algorithm is given the number $n$ of vertices and the fixed pattern $H$ or $\vec{H}$ in advance. It may scan the stream from beginning to end at most $p$ times, without reading backward, and uses at most $S$ bits of working memory. We call such an algorithm a \emph{$p$-pass $S$-space streaming algorithm}. Since the stream contains only insertions and no deletions, this is the \emph{insertion-only streaming model}.
\end{definition}

Our dichotomy results are summarized in
\cref{tab:dichotomy,tab:induceddichotomy}. All upper bounds are
attained by deterministic algorithms that find a copy of the
pattern, whereas all lower bounds hold even for randomized
detection with success probability at least $2/3$. Throughout,
$H$ denotes the underlying undirected graph of an oriented
pattern $\vec H$. A connected component of $\vec H$ is an induced subgraph $\vec H[V(C)]$ where $C$ is a connected component of $H$.

For $\textsc{Sub}(H)$, bipartiteness determines the boundary in
Question~1: non-bipartite patterns require $\Omega(n^2)$ bits
for a constant number of passes, whereas bipartite patterns
admit single-pass algorithms using
$\widetilde{O}(n^{2-1/\Delta'(H)})$ bits. Here, $\Delta'(H)$ is the maximum, over all connected components of $H$, of the smaller of the maximum degrees on their two bipartition sides.
The displayed exponent assumes that $H$ has at least one edge;
empty patterns are trivial.

For $\textsc{Sub}(\vec H)$, we obtain a complete single-pass
classification. A vertex is \emph{internal} if it has positive
in-degree and positive out-degree. An orientation is
\emph{well-oriented} (WO) if it has no internal vertices, and \emph{non-well-oriented} (NWO) otherwise. 
If $H$ is non-bipartite, the $\Omega(n^2)$-bit lower bound for $\textsc{Sub}(H)$ with a constant number of passes extends to $\textsc{Sub}(\vec H)$. If $\vec H$ is well-oriented and its underlying graph $H$ is bipartite, the single-pass $\widetilde{O}(n^{2-1/\Delta'(H)})$-bit upper bound for $\textsc{Sub}(H)$ also extends to $\textsc{Sub}(\vec H)$.
Among the remaining oriented patterns $\vec H$, $\Omega(n^2)$ bits are necessary in one pass if $\vec H$ has a connected component $\vec H[V(C)]$ which contains at least two internal vertices, or contains an internal vertex and $C$ contains a cycle. We denote these conditions by $\vec H\owns\mathrm{NWOC}_2$ and $\vec H\owns\mathrm{NWOC}_{1+C}$, respectively. 
All remaining patterns admit subquadratic-space single-pass
algorithms: $\widetilde{O}(n)$ bits when $H$ is a forest, and
$\widetilde{O}(n^{2-1/\Delta'(H)})$ bits otherwise.
\input{tables}

For induced patterns, we establish space-complexity dichotomies for algorithms using polylogarithmically many passes. The problems $\textsc{IndSub}(H)$ with $H\in \{P_3,P_4,\text{co-}P_3\}$ admit finding algorithms using $\widetilde{O}(n)$ bits, whereas every other undirected pattern requires $\widetilde{\Omega}(n^2)$ bits even for detection. For $\textsc{IndSub}(\vec H)$, the same upper bound holds when the underlying graph is $\text{co-}P_3$, and all other oriented patterns satisfy the same lower bound. Here, $\text{co-}P_3$ denotes the complement of $P_3$. The lower bounds follow from the $\Omega(n^2/p)$-bit bounds for $p$-pass algorithms, yielding a separation between near-linear and near-quadratic space when $p$ is polylogarithmic (see~\cref{thm:undirected-induced-dichotomy,thm:oriented-induced-dichotomy}).

\begin{graytheorembox}

For Question~2, it remains unknown whether any fixed connected undirected pattern $H$ admits a clever algorithm for $\textsc{Sub}(H)$ using $o(\text{ex}(n,H))$ bits and a constant number of passes. 

\medskip

\indent For oriented subgraphs, however, we exhibit patterns that admit clever single-pass algorithms, showing that the corresponding Tur\'{a}n number does not always give a space lower bound.
Consider an oriented forest $\vec H$ with exactly one internal vertex. By \cref{thmOSPC}, a copy of $\vec H$ can be found deterministically in one pass using $\widetilde{O}(n)$ bits. On the other hand, a balanced complete bipartite graph with every edge oriented from one side to the other is well-oriented and hence contains no copy of $\vec H$. Therefore, $\text{ex}(n,\vec H)=\Theta(n^2)$, and our algorithm uses $o(\text{ex}(n,\vec H))$ bits.

\medskip

\indent The na\"{i}ve algorithm does not extend to induced patterns, since later edge or arc insertions may destroy an induced copy found in the stored graph. Thus, Question~2, as formulated using this algorithm and its Tur\'{a}n-number benchmark, concerns only ordinary subgraphs.

\end{graytheorembox}

\subsection{Our Techniques}
Our lower bounds follow the standard approach in streaming complexity: we reduce communication problems, including $\textsc{Index}$, $\textsc{Set-Disjointness}$, and their variants, to pattern detection and translate communication lower bounds into space lower bounds. The novelty lies in the graph gadgets used in these reductions, whose construction exploits structural properties of the patterns. For example, for non-bipartite undirected patterns, we identify a largest \emph{shortest odd cycle component} and use it to construct the required gadget. Once the general reduction framework is established, the main task is therefore to design suitable gadgets for the different classes of patterns.

For upper bounds, our answer to Question~2 shows that some patterns admit clever algorithms that improve on the na\"{i}ve algorithm. For oriented trees with exactly one internal vertex, we introduce a novel \emph{$k$-core-like structure} to obtain a space-efficient single-pass algorithm. We also use a \emph{recursive cograph decomposition} for induced $P_4$ finding, obtaining a space-efficient algorithm with polylogarithmically many passes.

\subsection{Related Work}

Oostveen and van Leeuwen~\cite{OostveenL24} studied the
parameterized streaming complexity of diameter and connectivity,
including how a supplied small modulator to an induced-$H$-free
graph affects the complexity. Here, a modulator is a vertex set
$X$ such that $G-X$ contains no induced copy of $H$. Their
structural-parameter perspective complements our classification
by the forbidden pattern. Our induced-pattern results also
address the related task of checking whether a proposed
remainder is $H$-free when this property is not guaranteed.

A related complexity classification for induced-$H$-free graph recognition
was obtained by Kari et al. in the randomized simultaneous-messages broadcast
model (SM-BCAST)~\cite{KariMRS15}. For fixed undirected patterns $H$ on at least
three vertices, they obtained exactly the same classification as
Theorem~\ref{thm:undirected-induced-dichotomy}: $O(\log n)$-bit messages per vertex
suffice for the three tractable patterns, whereas every other pattern requires
messages of size $n^{\Omega(1)}$.
These three exceptional patterns correspond
to the recognition of cluster graphs, complete multipartite graphs, and
cographs, respectively.
Our reduction for $C_4$ also improves the SM-BCAST lower bound;
see Proposition~\ref{prop:c4-sm-bcast}.

A substantial line of work concerns approximately counting
pattern occurrences rather than finding a single copy.
Kane et al.~\cite{KaneMSS12} and Fichtenberger and
Peng~\cite{Fichtenberger022} give streaming algorithms for
counting arbitrary fixed subgraphs. Their space bounds depend
on the number of occurrences, allowing space savings when
sufficiently many copies are present. Numerous further results obtain sharper bounds for particular patterns such as triangles and short cycles \cite{BeraC17, McGregorVV16, KallaugherMPV19, McGregorV20}, some of them in alternative stream orders such as adjacency-list streams.

Halld\'{o}rsson et al.\ and Braverman et
al.~\cite{HalldorssonSSW12,BravermanLSVY18} studied streaming
clique approximation through the promise problem of
distinguishing graphs with clique number at least $r$ from
those with clique number at most $s$, where $r>s$. Equivalently,
one distinguishes the presence of $K_r$ from the absence of
$K_{s+1}$. Their results establish tight space bounds for this
promise problem as a function of the clique thresholds.

%% file: tables.tex
\begin{table}[!h]
\centering
\small
\setlength{\tabcolsep}{4pt}
\begin{tabular}{lllll}
\toprule
Problem & Condition & Space bound (bits) & Passes & Reference \\
\midrule
\multirow{2}{*}{\textsc{Sub}$(H)$}
& if $H$ is non-bipartite
& \cellcolor{red!15}$\Omega(n^2)$ necessary
& \cellcolor{red!15}$O(1)$
& \multirow{2}{*}{\cref{thmUNB}} \\
& otherwise
& \cellcolor{green!15}$\widetilde{O}(n^{2-1/\Delta'(H)})$ sufficient
& \cellcolor{green!15}$1$
& \\
\midrule
\multirow{6}{*}{\textsc{Sub}$(\vec{H})$}
& if $H$ is non-bipartite
& \cellcolor{red!15}$\Omega(n^2)$ necessary
& \cellcolor{red!15}$O(1)$
& \multirow{2}{*}{Corollary~\ref{ONB}} \\
& else if $\vec{H}$ is WO
& \cellcolor{green!15}$\widetilde{O}(n^{2-1/\Delta'(H)})$ sufficient
& \cellcolor{green!15}$1$
& \\
& else if $\vec{H}\owns\mathrm{NWOC}_2$
& \cellcolor{red!15}$\Omega(n^2)$ necessary
& \cellcolor{red!15}$1$
& \multirow{4}{*}{\cref{thmOSPC}} \\
& else if $\vec{H}\owns\mathrm{NWOC}_{1+C}$
& \cellcolor{red!15}$\Omega(n^2)$ necessary
& \cellcolor{red!15}$1$
& \\
& else if $H$ is a forest
& \cellcolor{green!15}$\widetilde{O}(n)$ sufficient
& \cellcolor{green!15}$1$
& \\
& otherwise
& \cellcolor{green!15}$\widetilde{O}(n^{2-1/\Delta'(H)})$ sufficient
& \cellcolor{green!15}$1$
& \\
\bottomrule
\end{tabular}
\caption{Dichotomy results for the subgraph variants of
\cref{prob:main}. For each problem, the first applicable row
 determines the bound. Red entries give necessary space for
randomized detection with success probability at least $2/3$
and the indicated pass allowance. Green entries give sufficient
space for deterministic finding using at most the indicated
number of passes.\vspace{-0.7cm}}
\label{tab:dichotomy}
\end{table}
\vspace{-0.5cm}
\begin{table}[!h]
\centering
\begin{tabular}{lllll}
\toprule
Problem & Condition & Space bound (bits) & Passes & Reference \\
\midrule
\multirow{2}{*}{\textsc{IndSub}$(H)$}
& if $H \notin \{P_3,P_4,\text{co-}P_3\}$
& \cellcolor{red!15}$\widetilde{\Omega}(n^2)$ necessary
& \cellcolor{red!15}$\widetilde{O}(1)$
& \multirow{2}{*}{\cref{thm:undirected-induced-dichotomy}} \\
& otherwise
& \cellcolor{green!15}$\widetilde{O}(n)$ sufficient
& \cellcolor{green!15}$\widetilde{O}(1)$
& \\
\midrule
\multirow{2}{*}{\textsc{IndSub}$(\vec{H})$}
& if $H \ne \text{co-}P_3$
& \cellcolor{red!15}$\widetilde{\Omega}(n^2)$ necessary
& \cellcolor{red!15}$\widetilde{O}(1)$
& \multirow{2}{*}{\cref{thm:oriented-induced-dichotomy}} \\
& otherwise
& \cellcolor{green!15}$\widetilde{O}(n)$ sufficient
& \cellcolor{green!15}$\widetilde{O}(1)$
& \\
\bottomrule
\end{tabular}
\caption{Dichotomy results for the induced variants of
\cref{prob:main}. Red entries give necessary space for randomized
detection with success probability at least $2/3$, even with
polylogarithmically many passes. Green entries give sufficient
space for deterministic finding using polylogarithmically many
passes.}
\label{tab:induceddichotomy}
\end{table}

%% file: prel.tex
\section{Notation}
\theoremstyle{plain}  
\newtheorem*{thm}{Theorem}

We introduce notation common to all sections and defer background on Tur\'{a}n numbers, the communication complexity of the games used in our reductions, and color-coding techniques to \cref{AppendixA}. All graphs are simple: an undirected graph has no loops or parallel edges, and an oriented graph is obtained by assigning a direction to each edge of a simple undirected graph. We denote an undirected edge by an unordered pair $\{u,v\}$ and a directed edge, or \emph{arc}, from $u$ to $v$ by an ordered pair $(u,v)$. We use arrow notation, such as $\vec H$, for oriented graphs and write $H$ for the underlying undirected graph.

We define $R^p_{2/3}(\textsc{Sub}(H))$ to be the minimum space complexity of a randomized $p$-pass streaming algorithm for $\textsc{Sub}(H)$ with success probability at least $2/3$, and $D^p(\textsc{Sub}(H))$ to be the minimum space complexity of a deterministic $p$-pass streaming algorithm for $\textsc{Sub}(H)$, likewise for the other three types of subgraph problems.

\if 0\begin{definition}
    Let $\vec H$ be an oriented bipartite graph. We say that $H$ is well-oriented if there is a bipartition of $H$ into parts $A(H)$, $B(H)$ such that 
    all arcs are oriented from $A(H)$ to $B(H)$, and non-well-oriented otherwise. We always write WO for well-oriented and NWO for non-well-oriented.
    
    Also, in an oriented graph $\vec H$, an internal vertex in $\vec H$ is a vertex whose in-degree and out-degree in $\vec H$ are both nonzero. 
\end{definition}
\fi

%% file: section3.tex
\section{Undirected Simple Graphs}
\label{sec:section3}

All graphs in this section are undirected and simple unless otherwise stated. 
We establish a dichotomy for \(\textsc{Sub}(H)\), according to whether the fixed pattern \(H\) is bipartite.

\begin{theorem}\label{thmUNB}
    $\textsc{Sub}(H)$ for a fixed graph $H$  
    admits the dichotomy:

    \begin{itemize}
        \item If $H$ is non-bipartite, then $R^p_{2/3}(\textsc{Sub}(H))=\Omega(n^2/p)$.
        \item Otherwise, $D^1(\textsc{Sub}(H))=\widetilde{O}(n^{2-1/\Delta'(H)})$.
    \end{itemize}
\end{theorem}
The upper bound for bipartite patterns follows from the
Tur\'an-number bound  and the na\"ive algorithm described in~\cref{sec:intro}. It therefore remains to prove the lower
bound for non-bipartite patterns.

To this end, we define a class of graphs called \emph{H-gadgets}, and use them to establish in \cref{sec:sub-space}  the desired $p$-pass $\Omega(n^2/p)$-bit space lower bound. We show in \cref{sec:sub-gadget} that,  if $H$ is non-bipartite, then for every sufficiently large integer $n$, there is an $n$-vertex $H$-gadget with parameters $(k,r)$, where $k=\Omega_H(n^2)$ and $r\ge2$ is some constant depending only on $H$.

\begin{definition}[$H$-gadget]\label{transgadget}
Let $k\ge1$ and $r\ge2$ be integers. A graph $G=(V,E)$ is an \emph{$H$-gadget} with parameters $(k,r)$ if it has pairwise disjoint edge sets $E_1,\ldots,E_k$, each of size $r$, satisfying the conditions below simultaneously. Write $F=E\setminus\bigcup_{i\in[k]}E_i$.
\begin{itemize}
    \item For every $i\in[k]$, the graph $(V,F\cup E_i)$ contains $H$ as a subgraph.
    \item For every set $Y$ consisting of exactly one edge from each $E_i$ for $i\in[k]$, the graph $(V,F\cup Y)$ is \emph{$H$-free}, meaning that it does not contain $H$ as a subgraph.
\end{itemize}
\end{definition}

\subsection{Space Lower Bounds for \texorpdfstring{$\textsc{Sub}(H)$}{Sub(H)}}\label{sec:sub-space}

Suppose for every sufficiently large integer $n$, there is an $n$-vertex $H$-gadget with parameters $(k,r)$, where $k=\Omega(n^2)$ and $r$ depends only on $H$, then we can reduce \textsc{Multi-Disjointness} to $\textsc{Sub}(H)$. Later in \cref{sec:sub-gadget} we shall see that all non-bipartite $H$ satisfy this condition.
The communication lower bounds on \textsc{Multi-Disjointness} can be found in \cref{prop:CCbound}.

\begin{lemma}\label{lem:subgraph-space}
Let $H$ be a fixed graph. If for every sufficiently large integer $n$, there exists an $n$-vertex $H$-gadget with parameters $(k,r)$ where $k=\Omega(n^2)$ and $r$ depends only on $H$, then $R^p_{2/3}(\textsc{Sub}(H))=\Omega(n^2/p)$.
\end{lemma}
\begin{proof}
    For every sufficiently large $n$, fix an $n$-vertex $H$-gadget $G$ with parameters $(k,r)$, where $k=\Omega(n^2)$ and $r\ge2$ depends only on $H$. Let
$E_i=\{e_{i,1},\ldots,e_{i,r}\} \mbox{ for each } i\in[k]$,
and write $F=E(G)\setminus\bigcup_{i\in[k]}E_i$. These choices are known to all parties.

Given an instance of \textsc{Multi-Disjointness}, party $q$ contributes the edges $e_{i,q}$ for $i\in S_q$. Together with $F$, these edges form
\[
G_{\mathcal S}
=
\Big(
V(G),
F\cup\bigcup_{q\in[r]}\{e_{i,q}:i\in S_q\}
\Big).
\]
In a YES instance, a common element $i$ ensures that all edges of $E_i$ are present, so $G_{\mathcal S}$ contains $H$ as a subgraph. In a NO instance, each $E_i$ contributes at most one edge. Extend the contributed edges to a set $Y$ containing exactly one edge from every $E_i$. The gadget property makes $(V(G),F\cup Y)$ $H$-free, so its subgraph $G_{\mathcal S}$ is also $H$-free.

Use the stream consisting of $F$, followed by each party's contributed edges in order. A randomized $p$-pass streaming algorithm for $\textsc{Sub}(H)$ using $b$ bits of space therefore yields a protocol with communication $O(prb)$ and the same error probability. By the communication lower bound on \textsc{Multi-Disjointness}, we get $
b=\Omega\!\left(k/pr\right)
 =\Omega\!\left(n^2/p\right).
$
\end{proof}

\subsection{Existence of \texorpdfstring{$H$}{H}-Gadgets}\label{sec:sub-gadget}

Let $H$ be a non-bipartite graph with \emph{odd girth} $\ell$, the length of a shortest odd cycle in $H$. For now, assume that $H$ is connected and that every vertex of $H$ lies on a cycle of length $\ell$; we remove these assumptions later.\\

\textbf{\textit{The special case.}} Let $t \ge 1$ be an integer and $s$ be an arbitrary vertex in $H$. Form $G$ from $t$ vertex-disjoint copies of $H$, writing $v^{(i)}$ for the copy of $v$ in the $i$th copy $H^{(i)}$ of \nolinebreak $H$. For every neighbor $u$ of $s$ in $H$, add the edge $\{s^{(i)},u^{(j)}\}$ for every $i \ne j \in[t]$. For each ordered pair $(i,j)\in[t]^2$, define
\[
E_{(i-1)t+j}
=
\bigl\{\{s^{(i)},u^{(j)}\}:u \mbox{ is a neighbor of  } s \mbox{ in } H\bigr\}.
\]
These edge sets are pairwise disjoint and each has size $r$ equal to the degree of $s$ in $H$, noting that $r \ge 2$ as $s$ is contained in some $C_\ell$ in $H$. Moreover,
\[
F=E(G)\setminus\bigcup_{x\in[t^2]}E_x
\]
consists precisely of the edges of the $t$ copies of $H-s$.

We show that $G$ is an $H$-gadget with parameters $(t^2,r)$. For $x=(i-1)t+j$, the graph $(V(G),F\cup E_x)$ contains $H$ as a subgraph, obtained by mapping $s$ to $s^{(i)}$ and every other vertex $u$ to $u^{(j)}$.

Now let $Y$ contain exactly one edge from each $E_x$. The natural projection $u^{(i)}\mapsto u$ maps every $C_\ell$ in $(V(G),F\cup Y)$ to a closed walk of length $\ell$ in $H$. By the minimality of $\ell$, this walk must itself be a simple $C_\ell$ in $H$. Consequently, a $C_\ell$ in $(V(G), F \cup Y)$ containing $s^{(i)}$ contains no other copy of $s$. Recall that deleting all copies of $s$ from $G$ leaves the disjoint union of $t$ copies of $H-s$. Since the cycle contains no copy of $s$ other than $s^{(i)}$, deleting $s^{(i)}$ from it leaves a path entirely within one of these copies of $H - s$, say the $j$th copy. Both cycle edges incident with $s^{(i)}$ therefore belong to $E_{(i-1)t+j}$, contradicting the choice of $Y$. Thus no copy of $s$ lies on a $C_\ell$.

Since every vertex of $H$ lies on a $C_\ell$, any copy of $H$ would avoid all copies of $s$. As $H$ is connected, it would consequently lie within one copy of $H-s$, which has too few vertices. Hence $(V(G),F\cup Y)$ is $H$-free. As a result, $G$ is an $H$-gadget with parameters $(\Omega(n^2), r)$.\\

\textbf{\textit{The general case.}} The construction above handles, in particular, the case where $H$ consists of a single \emph{shortest odd cycle component}, defined in~\cref{def:resocc}. We now extend it to an arbitrary non-bipartite graph using these components.

\begin{definition}[shortest odd cycle component (SOCC)]\label{def:resocc}
Let $H=(V,E)$ be a non-bipartite graph with odd girth $\ell$. We set
$
E_\ell=\{e\in E:e\text{ lies on some }C_\ell\text{ in }H\}
$
and define $H_\ell=(V,E_\ell)$. 
A shortest odd cycle component is the subgraph of $H$ induced by the vertex set of a connected component of $H_\ell$.
\end{definition}

Let $A\subseteq V(H)$ be the vertex set of an SOCC with the largest number of vertices. Since $|A|\ge\ell$, this SOCC is nontrivial: it is connected, and every vertex lies on a $C_\ell$ in $H[A]$. Hence $H[A]$ satisfies the assumptions of the preceding construction. Choose an arbitrary vertex $s\in A$, and let $G$ now denote the graph obtained by applying the preceding construction to $H[A]$ with the distinguished vertex $s$.

Add to $G$ a single copy of $H-A$. For every edge $\{v,w\}\in E(H)$ with $v\in A$ and $w\notin A$, add an edge $\{v^{(j)}, w\}$ to $G$ for every $j\in[t]$. 
\begin{figure}[H]
        \centering
        \includesvg[width=0.75\linewidth]{modified_thm1.svg}
        \caption{The $H$-gadget in the general case.}
        \label{fig1}
\end{figure}
Here we verify the first gadget condition. For $x=(i-1)t+j$, the graph $(V(G),F\cup E_x)$ contains $H$ as a subgraph using $s^{(i)}$, the $j$th copy of $A\setminus\{s\}$, and the single copy of $H-A$. 

We now verify the other gadget condition. Let $m=|A|$, and let $Y$ contain exactly one edge from each $E_x$. Suppose that $G_Y=(V(G),F\cup Y)$ contains $H$ as a subgraph. As in the special case, the closed-walk argument shows that $G_Y$ has odd girth at least $\ell$ and that every $C_\ell$ in $G_Y$ projects to a simple $C_\ell$ in $H$. Since each projected cycle lies within a single SOCC in $H$, no cycle $C_\ell$ in $G_Y$ crosses between the duplicated part and $H-A$. The special-case argument therefore still excludes all copies of $s$ from any cycle $C_\ell$. Consequently, every SOCC in the duplicated part has fewer than $m$ vertices, while the SOCCs in $H-A$ remain unchanged.

To obtain a contradiction, we count the SOCCs of maximum order. Let $q$ be the number of $m$-vertex SOCCs of $H$. Then $G_Y$ has exactly $q-1$ SOCCs of that order and none larger, which is fewer than the $q$ such SOCCs in $H$; as a copy of $H$ maps every $C_\ell$ to a $C_\ell$, this rules out a copy of $H$. Hence $G_Y$ is $H$-free, proving that the augmented graph is an $H$-gadget.

This gadget has $tm+|V(H)\setminus A|$ vertices. For any sufficiently large integer $n$, choose
$
t=\left\lfloor
{(n-|V(H)\setminus A|)}/{m}
\right\rfloor
$
and add isolated vertices to reach exactly $n$ vertices. Padding preserves the two gadget conditions. Since $H$ is fixed, $r=\deg_{H[A]}(s)\ge2$ is constant and $t^2=\Omega(n^2)$. We have therefore established the following result.
\begin{lemma}\label{lem:sub-gadget}
For every fixed non-bipartite simple undirected graph $H$, there is a constant $r\ge2$ depending only on $H$ such that, for every sufficiently large integer $n$, there exists an $n$-vertex $H$-gadget with parameters $(k,r)$, where $k=\Omega(n^2)$.
\end{lemma}
Combining \cref{lem:subgraph-space} and \cref{lem:sub-gadget}, we get the lower bound for every non-bipartite $H$.
\begin{proposition}\label{prop:subgraph-space}
For every fixed non-bipartite undirected  $H$, $R^p_{2/3}(\textsc{Sub}(H))=\nolinebreak \Omega(n^2/p)$.
\end{proposition}

%% file: oriented.tex
\section{Oriented Graphs}\label{sec:section5}
We now turn to subgraph detection for oriented graphs. As in the undirected setting, the Turán number gives an immediate upper bound on the space complexity. By Valadkhan's bound \cite{Valadkhan09} (\cref{thm:valadkhan}, see also \cite{Grzesik23}), an oriented graph $\vec H$ has $\text{ex}(n,\vec{H})=o(n^2)$ if and only if it is bipartite and well-oriented. Thus WO bipartite patterns admit subquadratic-space algorithms by the na\"ive algorithm.

Unlike in the undirected setting, however, a quadratic Turán number does not necessarily imply quadratic streaming space. Consider the directed path $u\to v\to w$. Its Turán number is quadratic, since a WO complete bipartite graph contains no such directed path. Nevertheless, it can be detected in one pass using 
$O(n\log n)$ bits by recording whether each vertex has positive in-degree and positive out-degree. Thus, among NWO bipartite patterns, quadratic Turán number alone does not determine the streaming complexity.

In this section, we establish a single-pass dichotomy of $\textsc{Sub}(\vec H)$. The non-bipartite and WO bipartite cases are straightforward. The former follows by reduction from $\textsc{Sub}(H)$, while the latter follows from the na\"ive algorithm. Thus it remains to classify NWO bipartite patterns in the single-pass setting, and our results are as follows.
\begin{theorem}\label{thmOSPC}
    Let $\vec H$ be an NWO bipartite graph.
    \begin{itemize}
        \item If some component of \(\vec H\) contains at least two internal vertices, or contains a cycle and exactly one internal vertex, then $R^1_{2/3}(\textsc{Sub}(\vec H))=\Omega(n^2)$.
        \item Otherwise, if $H$ is a forest, then $D^1(\textsc{Sub}(\vec H))=\tilde O(n)$.
        \item Otherwise, every connected component of \(\vec H\) is either WO or an oriented tree with exactly one internal vertex, and $D^1(\textsc{Sub}(\vec H))=\tilde O(n^{2-1/\Delta'(H)})$.
    \end{itemize}
\end{theorem}
The lower bounds in \cref{thmOSPC} follow by reductions from \(\textsc{Index}\) using suitable \(\vec H\)-gadgets. On the tractable side, the essential new components are oriented trees with exactly one internal vertex. We develop a deterministic single-pass \(\widetilde O(n)\)-space algorithm for them and combine it with the na\"ive algorithm for WO components, giving the bounds in \cref{thmOSPC}.

A complete multi-pass classification remains open. In Section~\ref{sec:nwo-multipass-partial}, we present partial results, including hardness for NWO even cycles and NWO orientations of $K_{s,t}$ with $2\leq s\leq t$, and tractability for NWO forests in a constant number of passes.

\subsection{Reduction from the Undirected Patterns}

We first show that any algorithm for a prescribed orientation \(\vec H\) can be used to detect its underlying graph \(H\), with only a constant-factor increase in the number of passes.

\begin{proposition}\label{prop:oriented-to-undirected}
Let \(H\) be a fixed undirected graph and let \(\vec H\) be any orientation of \(H\). If some randomized \(p\)-pass \(S\)-space streaming algorithm decides whether an input oriented graph contains \(\vec H\) as a subgraph with success probability at least \(2/3\), then some randomized \(O_H(p)\)-pass \(O(S+\log n)\)-space streaming algorithm decides whether an input undirected graph contains \(H\) as a subgraph with success probability at least \(2/3\).
\end{proposition}

\begin{proof}
Let \(\mathcal A\) be a randomized \(p\)-pass \(S\)-space algorithm deciding whether its input contains \(\vec H\), with success probability $2/3$, and let \(k=|E(H)|\). Given the input graph \(G\), we generate an oriented graph \(\vec G\) by assigning a random orientation to each edge. To keep the same orientation in different passes, we sample a \(k\)-wise independent hash function that, on each edge \(\{a,b\}\), determines whether it is oriented as \(a\to b\) or \(b\to a\). Such a hash function can be constructed with the classic polynomial method over finite fields using a random seed of \(O(k\log n)\) bits (see, e.g., \cite{ChristianiPaghFOCS2014}, which attributes the result to \cite{Joffe1974,WegmanCarter1981}).

Since \(\mathcal A\) may err on both sides, we first reduce its error. Running \(\mathcal A\) independently \(O(k)\) times on the same stream and taking the majority vote gives, by a Chernoff bound, an algorithm \(\mathcal A'\) with error probability at most \(\delta:=2^{-k}/12\). We then repeat the following for \(T:=2^{k+1}\) independent trials: draw a fresh random orientation \(\vec G\) and run \(\mathcal A'\) on it. We output YES if and only if at least one trial outputs YES.

If \(G\) contains a copy of \(H\), then in each trial its \(k\) edges are oriented consistently with \(\vec H\) with probability at least \(2^{-k}\), so no trial produces a copy of \(\vec H\) with probability at most \((1-2^{-k})^T\le e^{-2}<1/6\). If \(G\) is \(H\)-free, then no orientation of \(G\) contains \(\vec H\). In both cases, by a union bound, \(\mathcal A'\) errs in some trial with probability at most \(T\delta=1/6\). Hence the failure probability is at most \(1/3\) on YES-instances and at most \(1/6\) on NO-instances.

We run the trials sequentially, so their workspace can be reused. Each trial takes \(O(kp)\) passes, so the reduction uses \(O(S+k\log n)=O(S+\log n)\) space and \(O(kp2^k)=O_H(p)\) \nolinebreak passes.
\end{proof}

If \(\vec H\) is non-bipartite, then  \cref{prop:oriented-to-undirected} together with \cref{prop:subgraph-space}, whose proof already applies to detection, gives a $p$-pass \(\Omega(n^2/p)\)-space lower bound. On the other hand, suppose that \(\vec H\) is WO and bipartite, and let $\vec H_1,\ldots,\vec H_k$ be its connected components with at least one arc. For each $i$, \cref{Oex} applied to the bipartition of $\vec H_i$ into sources and sinks gives
$\text{ex}(n,\vec H_i)=O(n^{2-1/\Delta'(H_i)})$.
These bounds combine: an $n$-vertex oriented graph with more than $\sum_{i=1}^k\text{ex}(n,\vec H_i)+|V(H)|n$ arcs contains $\vec H$. Indeed, we can find a copy of $\vec H_1$, delete its vertices and hence at most $|V(H_1)|n$ arcs, and repeat with $\vec H_2,\ldots,\vec H_k$, using that $\text{ex}(n,\cdot)$ is nondecreasing in $n$; the isolated vertices of $\vec H$ are then mapped to unused vertices, which exist when $n\geq|V(H)|$. Hence
$\text{ex}(n,\vec H)=O(n^{2-1/\Delta'(H)})$ with $\Delta'(H)$ taken componentwise as in \cref{sec:intro}. Applying \cref{Oex} directly to a single sources/sinks bipartition of the whole of $\vec H$ would not suffice, since, unlike in the undirected case, the two sides of different components cannot be swapped independently. The Turán-number algorithm therefore yields a deterministic single-pass $\widetilde O(n^{2-1/\Delta'(H)})$-space algorithm. Thus we obtain the following bounds.

\begin{corollary}\label{ONB}
    Let $\vec H$ be a fixed oriented graph. 
    \begin{itemize}
        \item If $\vec H$ is non-bipartite, then $R^p_{2/3}(\textsc{Sub}(\vec H))=\Omega(n^2/p)$.
        \item If $\vec H$ is bipartite and WO, then $D^1(\textsc{Sub}(\vec H))=\tilde O(n^{2-1/\Delta'(H)})$.
    \end{itemize}
\end{corollary}

To complete the single-pass dichotomy, it remains to classify NWO bipartite patterns.

\subsection{Hard NWO Bipartite Patterns in the Single-Pass Model}\label{sec:oriented-lower-bounds}

We now prove the \(\Omega(n^2)\) lower bound for the NWO bipartite patterns in \cref{thmOSPC}. Each of the following structures makes $\vec H$ a hard instance: a component containing at least two internal vertices, and a component containing  a cycle and exactly one internal vertex.
As in \cref{sec:section3}, we first define \(\vec H\)-gadgets and show that their existence yields the desired single-pass lower bound. We then construct such gadgets for the two cases above.

\begin{definition}[$\vec H $-gadget]\label{ReductionOriented}
    Let $\vec H$ be a fixed oriented graph.
    An oriented graph $\vec G=(V,E)$ is an $\vec H$-gadget with parameter $k$ if $\vec G=\vec H_0\cup \vec K$, where $\vec K$ is a well-oriented $\vec K_{k,k}$ and $\vec H_0\simeq \vec H$, and the conditions below hold simultaneously.
    \begin{itemize}
        \item There exists $\vec e_0=(u,v)\in E(\vec G)$ such that $V(\vec H_0)\cap V(\vec K)=\{u,v\}$, $E(\vec H_0)\cap E(\vec K)=\{\vec e_0\}$.
        \item $\vec G-\vec e_0$ is $\vec H$-free.
    \end{itemize} 
\end{definition}

Suitable \(\vec H\)-gadgets yield the desired single-pass \(\Omega(n^2)\)-space lower bound via a reduction from \(\textsc{Index}\). The communication lower bound for \(\textsc{Index}\) can be found in \cref{prop:CCbound}.

\begin{lemma}\label{ReductionLemma}
    Let $\vec H$ be a fixed oriented graph. If for every sufficiently large integer $n$ there exists an $\vec H$-gadget with parameter $n$, then $R^1_{2/3}(\textsc{Sub}(\vec H))=\Omega(n^2).$
\end{lemma}
\begin{proof}
    Fix an \(\vec H\)-gadget
$\vec G=\vec H_0\cup\vec K$
with parameter \(n\). Let \(X,Y\) be the two parts of \(\vec K\), with all arcs directed from \(X\) to \(Y\), and let \(\vec e_0=(x,y)\) be the common arc of \(\vec H_0\) and \(\vec K\), where \(x\in X\) and \(y\in Y\). Fix an isomorphism \(\vec H_0\simeq \vec H\), and let \(\vec e\in E(\vec H)\) correspond to \(\vec e_0\).

For the reduction, take fresh vertex sets $U=\{u_1,\ldots,u_n\}, V=\{v_1,\ldots,v_n\}$.
Given an instance of \(\textsc{Index}\) on \([n]\times[n]\), Alice receives \(S\subseteq[n]\times[n]\) and builds $\vec A(S)$ on $U\sqcup V$ using the arcs $
\{(u_i,v_j):(i,j)\in S\}.
$
Bob receives \((s,t)\in[n]\times[n]\) and inserts a copy \(\vec B(s,t)\) of \(\vec H-\vec e\), identifying the tail and head of \(\vec e\) with \(u_s\) and \(v_t\), respectively.

If \((s,t)\in S\), then \((u_s,v_t)\) together with \(\vec B(s,t)\) forms a copy of \(\vec H\). Suppose instead that \((s,t)\notin S\). Choose bijections $
 U\to X, V\to Y
$
sending \(u_s\) to \(x\) and \(v_t\) to \(y\). Under these bijections, $\vec A(S)$ forms a subgraph of \(\vec K-\vec e_0\), while $\vec B(s,t)$ corresponds to \(\vec H_0-\vec e_0\). Hence $\vec A(S)\cup \vec B(s,t)$ is isomorphic to a subgraph of \(\vec G-\vec e_0\), and is therefore \(\vec H\)-free by \cref{ReductionOriented}.
Thus $\vec A(S)\cup \vec B(s,t)$ contains \(\vec H\) if and only if \((s,t)\in S\). 

A single-pass streaming algorithm using $b$ bits of space therefore gives a one-way protocol for $\textsc{Index}$ on ground set $[n]\times [n]$: Alice processes her arcs in $\vec A(S)$ and sends the memory state to Bob, who processes his arcs in $\vec B(s,t)$ and determines the answer. The communication lower bound on $\textsc{Index}$ then gives $b=\Omega(n^2)$. Since $\vec A(S)\cup \vec B(s,t)$ has $2n+|V(H)|-2=\Theta(n)$ vertices, this implies $R^1_{2/3}(\textsc{Sub}(\vec H))=\Omega(n^2).$
\end{proof}

By \cref{ReductionLemma}, it remains to construct \(\vec H\)-gadgets for the two hard cases, and the key is to choose the arc \(\vec e\in E(\vec H)\) corresponding to \(\vec e_0\). For a component containing at least two internal vertices, we choose \(\vec e\) on a shortest path between two distinct internal vertices.

\begin{proposition}\label{TwoInternalVertex}
    Let $\vec H$ be an NWO bipartite graph. If $\vec H$ has a component containing at least $2$ internal vertices, then for every sufficiently large integer $n$, there exists an $\vec H$-gadget with parameter $n$.
\end{proposition}
\begin{proof}
    Let $I(\vec H)$ be the set of internal vertices in $\vec H$ and 
    \begin{align*}
        d=\min_{x,y\in I(\vec H),x\ne y}d_H(x,y),\quad \mathcal P=\{\text{length-}d\text{ paths in }H\text{ with endpoints in }I(\vec H)\}.
    \end{align*}
    Choose $ P\in\mathcal P$ and $\vec e$ along $ P$. Consider a copy $\vec H_0$ of $\vec H$, where we fix an isomorphism $\varphi:\vec H_0\simeq \vec H$ with $\varphi(\vec e_0)=\vec e$.
    Build $\vec G=\vec H_0\cup \vec K$ as in \cref{ReductionOriented}, where $\vec K$ is a well-oriented $\vec K_{n,n}$ and $V(\vec H_0)\cap V(\vec K)=\{u,v\},E(\vec H_0)\cap E(\vec K)=\{\vec e_0\}$.  
    To show that $\vec G$ is an $\vec H$-gadget of parameter $n$, it suffices to show that $\vec G-\vec e_0$ is $\vec H$-free.
    \begin{figure}[H]
        \centering
        \includesvg[width=0.45\linewidth]{TwoInternal.svg}
        \caption{The construction of an $\vec H$-gadget with parameter $n$ in \cref{TwoInternalVertex}. }
    \end{figure}
     Since $\vec K$ is well-oriented, each internal vertex of $\vec G-\vec e_0$ lies in $V(\vec H_0)$, and corresponds to a vertex in $I(\vec H)$ by $\varphi$.
    Consider any length-$d$ path $ Q$ in $ G-e_0$ between internal vertices of $\vec G$. If $ Q$ uses some edge of $K$, then its segment $Q|_K$ in $K$ joins $u$ and $v$, because these are the only vertices shared by $K,H_0$ and the endpoints of $Q$ are in $H_0$. Since $(u,v)\notin \vec E(\vec G-\vec e_0)$,  $Q|_K$ has length greater than one. It follows that $\varphi(Q-Q|_K)+\vec e$ is a path in $ H$ with endpoints in $I(\vec H)$ of length less than $d$, contradicting the minimality of $d$. Thus $Q\subseteq  H_0$.

    Therefore, the set of such paths $Q$ maps injectively into the set of $e$-avoiding paths in $\mathcal P$, whose size is less than $|\mathcal P|$ since $e\in P$ and $P\in\mathcal P$.
    This shows that the number of length-$d$ paths in $G-e_0$ between internal vertices of $\vec G$ is less than $|\mathcal P|$.
    However, any copy of $\vec H$ in $\vec G-\vec e_0$ would supply $|\mathcal P|$ distinct such paths, because its internal vertices remain internal in $\vec G$, and this gives a contradiction. Thus $\vec H\nsubseteq \vec G-\vec e_0$.
\end{proof}
The next case is when \(\vec H\) has a component \(\vec F\) containing a cycle and exactly one internal vertex. Adjoining a fixed disjoint copy of \(\vec H-\vec F\) reduces \(\textsc{Sub}(\vec F)\) to \(\textsc{Sub}(\vec H)\). Thus it suffices to consider the case where \(\vec H\) is connected.

Let \(w\) be the unique internal vertex in $\vec H$, and let \(m\) be the minimum distance from \(w\) to a shortest cycle in $H$. We focus on the shortest cycles containing the maximum number of vertices at distance exactly \(m\) from \(w\), and choose \(\vec e\) on one of them. 
\begin{proposition}\label{pure}
    Let $\vec H$ be an NWO bipartite graph. If $\vec H$ has a component containing a cycle and exactly one internal vertex, then for every sufficiently large integer $n$, there exists an $\vec H$-gadget with parameter $n$.
\end{proposition}

\begin{proof}
    By the preceding reduction, assume that \(\vec H\) is connected. Let \(w'\in V(\vec H)\) be its unique internal vertex, and let \(\ell\) be the length of a shortest cycle in \(H\). Define
\[
m=\min\{d_H(w',z): z\in V(C),\ C\text{ is an }\ell\text{-cycle in }H\}.
\]
For every \(\ell\)-cycle \(C\) in \(H\), let
$
A(C)=\{z\in V(C):d_H(w',z)=m\}
$. Define
\begin{align*}
    M&=\max\{|A(C)|: C\text{ is an }\ell\text{-cycle in }H\},\\
    \mathcal C&=\{C:C\text{ is an }\ell\text{-cycle in }H,\ |A(C)|=M\}.
\end{align*}
Note that \(M\ge1\), and hence every \(C\in\mathcal C\) contains a vertex at distance \(m\) from \(w'\).
    Choose $C^*\in\mathcal C$, and choose an arc $\vec e=(u',v')$ on $\vec C^*$ such that 
    $
        w'\notin\{u',v'\},A(C^*)\nsubseteq\{u',v'\}
    $, which is possible since $\ell\geq 4$.
    
    Consider a copy $\vec H_0$ of $\vec H$, where we fix an isomorphism $\varphi:\vec H_0\to \vec H$ with $\varphi(\vec e_0)=\vec e$. Build $\vec G=\vec H_0\cup \vec K$ as in \cref{ReductionOriented}, where $\vec K$ is a well-oriented $\vec K_{n,n}$ and $V(\vec H_0)\cap V(\vec K)=\{u,v\}$, $E(\vec H_0)\cap E(\vec K)=\{\vec e_0\}$.
To show that $\vec G$ is an $\vec H$-gadget with parameter $n$, it suffices to show that $\vec G-\vec e_0$ is $\vec H$-free.
\begin{figure}[H]
        \centering
        \includesvg[width=0.55\linewidth]{purecycle.svg}
        \caption{The construction of an $\vec H$-gadget with parameter $n$ in~\cref{pure}. }
\end{figure}
Let $w'=\varphi(w)$, then $w$ is the unique internal vertex of $\vec G-\vec e_0$. Observe the following.
\begin{enumerate}
    \item Every path $Q$ in $G-e_0$ of length at most $m$ starting from $w$ lies in $H_0$, because before using any edge in $K$ 
it must first reach $u$ or $v$ through $H_0$, and that segment already requires at least $m$ edges since its image under $\varphi$ goes from $w'$ to $\{u',v'\}\subseteq V(C^*)$ in $H$.
\item No cycle in $G-e_0$ of length $\ell$  can use edges from both $K$ and $H_0$. Otherwise, such a cycle can be divided into one $u$,$v$-path $Q_1$ in $K$ and one $u$,$v$-path $Q_2$ in $H_0$. The path $Q_1$ has length at least three because $(u,v)=\vec e_0$ is absent, but this means that \(Q_2+e_0\) is a cycle of length less than \(\ell\) in \(H_0\), and hence $\varphi(Q_2+e_0)$ is a cycle in $H$ shorter than $\ell$, contradicting the minimality of $\ell$.
\end{enumerate}

Now suppose that $\psi:\vec H\hookrightarrow\vec G-\vec e_0$ embeds a copy of $\vec H$ into $\vec G-\vec e_0$. The unique internal vertex remains internal under $\psi$, so $\psi(w')=w$. Consider any cycle $C\in\mathcal C$ in $H$. By definition, for each $z'\in A(C)$ there is a $w',z'$-path $Q$ of length $m$ in $H$, and  $\psi(Q)$ lies in $H_0$ by our first observation, so $\psi(z')\in V(H_0)$.

Therefore, if $\psi(C)\subseteq K$, then  $\psi(A(C))\subseteq V(K)\cap V(H_0)=\{u,v\}$, and hence $\varphi$ maps $\psi(A(C))$ bijectively to a set $S\subseteq\{u',v'\}$, where each vertex in $S$ can be reached from $w'$ by a length-$m$ path in $H$. By definition we get $S\subseteq A(C^*)$, but our choice of $\vec e=(u',v')$ also gives $A(C^*)\nsubseteq \{u',v'\}$, and hence
$
    M=|\psi(A(C))|=|S|\leq |A(C^*)\cap \{u',v'\}|<|A(C^*)|=M,
$
which is a contradiction. Thus $\psi(C)\nsubseteq K$, and together with our second observation we get $\psi(C)\subseteq H_0-e_0$ for every $C\in \mathcal C$.

Since $\psi(C)\subseteq H_0-e_0$, we know $\varphi$ maps $\psi(C)$ to an $\ell$-cycle $D$ in $H-e$. Under $\varphi$, each length-$m$ path between $w$ and $\psi(A(C))$ in $H_0-e_0$ is mapped to a length-$m$ path between $w'$ and $A(D)$ in $H-e$. This shows that $|A(D)|\geq |\psi(A(C))|=M$, and hence $|A(D)|=M$ by the maximality of $M$. Thus $D\in\mathcal C$, and $D\ne C^*$ because $e\in C^*$ and $e\notin D$. 
Therefore, we get the injections
$
    \mathcal C\overset\psi\hookrightarrow \{\psi(C):C\in\mathcal C\}\overset\varphi\hookrightarrow \mathcal C\setminus \{C^*\},   
$
which is impossible. In conclusion, such an embedding $\psi$ cannot exist, and $\vec H\nsubseteq \vec G-\vec e_0$.
\end{proof}
Combining Propositions \ref{TwoInternalVertex} and \ref{pure} with \cref{ReductionLemma}, we  get the lower bounds in \cref{thmOSPC}.
\begin{proposition}
    Let $\vec H$ be a fixed NWO bipartite graph. If $\vec H$ has a component containing two internal vertices, or a component containing a cycle and exactly one internal vertex, then $R^1_{2/3}(\textsc{Sub}(\vec H))=\Omega(n^2)$.
\end{proposition}

\subsection{Tractable NWO Bipartite Patterns in the Single-Pass Model}\label{sec:oriented-upper-bounds}

The remaining patterns have components that are either WO or trees with exactly one internal vertex. The WO case follows from the Turán bound, so we first handle trees with a unique internal vertex and then extend to disconnected patterns.

\begin{proposition}\label{novel}
Let $\vec H$ be a fixed oriented tree with at most one internal vertex. Then $D^1\bigl(\textsc{Sub}(\vec H)\bigr)=\widetilde O(n)$.
\end{proposition}

If $\vec H$ has no internal vertex, it is a WO tree and its linear extremal bound (\cref{sec:turan}) suffices. Otherwise, let $x$ be its unique internal vertex. In this case, $\vec H$ can be divided into two WO subtrees $\vec T_1,\vec T_2$ with $V(\vec T_1)\cap V(\vec T_2)=\{x\}$, where all in-neighbors of $x$ lie in $\vec T_1$ and all out-neighbors of $x$ lie in $\vec T_2$. Namely, this decomposes $\vec H$ into two WO subtrees joined at $x$.

For simplicity, we first handle the case where the host graph $\vec G$ can be partitioned into three layers $L_1,L_2,L_3$, with all arcs directed from $L_1$ to $L_2$ or from $L_2$ to $L_3$. We call this a \textit{three-layer graph}. Furthermore, we assume that the layer labels of all vertices are known before the algorithm begins. The strategy of the subgraph-finding algorithm is to find a linear-sized sparse subgraph $\vec{G}'$ of $\vec{G}$ that preserves a copy of $\vec{H}$, if one exists, so that we only need to search for $\vec{H}$ in this sparse certificate.

\paragraph*{A sparse certificate for rooted trees.}

To decide whether a three-layer graph $\vec{G}$ contains a copy of $\vec H$, it suffices to check, for each $v\in L_2$, whether $\vec G[L_1\cup L_2]$ and $\vec G[L_2\cup L_3]$ contain copies of $\vec T_1$ and $\vec T_2$, respectively, both mapping the root $x$ to $v$ and sharing no other vertex in $L_2$. We can first focus on the subproblem of embedding a WO rooted subtree $\vec{T_i}$ into a WO host graph $\vec{G}[L_i\cup L_{i+1}]$. Since both graphs are WO and their orientations match, we can treat this as an undirected version of the rooted tree embedding problem.
Let $G$ denote the underlying graph of $\vec G[L_i\cup L_{i+1}]$.

The sparsification idea is to utilize the $k$-core $C = K_k(G)$ in $G$ for some properly chosen $k$, as defined in Appendix~\ref{AppendixA}. This vertex set has a clean property that few edges have an endpoint outside $C$, while the subgraph induced by $C$ has a sufficiently large minimum degree. We then sparsify only the edges inside $C$, while preserving a sufficiently large minimum degree there, so that we retain all edges with an endpoint outside $C$ and can greedily grow any bounded-size tree in $G'[C]$. The following algorithm constructs this sparsification.

\begin{samepage}
\begin{lemma}\label{lem:kcore-certificate}
Let $G$ be an $n$-vertex undirected graph, let $k\geq 1$ be fixed, and let $C=K_k(G)$ be the vertex set of its $k$-core. A deterministic one-pass algorithm stores a subgraph $G'\subseteq G$ using $O(kn\log n)$ bits such that:
\begin{enumerate}
    \item Every edge of $G$ with an endpoint outside $C$ belongs to $G'$;
    \item Every vertex of $C$ has at least $k$ neighbors in $G'[C]$;
    \item $|E(G')|=O(kn)$.
\end{enumerate}
\end{lemma}
\end{samepage}

\begin{proof}
Let $G_t$ be the graph formed by the first $t$ input edges, and let $G'_t$ be the graph retained after processing them, with its $k$-core vertex set $C'_t=K_k(G'_t)$. We initialize $G'_0$ as an edgeless graph with vertex set $[n]$. When the $t$-th edge $uv$ arrives, we keep $G'_t$ as $G'_{t-1}$ if both $u,v\in C'_{t-1}$; otherwise, we add $uv$ to $G'_{t-1}$ to obtain $G'_t$ and update $C'_t$ accordingly. Finally, $G_m=G$, where $m=|E(G)|$, and we output $G'=G'_m$.

We first show that $C'_t=K_k(G_t)$ for every $t$. The equality holds for $t=0$. Assume that it holds at $t-1$. If $u,v\in C'_{t-1}$, the core vertex sets of both $G_t$ and $G'_t$ remain unchanged, so the equality is unaffected at this step. Otherwise, the algorithm retains $uv$ in $G'_t$. Let $L=K_k(G_t)$, so every vertex in $L$ has at least $k$ neighbors in $G_t[L]$. Vertices in the old core $C'_{t-1}$ already have degree at least $k$ there, while all previous edges incident to $L\setminus C'_{t-1}$ are preserved in $G'_t$, since all previously discarded edges lie inside $C'_{t-1}$. Thus $G'_t[L]$ has minimum degree at least $k$, giving $L\subseteq K_k(G'_t)$, while $K_k(G'_t)\subseteq K_k(G_t)=L$. Thus the equality also holds at time $t$.

For the final graphs $G$ and $G'$, property~\textup{(1)} follows because every discarded edge lies inside $C=K_k(G)$, and property~\textup{(2)} follows from $C=K_k(G')$. To prove property~\textup{(3)}, define
\[
X_t=C'_t\setminus C'_{t-1}\quad (1\leq t\leq m),
\qquad X_{m+1}=V(G)\setminus C'_m.
\]
These sets form a partition of $V(G)$. We also partition $E(G')$ into $E_1,\ldots,E_{m+1}$, where $E_t$ consists of edges within $X_t$ or joining $X_t$ to $X_i$ with $i<t$. 

For $t\leq m$, no edge is added to $E_t$ after time $t$, since both endpoints are then in the core. At time $t-1$, all vertices of $X_t$ lie outside the core. Applying the peeling process described in Appendix~\ref{AppendixA} to $G_{t-1}'$, every edge of $E_t$ present at time $t-1$ is removed when its endpoint in $X_t$ is peeled. Each vertex of $X_t$ has degree less than $k$ when it is peeled, so at most $(k-1)|X_t|$ such edges are present. Thus $|E_t|\le (k-1)|X_t|+1.$ In fact, $|E_t|\le k|X_t|$, since $|X_t|=0$ implies $|E_t|=0,$ and $(k-1)|X_t|+1\le k|X_t|$ when $X_t$ is nonempty.

We can derive $|E_{m+1}|\le k|X_{m+1}|$ by a similar argument on $G_m'$.
Summing up, we obtain
\[
|E(G')|=\sum_{t=1}^{m+1}|E_t|
\leq\sum_{t=1}^{m+1}k|X_t|=kn.
\]
The argument applies at any time during the stream, so the algorithm uses $O(kn\log n)$ bits in one pass.
\end{proof}

Lemma~\ref{lem:kcore-certificate} also determines the $k$-core of $G$ using near-linear space for a constant $k$. In contrast, when $k$ is part of the input, Assadi et al.~\cite{AssadiGLMM24} proved a $\widetilde\Omega(n^2/p^3)$ space lower bound for $p$-pass streaming algorithms, even for deciding whether a nonempty $k$-core exists.

The next lemma shows that the certificate from \cref{lem:kcore-certificate} preserves rooted-tree embeddings at prescribed roots, even after forbidding a vertex set $F$ avoiding the root. This later prevents overlap between the two rooted trees.

\begin{lemma}\label{lem:rooted-tree-preservation}
Let $G'\subseteq G$ be a subgraph on the same vertex set satisfying properties~\textup{(1)} and~\textup{(2)} of $\cref{lem:kcore-certificate}$, with $C$ the vertex set of the $k$-core of $G$. Let $(T,r)$ be a rooted tree on $t$ vertices, let $F\subseteq V(G)$ be a set of forbidden vertices, and fix $v\in V(G)\setminus F$. If $k\geq t+|F|$, then $G-F$ contains a copy of $T$ mapping $r$ to $v$ if and only if $G'-F$ contains a copy of $T$ mapping $r$ to the same vertex $v$.
\end{lemma}

\begin{proof}
Since $G'\subseteq G$, the reverse implication is immediate. If $v\in C\setminus F$, property~\textup{(2)} ensures that every vertex of $C\setminus F$ has at least $k-|F|\geq t$ neighbors within $C\setminus F$ in $G'-F$, and hence also in $G-F$. We can therefore greedily grow a copy of the whole rooted tree $T$ at $v$ inside $C\setminus F$ in either graph.

Otherwise, suppose $v\in V(G)\setminus(C\cup F)$ and fix a copy $T^*$ of $T$ rooted at $v$ in $G-F$. Starting
from $v$, keep the original embedding along each branch until it first enters $C$. Every edge in this
partial embedding has an endpoint outside $C$, so it belongs to $G'-F$ by property~\textup{(1)}. After fixing these partial embeddings on all branches, we can again complete the remaining subtrees inside $G'[C\setminus F]$ by the greedy embedding, using the same minimum-degree bound while avoiding all vertices already used. This gives another copy of $T$ in $G'-F$ rooted at $v$.
\end{proof}

With these sparse certificates, we first present a deterministic single-pass streaming algorithm that finds a copy of $\vec H$, if one exists, when $\vec{G}$ is a three-layer graph, as long as the layering is known before the algorithm. 

\begin{lemma}\label{lem:three-layering-embedding}
Let $\vec H$ be a fixed oriented tree on $h$ vertices with exactly one internal vertex. Let $\vec G$ be an $n$-vertex three-layer oriented graph whose layer labels are known before the input stream begins. There is a deterministic one-pass algorithm using $O(hn\log n)$ bits that finds a copy of $\vec{H}$ inside $\vec G$ or correctly reports that no such copy exists.
\end{lemma}

\begin{proof}
Let $V(\vec{G})$ have three layers $L_1,L_2,L_3$, where all input arcs of $\vec{G}$  point from $L_i$ to $L_{i+1}.$ We treat the input stream of $\vec G$ as an interleaving of the streams for $\vec G_1$ and $\vec G_2$, where $\vec G_1=\vec G[L_1\cup L_2]$ and $\vec G_2=\vec G[L_2\cup L_3]$. Applying Lemma~\ref{lem:kcore-certificate} to their underlying undirected graphs with $k=2h$, we store sparse certificates $\vec G'_1,\vec G'_2$ with the original arc directions.

We claim that $\vec G$ contains $\vec H$ if and only if $\vec G'_1\cup\vec G'_2$ does. The reverse implication follows since both certificates are subgraphs of $\vec G$. For the forward implication, fix a copy $\vec{H}^* =\vec R_1\cup\vec R_2$ in $\vec G$ with common root $v\in L_2$ and $\vec{R_i}\cong \vec{T_i}$. It suffices to show the presence of rooted tree embeddings from $(\vec T_1,x)$ and $(\vec T_2,x)$ to $G'_1$ and $G'_2$, respectively, both mapping $x$ to $v$ and sharing no other vertex. Since the layers match the arc directions between WO patterns $\vec G_i$ and $\vec T_i$, we can treat them as undirected graph embeddings, and Lemma~\ref{lem:kcore-certificate} guarantees the sparse certificate structure of $\vec{G_1'},\vec G_2'$.

Set $F=(V(\vec R_2)\cap L_2)\setminus\{v\}\subseteq V(\vec G_1')$, $|F| \le h$. Lemma~\ref{lem:rooted-tree-preservation} shows that there exists a copy $\vec R'_1$ of $\vec T_1$ rooted at $v$ in $\vec G'_1-F$. The copy $\vec R_2$ remains disjoint from $\vec R'_1$ outside $v$, so applying the same lemma again with $F'=(V(\vec R'_1)\cap L_2)\setminus\{v\}$ shows that there exists a copy $\vec R'_2$ of $\vec T_2$ rooted at $v$ in $\vec G'_2-F'$. Thus $\vec R'_1$ and $\vec R'_2$ meet only at $v$ and together form a copy of $\vec H$.

It takes $O(hn\log n)$ bits in one pass to store the certificates. After the stream, search their union for $\vec H$. We either find a copy or assert that no copy exists.
\end{proof}

\paragraph*{Color-coding and derandomization.}

We apply the color-coding technique and a derandomization routine to extend the deterministic algorithm on three-layer graphs to general oriented graphs.

\begin{proof}[Proof of Proposition~\ref{novel}]

The WO case again follows from the linear extremal bound in \cref{sec:turan}. For oriented tree $\vec H$ with exactly one internal vertex, we first present a one-sided-error randomized algorithm that finds a copy with constant probability.

 Suppose $h=|V(\vec H)|$, and a copy $\vec{H}^*$ exists inside $\vec G$. We randomly color each input vertex independently and uniformly with one of three colors before the stream begins. This gives a coloring $\chi:V(G)\to\{1,2,3\}$. Let $L_i$ be the set of vertices of color $i$, and let $\vec G_\chi$ be the three-layer subgraph of $\vec G$ containing only arcs from color $1$ to $2$ or from color $2$ to $3$, we treat the input stream of $\vec{G}$ as an input stream of $\vec{G}_\chi$. We ignore all other arcs in the stream.

 With probability $3^{-h}$,  $\chi$ correctly assigns every vertex of $\vec H^*$ its prescribed layer color, so $\vec G_\chi$ contains $\vec H$. Running the algorithm from Lemma~\ref{lem:three-layering-embedding} on $\vec{G}_\chi$ then finds a copy. This yields a one-pass randomized algorithm using $O(nh\log n)$ bits that succeeds with probability at least $3^{-h}$ when a copy exists, and never returns a false positive.

The preceding algorithm can be derandomized with $\widetilde{O}(1)$ space overhead by applying Lemma~\ref{ColorCoding} from Appendix~\ref{AppendixA} with $q=h$. This gives a family of $O(\log^2 n)$ colorings $\mathcal{C}$ with $h$ colors, where any vertex set of size $h$, including $V(\vec H^*)$, is colorful under some coloring in $\mathcal{C}$. Composing each coloring $\chi\in\mathcal{C}$ with a map $f:[h]\to[3]$ gives a three-coloring of $\vec G$. Taking all such $\chi$ and $f$ yields a family of size $3^h\cdot O(\log^2 n)=O(\log^2 n)$ that realizes every three-color assignment on $V(\vec H^*)$. 

Running the three-layer algorithm in Lemma~\ref{lem:three-layering-embedding} on all these colorings in parallel over the same pass, we can deterministically find a copy of $\vec H$ using $O(nh\log^3n) = \widetilde{O}(n)$ space.
\end{proof}

\paragraph*{Extending to disconnected patterns.}
Proposition~\ref{novel} can be extended to disconnected patterns using similar color-coding and derandomization techniques.

\begin{corollary}\label{NWOFandWO}
Let $\vec H$ be a fixed NWO bipartite graph whose components are either WO or trees with exactly one internal vertex. Then
$
D^1(\textsc{Sub}(\vec H))=\widetilde O(n^{2-1/\Delta'(H)})
$; moreover, $D^1(\textsc{Sub}(\vec H))=\widetilde O(n)$ if $H$ is a forest.
\end{corollary}

\begin{proof}
    Write $\vec H=\bigsqcup_{i=1}^r \vec H_i$ for its connected components. Color each vertex of $\vec G$ independently from $[r]$, and in color class $i$ run Corollary~\ref{ONB} or Proposition~\ref{novel} for $\vec H_i$. Any fixed copy of $\vec H$ is colored consistently with probability $r^{-|V(H)|}$; then the routines find vertex-disjoint component copies whose union is a valid copy of $\vec H$.

Derandomize as in Proposition~\ref{novel}, obtaining $O(\log^2 n)$ colorings that include the required component-color assignment for any existing copy. Run all component routines for all colorings in parallel. Since $\Delta'(H_i)\leq \Delta'(H)$ for each component with an edge, each routine uses at most $\widetilde O(n^{2-1/\Delta'(H)})$ space, and the fixed number of components and polylogarithmically many colorings preserve this bound. If $H$ is a forest, every component is a tree with at most one internal vertex, so Proposition~\ref{novel} applies to each of them and every routine uses $\widetilde O(n)$ space.
\end{proof}

%% file: induced-main.tex
\section{Induced Subgraph Finding}\label{sec:induced}

We extend our space complexity dichotomies to the induced
subgraph setting, including both undirected and oriented patterns.
In general, the induced problem is markedly more difficult. A
very dense graph, such as a complete graph, may fail to contain a prescribed induced pattern.
The argument that sufficiently many edges guarantee a positive instance
does not apply here, and the Tur\'{a}n number no longer characterizes the
space complexity of the problem.

We obtain a dichotomy in the insertion-only streaming model for undirected patterns.
Our upper bounds are deterministic and support finding an induced copy when
one exists, while our lower bounds hold even for randomized algorithms that
only decide whether such a copy exists.

\begin{theorem}[Undirected induced dichotomy]\label{thm:undirected-induced-dichotomy}
Let $H$ be a fixed undirected graph with at least three vertices.
If $H\in\{P_3,\mathrm{co}\text{-}P_3,P_4\}$, then $\textsc{IndSub}(H)$ admits a deterministic algorithm using a polylogarithmic number of passes and $\widetilde O(n)$ space. For every other $H$,
\[
    R^p_{2/3}\bigl(\textsc{IndSub}(H)\bigr)=\Omega(n^2/p).
\]
\end{theorem}

\subsection{Hard Undirected Patterns}

We first show that non-bipartite patterns remain hard in the induced setting, as in Theorem~\ref{thmUNB}.

\begin{proposition}\label{prop:induced-nonbipartite}
If $H$ is non-bipartite, then $R^p_{2/3}(\textsc{IndSub}(H)) = \Omega(n^2/p)$.
\end{proposition}
\begin{proof}
In the reduction used for the non-bipartite case of Theorem~\ref{thmUNB}, every yes-instance contains a copy of $H$, and moreover, this copy is induced.
Indeed, let $x=(i-1)t+j$ be the common element, and consider the copy of $H$ formed by $s^{(i)}$, the $j$th copy of $A\setminus\{s\}$, and the copy of $H-A$ (which is empty in the special case).
Every fixed edge with both endpoints in this vertex set is an edge of this copy, while every input edge outside $E_x$ joins some $s^{(i')}$ to the $j'$th copy of $A\setminus\{s\}$ with $(i',j')\neq(i,j)$, and hence has an endpoint outside this vertex set. On the other hand, every no-instance contains no copy of $H$ at all (not even induced). Therefore, an algorithm for $\textsc{IndSub}(H)$ also solves the communication game used in that reduction, yielding the claimed lower bound.
\end{proof}
In fact, via graph complementation, the hardness result for non-bipartite patterns extends to bipartite
patterns whose complements are non-bipartite. The key observation is that each party can generate
the missing edges in its private edge family from its own input. This allows the parties to construct
the complement of the reduction graph using insertions only, leading to a hardness reduction.
\begin{proposition}\label{prop:induced-complement-nonbipartite}
If $\overline{H}$ is non-bipartite, then
\[
R^p_{2/3}(\textsc{IndSub}(H)) = \Omega(n^2/p).
\]
\end{proposition}

\begin{proof}
Given a graph $G$, observe that
\[
G \text{ contains an induced copy of } \overline{H}
\quad\Longleftrightarrow\quad
\overline{G} \text{ contains an induced copy of } H .
\]
Apply the multiparty reduction from \textsc{Multi-Disjointness} in Lemma~\ref{lem:subgraph-space}, with the gadget used in Theorem~\ref{thmUNB}, to the non-bipartite pattern $\overline H$.
By the proof of Proposition~\ref{prop:induced-nonbipartite}, it constructs a graph $G_{\overline H}$ such that
\[
G_{\overline H}\text{ contains an induced copy of }\overline H
\quad\Longleftrightarrow\quad
\bigcap_{q=1}^r S_q\neq\emptyset.
\]
By the preceding equivalence, $\overline{G_{\overline H}}$ contains an induced copy of $H$ under exactly the same condition.
It therefore remains only to show that the parties can generate $\overline{G_{\overline H}}$ in the insertion-only model.

In that reduction, the gadget has edge sets $E_i=\{e_{i,1},\ldots,e_{i,r}\}$ for $i\in[k]$ and fixed edge set $F$, and party $q$ contributes the edges $e_{i,q}$ for $i\in S_q$.
Write $\mathcal{E}_q=\{e_{i,q}:i\in[k]\}$ for the edges that party $q$ may contribute.
Since the sets $E_i$ are pairwise disjoint, each of size $r$, and disjoint from $F$, the families $\mathcal{E}_1,\ldots,\mathcal{E}_r$ are pairwise disjoint and disjoint from $F$.
Moreover,
\[
E(G_{\overline H})=F\cup\bigcup_{q\in[r]}\{e_{i,q}:i\in S_q\}.
\]
Consequently,
\[
E(\overline{G_{\overline H}})
=
\left(\binom{V(G_{\overline H})}{2}\setminus
  \left(F\cup\bigcup_{q\in[r]}\mathcal{E}_q\right)\right)
\cup
\bigcup_{q\in[r]}\{e_{i,q}:i\in[k]\setminus S_q\}.
\]
The first set is fixed and can be streamed by a designated party, while each party $q$ can generate $\{e_{i,q}:i\in[k]\setminus S_q\}$ from its own input $S_q$.
Thus the parties generate $\overline{G_{\overline H}}$ using insertions only.
Running an induced-$H$ detection algorithm on this stream would therefore solve \textsc{Multi-Disjointness}.
Since the number $r$ of parties depends only on the fixed graph $H$, a $p$-pass streaming algorithm using $b$ bits of space gives an $O(prb)=O(pb)$-communication protocol.
Since $k=\Omega(n^2)$, the $\Omega(k)$ communication lower bound therefore implies $b=\Omega(n^2/p)$.
\end{proof}

These two propositions alone already establish hardness for most pattern graphs.

\begin{corollary}\label{cor:induced-five-vertices}
If $H$ has at least $5$ vertices, then
\[
R^p_{2/3}(\textsc{IndSub}(H)) = \Omega(n^2/p).
\]
\end{corollary}

\begin{proof}
Suppose that $H$ has $t$ vertices and that both $H$ and its complement $\overline{H}$ are bipartite.
Let $\phi:V(H)\to\{0,1\}$ and $\psi:V(H)\to\{0,1\}$ be proper $2$-colorings of $H$ and $\overline{H}$, respectively. 
Then the map $v\mapsto (\phi(v),\psi(v))\in\{0,1\}^2$ is a proper $4$-coloring of $K_t$. 
Thus
\[
t=\chi(K_t)\le 4.
\]
Consequently, every graph on at least $5$ vertices is either non-bipartite or has a non-bipartite complement. Propositions~\ref{prop:induced-nonbipartite} and~\ref{prop:induced-complement-nonbipartite} now give the claimed lower bound.
\end{proof}

For graphs with $|V(H)|\in\{3,4\}$, the only cases for which both $H$ and $\overline{H}$ are bipartite are
$P_3$, $\mathrm{co}\text{-}P_3$, $C_4$, $2K_2$, and $P_4$.
The patterns $C_4$ and $2K_2$ are also hard, as shown by a separate reduction.
It suffices to prove hardness for $\textsc{IndSub}(C_4)$; since $2K_2=\overline{C_4}$,
the same complementation argument as in Proposition~\ref{prop:induced-complement-nonbipartite}
yields the lower bound for $2K_2$.

\begin{proposition}\label{prop:induced-c4}
    $R^p_{2/3}(\textsc{IndSub}(C_4)) = \Omega(n^2/p)$.
\end{proposition}

The reduction uses the graph construction of Le Gall and Miyamoto~\cite[Section~3]{LeGallM21}
for induced-$C_4$ detection in CONGEST (see also~\cite{NikabadiK21}).
Appendix~\ref{sec:induced-undirected-proofs} gives the streaming proof and shows how
the same construction strengthens the SM-BCAST lower bound of Kari et al.~\cite{KariMRS15}
to $\Omega(n)$ bits per vertex for $C_4$ and, by complementation, for $2K_2$.

\begin{remark}
\Cref{prop:induced-nonbipartite,prop:induced-complement-nonbipartite} imply an $\Omega(n^2/p)$ space lower bound for $p$-pass randomized
recognition of claw-free and diamond-free graphs, two classes that have been extensively
studied for their structure and algorithmic applications~\cite{ChudnovskyS08,FaenzaOS11,LinSS12}.
A graph is diamond-free exactly when every vertex neighborhood induces a cluster graph.
Yet cluster graphs themselves admit deterministic one-pass recognition in $\widetilde O(n)$ space,
as shown in the next subsection (\cref{prop:induced-p3-upper}). This illustrates that recognizing a cluster graph and checking
that every neighborhood induces one may have sharply different streaming space complexities.
\end{remark}

\subsection{Tractable Undirected Patterns}
The three remaining graphs from the previous subsection can all be handled deterministically using $\widetilde O(1)$ passes and $\widetilde O(n)$ space.
\begin{proposition}\label{prop:induced-p3-upper}
$D^2(\textsc{IndSub}(P_3)) = \widetilde O(n)$.
\end{proposition}

\begin{proof}
A graph contains no induced $P_3$ if and only if it is a disjoint union of cliques.

In the first pass, retain a spanning forest $F$ and compute the degree of every vertex using $\widetilde O(n)$ space.
This determines the connected components of $G$ and their sizes.
If every vertex has degree equal to the size of its component minus one, then every component is a clique and $G$ contains no induced $P_3$.

Otherwise, choose a component $C$ and a vertex $a\in C$ with $\deg(a)<|C|-1$.
In a second pass, store the set $N(a)$.
The set $C\setminus(N(a)\cup\{a\})$ is nonempty; choose any vertex $b$ in this set.
Let
\[
a=x_0,x_1,\ldots,x_t=b
\]
be the unique $a$-to-$b$ path in $F[C]$.
Since $x_1\in N(a)$ and $x_t\notin N(a)$, there is an index $i\in[t-1]$ such that $x_i\in N(a)$ but $x_{i+1}\notin N(a)$.
The vertices $a,x_i,x_{i+1}$ therefore induce a $P_3$.
\end{proof}

The first pass of the preceding algorithm already recognizes cluster graphs deterministically;
the second pass is needed only to recover an induced $P_3$.
For recognition, Lokshtanov et al.~\cite[Lemma~5.5]{LMP24} obtained the same one-pass space bound
with randomization in the insertion-only model, using the same characterization in terms of
a spanning forest and vertex degrees. Their approach also allows the algorithm to handle
turnstile streams, which permit both edge insertions and deletions.

The $\mathrm{co}\text{-}P_3$-free graphs are exactly the complete multipartite graphs, or equivalently,
the complements of cluster graphs. It is therefore natural to try to apply the preceding algorithm
to $\overline{G}$. However, obtaining the required edge information about $\overline{G}$ is more
difficult in an insertion-only stream. To deal with the complement, we give a routine that either
returns its connected-component partition or finds an induced co-path in $G$.
This alternative is useful for the next two cases: setting $k=1$ in the following lemma gives
$\mathrm{co}\text{-}P_3$, while setting $k=2$ gives $\mathrm{co}\text{-}P_4\cong P_4$.

\begin{lemma}
\label{lem:components-or-copath}
Given a graph $G$ and an integer $k\ge 1$, there is an $O(k\log n)$-pass deterministic streaming algorithm using $\widetilde O(kn)$ space that outputs the partition of $V(G)$ into the connected components of $\overline{G}$ or, if it fails to do so, outputs an induced copy of $\mathrm{co}\text{-}P_{k+2}$ in $G$.
\end{lemma}
\begin{proof}
To partition the vertices into connected components of $\overline G$, the idea is to associate each vertex
with the ``smallest'' vertex it can reach within a fixed distance $k$ in $\overline G$.
Two vertices with the same destination must lie in the same component.
If the converse fails, we can recover an induced $P_{k+2}$ in $\overline G$.

To make ``smallest'' precise, we fix a strict linear order $\prec$ on $V(G)$. For each vertex $v$ and $0\le i\le k$, define
\[
N(v,i) := \min_{\prec}\{u\in V(G)\mid \mathrm{dist}_{\overline G}(u,v)\le i\},
\]
denoting the smallest vertex within distance $i$ of $v$ in $\overline G$. In particular, we initialize $N(v,0)=v$.

Our goal is then to compute $N(v,k)$ for each $v$, which we do by dynamic programming.
Assuming that all values $N(w,i)$ are known, the value $N(v,i+1)$ is the minimum of
the values $N(w,i)$ over all nonneighbors $w$ of $v$ in $G$ and $N(v,i)$ itself.
Order the vertices as $x_{1,i}\prec_i\cdots\prec_i x_{n,i}$ by their values $N(w,i)$,
breaking ties by $\prec$.
To compute $N(v,i+1)$ for each $v$, it suffices to find the smallest index $t_v$ such that $x_{t_v,i}\not\sim v$ in $G$; that is,
\[
t_v=\min\{t\in[n]:v\not\sim x_{t,i}\},
\qquad N(v,i+1)=N(x_{t_v,i},i).
\]

For each $i$, all values $N(v,i+1)$ can be found simultaneously by binary search.
Initialize a search interval $[1,n]$ for each vertex $v$.
In each pass, split each non-singleton interval into left and right halves and count the neighbors of $v$
among the vertices indexed by the left half.
If all those vertices are neighbors of $v$, keep the right half; otherwise, keep the left half.
Continue until each interval contains only $t_v$; such an index always exists since $v\not\sim v$.
We also record the vertex $x_{t_v,i}$ from which each value is inherited,
allowing recovery of a path from $v$ to $N(v,i)$ in $\overline G$ of length at most $i$.

Next, we check whether the partition defined by $N(v,k)$ is exactly the connected-component partition of $\overline G$.
Let $B_1,\ldots,B_\ell$ denote its classes.
We perform this check by counting the number $F$ of edges of $G$ joining distinct classes.

If $F=\sum_{i<j}|B_i||B_j|$, no edge of $\overline G$ joins distinct classes.
Since vertices in each class are connected through their common label, the classes are exactly
the connected components of $\overline G$.

Otherwise, there must be an edge $\{u,v\}\in E(\overline G)$ joining distinct classes.
We find one by counting each vertex's neighbors outside its class to identify a vertex $u$ with a nonneighbor $v$
outside its class, then locating $v$ by storing $N_G(u)$ in a further pass.

Assume without loss of generality that $w=N(u,k)\prec N(v,k)$.
Since a $w$-to-$v$ path of length at most $k$ would imply $N(v,k)\preceq w$, we have
\[
k+1\le \mathrm{dist}_{\overline G}(w,v)
\le \mathrm{dist}_{\overline G}(w,u)+1\le k+1.
\]
Thus $\mathrm{dist}_{\overline G}(w,u)=k$ and $\mathrm{dist}_{\overline G}(w,v)=k+1$.
The predecessor records recover a $w$-to-$u$ path of length $k$;
appending $\{u,v\}$ gives a shortest, hence induced, path of length $k+1$ in $\overline G$.
Its vertices induce a copy of $\mathrm{co}\text{-}P_{k+2}$ in $G$.

Each of the $k$ dynamic-programming iterations uses $O(\log n)$ passes and $\widetilde O(n)$ working space.
The partition check and the search for $\{u,v\}$ require only $O(1)$ further passes and $\widetilde O(n)$ space,
while storing all values and predecessor records takes $\widetilde O(kn)$ space.
The total complexity is therefore $O(k\log n)$ passes and $\widetilde O(kn)$ space.
\end{proof}

The routine with $k=1$ immediately shows tractability for $\textsc{IndSub}(\mathrm{co}\text{-}P_3)$.
\begin{corollary}\label{co-p3}
There exists $p=O(\log n)$ such that
\[
D^{p}(\textsc{IndSub}(\mathrm{co}\text{-}P_3)) = \widetilde O(n).
\]
\end{corollary}

\begin{proof}
A graph $G$ contains no induced $\mathrm{co}\text{-}P_3$ if and only if $\overline G$ is a disjoint union of cliques.
By Lemma~\ref{lem:components-or-copath}, after $O(\log n)$ passes we either obtain an induced $\mathrm{co}\text{-}P_3$ or obtain the partition of $V(G)$ into the connected components of $\overline G$.

In the latter case, take one additional pass to look for an edge of $G$ within a component of $\overline G$.
If none exists, all components of $\overline G$ are cliques, and we are done.
Otherwise, let $\{a,b\}$ be such an edge.
Their common label $z=N(a,1)=N(b,1)$ is distinct from and nonadjacent to both endpoints in $G$,
so $\{a,b,z\}$ induces a $\mathrm{co}\text{-}P_3$.
\end{proof}

The final tractable pattern is $P_4$. The $P_4$-free graphs are precisely the \textit{cographs},
whose recursive structure is described by disjoint unions and complete joins.
Nikolopoulos and Palios~\cite{cographs} used this structure to obtain an efficient parallel
algorithm that recognizes cographs and constructs their cotrees, returning an induced $P_4$
when the input is not a cograph. We follow their structural framework and adapt their parallel
algorithm to the streaming model, using Lemma~\ref{lem:components-or-copath} to handle
complement-component computations while carefully controlling the depth of the recursion.

\begin{theorem}\label{thm:P4_hardness}
There exists $p=O(\log^2 n)$ such that
\[
D^{p}(\textsc{IndSub}(P_4)) = \widetilde O(n).
\]
\end{theorem}

The algorithm relies on a recursive graph decomposition: on an $m$-vertex graph, it either finds
an induced $P_4$ or partitions the graph into several blocks, such that every induced $P_4$ must lie
within a single block and each block has at most $3m/4$ vertices.
Processing the disjoint blocks simultaneously at each level yields an $O(\log^2 n)$-pass
algorithm using $\widetilde O(n)$ space that also returns an induced $P_4$ whenever one exists.

\begin{samepage}
\begin{proposition}\label{prop:balanced-cograph-phase}
There is an $O(\log m)$-pass deterministic streaming routine using $\widetilde O(m)$ space that, on an $m$-vertex graph $G$ with $m\ge4$, either returns an induced copy of $P_4$ or returns a partition $\mathcal Q$ of $V(G)$ such that
\begin{enumerate}
    \item every induced $P_4$ in $G$ is contained in one block $Q\in\mathcal Q$; and
    \item every block $Q\in\mathcal Q$ satisfies $|Q|\le3m/4$.
\end{enumerate}
\end{proposition}
\end{samepage}

We defer the proof of Proposition~\ref{prop:balanced-cograph-phase} to
Appendix~\ref{sec:induced-p4-proof}. There, we first review the structural properties
of cographs established by Nikolopoulos and Palios, and then prove the proposition by
using them to implement the balanced decomposition and recover an induced $P_4$ in the
streaming model.
Algorithm~\ref{alg:cograph} in Appendix~\ref{sec:induced-p4-proof} summarizes the resulting
procedure, including the recursion used in the following proof.

\begin{proof}[Proof of Theorem~\ref{thm:P4_hardness}]
Starting with the single block $V(G)$, apply Proposition~\ref{prop:balanced-cograph-phase} to $G[Q]$ for every current block $Q$ with at least four vertices, in parallel over the same stream scans.
If any subroutine returns an induced $P_4$, return that copy and stop.
Otherwise, replace each processed block by its returned partition, and continue until all blocks have at most three vertices.

If no copy is found along the way, every decomposition keeps each induced $P_4$ within a single block.
In each round of $O(\log n)$ passes, every processed block is decomposed into blocks of at most three quarters of its size.
Thus, within $O(\log^2 n)$ passes, either an induced $P_4$ is returned or all blocks have at most three vertices, ruling out any induced $P_4$ in $G$.
Since the current blocks are always pairwise disjoint, we can process them independently over the same passes, using separate working memory for each block and $\widetilde O(n)$ space in total by Proposition~\ref{prop:balanced-cograph-phase}.
This yields a deterministic algorithm that finds an induced $P_4$ or correctly reports that none exists in $O(\log^2 n)$ passes and $\widetilde O(n)$ space.
\end{proof}

We can now complete the proof of Theorem~\ref{thm:undirected-induced-dichotomy}.
For the lower bounds, Corollary~\ref{cor:induced-five-vertices} handles every pattern with at least five vertices.
Among patterns with three or four vertices, Propositions~\ref{prop:induced-nonbipartite} and~\ref{prop:induced-complement-nonbipartite} handle all patterns except $P_3$, $\mathrm{co}\text{-}P_3$, $C_4$, $2K_2$, and $P_4$, while Proposition~\ref{prop:induced-c4} handles $C_4$ and the complementation argument preceding it handles $2K_2$.
For the upper bounds, Proposition~\ref{prop:induced-p3-upper}, Corollary~\ref{co-p3}, and Theorem~\ref{thm:P4_hardness} handle $P_3$, $\mathrm{co}\text{-}P_3$, and $P_4$, respectively.

\subsection{Oriented Patterns}\label{sec:section6}

We also extend the induced-subgraph classification to the oriented setting.
Among patterns on at least three vertices, only the oriented $\mathrm{co}\text{-}P_3$ remains tractable.

\begin{theorem}[Oriented induced dichotomy]\label{thm:oriented-induced-dichotomy}
Let $\vec H$ be a fixed oriented graph on at least three vertices, and let $H$ be its underlying undirected graph.
If $H\cong\mathrm{co}\text{-}P_3$, then $\textsc{IndSub}(\vec H)$ admits a deterministic $O(\log n)$-pass algorithm using $\widetilde O(n)$ space. For every other $H$,
\[
    R^p_{2/3}\bigl(\textsc{IndSub}(\vec H)\bigr)=\Omega(n^2/p).
\]
\end{theorem}

The random-orientation argument in the proof of Proposition~\ref{prop:oriented-to-undirected} also applies to induced patterns, immediately establishing hardness for every underlying pattern other than $P_3$, $\mathrm{co}\text{-}P_3$, and $P_4$.

\begin{corollary}
Let $\vec H$ be an oriented graph on at least three vertices, and let $H$ be its underlying undirected graph. If $H$ is isomorphic to none of $P_3$, $\mathrm{co}\text{-}P_3$, and $P_4$, then
\[
    R^p_{2/3}\bigl(\textsc{IndSub}(\vec H)\bigr)=\Omega(n^2/p).
\]
\end{corollary}

It remains to consider orientations of the three exceptional undirected patterns.
Because $\mathrm{co}\text{-}P_3$ consists of a single edge and an isolated vertex, it has a unique orientation up to isomorphism.
Thus an oriented graph $\vec G$ contains an induced oriented $\mathrm{co}\text{-}P_3$ if and only if its underlying graph $G$ contains an induced $\mathrm{co}\text{-}P_3$.
Corollary~\ref{co-p3} therefore gives an efficient algorithm for this orientation.
The next two propositions show that every orientation of $P_3$ or $P_4$ instead satisfies the quadratic lower bound.

\begin{proposition}\label{prop:oriented-P3-hard}
For every orientation $\vec H$ of $P_3$,
\[
    R^p_{2/3}\bigl(\textsc{IndSub}(\vec H)\bigr)=\Omega(n^2/p).
\]
\end{proposition}

\begin{proposition}\label{prop:oriented-P4-hard}
For every orientation $\vec H$ of $P_4$,
\[
    R^p_{2/3}\bigl(\textsc{IndSub}(\vec H)\bigr)=\Omega(n^2/p).
\]
\end{proposition}

Appendix~\ref{sec:induced-oriented-proofs} gives two three-layer reductions from
\textsc{Set-Disjointness} for the orientations of $P_3$ up to reversal, one for the
directed path and one for the in-star. In the first, fixed chords rule out induced
paths whose endpoints have different indices. In the second, fixed arcs force
the two sources to be a matched nonadjacent pair. The $P_4$ proof extends these
constructions by attaching matching leaves and checks the three orientations up
to isomorphism and reversal.

Together with Corollary~\ref{co-p3}, these lower bounds complete the oriented classification.
The oriented $\mathrm{co}\text{-}P_3$ is the only tractable induced pattern on at least three vertices.

%% file: openproblem.tex
\section{Conclusion and Open Problems}\label{sec:open problem}

We establish space-complexity dichotomies for undirected subgraph finding with a constant number of passes in~\cref{sec:section3}, oriented subgraph finding with one pass in~\cref{sec:section5}, and induced pattern finding with polylogarithmically many passes in~\cref{sec:induced}. We also obtain clever single-pass algorithms for a class of oriented patterns. Whether such algorithms exist for undirected patterns or for oriented patterns outside this class remains open. We make partial progress toward these questions in~\cref{appendixB}. Color-coding allows disconnected patterns to be handled componentwise, so the hardest component determines the overall space complexity up to polylogarithmic factors. Thus we focus on connected patterns below.

\begin{question}
Does some fixed connected simple undirected pattern $H$ admit a single-pass algorithm for $\textsc{Sub}(H)$ using $o(\operatorname{ex}(n,H))$ space? Does some fixed oriented pattern $\vec H$ with connected underlying graph, other than an oriented tree with exactly one internal vertex, admit a single-pass algorithm for $\textsc{Sub}(\vec H)$ using $o(\operatorname{ex}(n,\vec H))$ space?
\end{question}

%% file: AppendixA.tex
\section{Preliminaries}\label{AppendixA}

\subsection{Turán Number}\label{sec:turan}

We recall the Turán number and introduce the bounds that we use throughout the paper.

\begin{definition}[Tur\'{a}n number]
    Let $H$ be a fixed graph, then the Tur\'{a}n number $\mathrm{ex}(n,H)$ is the maximum number of edges an $n$-vertex $H$-free graph can have. Likewise, let $\vec H$ be a fixed oriented graph, then the Tur\'{a}n number $\mathrm{ex}(n,\vec H)$ is the maximum number of edges an $n$-vertex $\vec H$-free oriented graph can have. 
\end{definition}
We will use the following  characterizations of when these Turán numbers are subquadratic.
\begin{theorem}[Erdős–Stone–Simonovits \cite{EST46,ErdosSimonovits66}]
    Let $H$ be a fixed graph with chromatic number $\chi(H)$, then
    \begin{align*}
        \mathrm{ex}(n,H)=\left(\frac{\chi(H)-2}{\chi(H)-1}+o(1)\right)\binom n2.
    \end{align*}
\end{theorem}

\begin{theorem}[Valadkhan~\cite{Valadkhan09}, see also \cite{Grzesik23}]
\label{thm:valadkhan}
    Let $\vec H$ be a fixed oriented graph. Let $z(\vec H)$ be the minimum integer $k$ such that $\vec H$ is homomorphic to every tournament on $k$ vertices, or $z(\vec H)=\infty$ otherwise, then
    \begin{align*}
        \mathrm{ex}(n,\vec H)=\left(\frac{z(\vec H)-2}{z(\vec H)-1}+o(1)\right)\binom n2.
    \end{align*}
\end{theorem}
In particular,  we have $\text{ex}(n,H)=o(n^2)$ if and only if $H$ is bipartite. Also, $z(\vec H)=2$ if and only if $\vec H$ is bipartite and WO, in which case $\text{ex}(n,\vec H)=o(n^2)$. In fact, for bipartite $H$ there is a polynomial gap between $\text{ex}(n,H)$ and $n^2$ as follows.
\begin{theorem}[Alon, Krivelevich, and Sudakov~\cite{AlonKS03}]\label{AlonBound}
    Let $H$ be bipartite with parts $A(H),B(H)$, and
    $
        \Delta'(H)=\min \left\{\max_{v\in A(H)}\mathrm{deg}(v),\max_{v\in B(H)}\mathrm{deg}(v)\right\}$, 
    then $\mathrm{ex}(n,H)=O(n^{2-1/\Delta'(H)})$.
\end{theorem}
The above result is proved by dependent random choice, and the proof easily generalizes to $\text{ex}(n,\vec H)$ for WO bipartite $\vec H$.
\begin{lemma}[Dependent random choice~\cite{fox2011dependent}]\label{DRC}
Let $G$ be a  bipartite graph on parts $V(G)=V_1\sqcup V_2$ with $|V_1|=|V_2|=n/2$ and $d=2|E(G)|/n$. If $a,b,r\in\mathbb N$ and
\begin{align*}
    \frac{d^r}{n^{r-1}}-\binom nr\left(\frac{b-1}{n}\right)^r\geq a,
\end{align*}
then there is some $A_0\subseteq V_1$, $|A_0|\geq a$ such that every $r$ vertices from $A_0$ share at least $b$ common neighbors.
\end{lemma}
\begin{corollary}\label{Oex}
    Let $\vec H$ be  WO with respect to the bipartition $V(H)=A(H)\sqcup B(H)$, and
    $
        \Delta'(H)=\min \left\{\max_{v\in A(H)}\mathrm{deg}_H(v),\max_{v\in B(H)}\mathrm{deg}_H(v)\right\}$,
    then $\mathrm{ex}(n,\vec H)=O(n^{2-1/\Delta'(H)})$.
\end{corollary}
\begin{proof}
    Assume without loss of generality that $\Delta'(H)=\max_{v\in B(H)}\text{deg}_H(v)$. Let $a=|A(H)|$, $b=|B(H)|$, and $C>0$ be a constant such that
    \begin{align*}
        (2C)^{\Delta'(H)}-\binom n{\Delta'(H)}\left(\frac{b-1}{n}\right)^{\Delta'(H)}\geq a.
    \end{align*}
    We show that every $n$-vertex oriented graph $\vec G_0$ with $|E(G_0)|\geq 5Cn^{2-1/\Delta'(H)}$ contains a copy of $\vec H$. Indeed, such a $\vec G_0$ has a bipartite subgraph $\vec G_1$ on balanced parts $V(\vec G_1)=V_1\sqcup V_2$ with $|E(\vec G_1)|\geq 2Cn^{2-1/\Delta'(H)}$, and then by choosing the major orientation between the parts we can further get a well-oriented subgraph $\vec G$ with $V(\vec G)=V_1\sqcup V_2$ and $|E(\vec G)|\geq Cn^{2-1/\Delta'(H)}$. Note that for $d=2|E(\vec G)|/n$ we have
    \begin{align*}
        \frac{d^{\Delta'(H)}}{n^{\Delta'(H)-1}}-\binom n{\Delta'(H)}\left(\frac{b-1}{n}\right)^{\Delta'(H)}\geq (2C)^{\Delta'(H)}-\binom n{\Delta'(H)}\left(\frac{b-1}{n}\right)^{\Delta'(H)}\geq a.
    \end{align*}
    Assume without loss of generality that the orientation on $\vec H$ is from $A(H)$ to $B(H)$ and the orientation on $\vec G$ is from $V_1$ to $V_2$. By \cref{DRC}, choose $A_0\subseteq V_1$ in which every $\Delta'(H)$ vertices share at least $b$ common neighbors in $V_2$, and take any bijection
    \begin{align*}
        \varphi:A(H)\to A_0.
    \end{align*}
    Let $B(H)=\{v_1,\ldots,v_b\}$. First consider $N_H(v_1)\subseteq A(H)$. Since $|N_H(v_1)|\leq \Delta'(H)$, by the choice of $A_0$ we know $\varphi (N_H(v_1))$ has at least $b$ common neighbors in $V_2$. Pick any $u_1$ from these common neighbors and extend the bijection to
    \begin{align*}
        \varphi:A(H)\cup \{v_1\}\to A_0\cup \{u_1\}.
    \end{align*}
    Inductively, if we have the bijection
    \begin{align*}
        \varphi:A(H)\cup \{v_1,\ldots,v_k\}\to A_0\cup \{u_1,\ldots,u_k\}
    \end{align*}
    for some $k<b$, then consider $N_H(v_{k+1})\subseteq A(H)$. Since $|N_H(v_{k+1})|\leq \Delta'(H)$, again we know $\varphi (N_H(v_{k+1}))$ has at least $b$ common neighbors in $V_2$, and in particular one of these common neighbors is not in $\{u_1,\ldots,u_k\}$, so we can let it be $u_{k+1}$ and then extend $\varphi$ by mapping $v_{k+1}\mapsto u_{k+1}$. Finally, the bijection
    \begin{align*}
        \varphi:A(H)\cup B(H)\to A_0\cup \{u_1,\ldots,u_b\}
    \end{align*}
    gives a copy of $\vec H$ in $\vec G$.
\end{proof}

Moreover, a forest $H$ on $t$ vertices has $\mathrm{ex}(n,H)=O_t(n)$~\cite{FurediS13}, since a graph with at least $(t-1)n$ edges contains a subgraph of minimum degree at least $t-1$, allowing a greedy embedding of any $t$-vertex forest. Similarly, a WO oriented forest $\vec H$ on $t$ vertices has $\mathrm{ex}(n,\vec H)=O_t(n)$~\cite{Graham70}, since an oriented graph with at least $4tn$ arcs contains a WO bipartite subgraph of minimum degree at least $t$, allowing a greedy embedding of any $t$-vertex WO forest that respects the orientation.

\subsection{Communication Complexity}
Our streaming lower bounds are proved via reductions from standard communication problems.
Concretely, we simulate any $p$-pass streaming algorithm by a communication protocol,
partitioning the stream among the players and passing the algorithm's memory state between them.
Therefore, lower bounds on the communication cost of the underlying problems imply space lower bounds
for streaming algorithms.

We introduce the communication games used in our streaming lower bounds. Fix a finite set \(I\) of size \(M\), known to all players. Each player receives only the input specified below. All communication costs are measured in bits, and all randomized protocols are required to succeed with probability at least \(2/3\). See~\cite{KushilevitzNisan97} for more background on communication \nolinebreak complexity. 

\subsubsection{\textsc{Index}}
In \(\textsc{Index}\), Alice receives a set \(S\subseteq I\), while Bob receives an element \(i\in I\). Alice sends a single message to Bob, who must determine whether \(i\in S\). Any randomized one-way protocol for \(\textsc{Index}\) requires \(\Omega(M)\) bits of communication~\cite{KremerNisanRon99}.

\subsubsection{\textsc{Set-Disjointness}}

In \(\textsc{Set-Disjointness}\), Alice and Bob receive sets \(S_1,S_2\subseteq I\), respectively, and must determine whether \(S_1\cap S_2\neq\varnothing\). Any randomized protocol for \(\textsc{Set-Disjointness}\) requires \(\Omega(M)\) bits of communication~\cite{KS92}.

\subsubsection{\textsc{Multi-Disjointness}}

In \(\textsc{Multi-Disjointness}\), there are \(r\ge2\) players, where player \(q\in[r]\) receives a set \(S_q\subseteq I\). We use the promise version in which the input satisfies exactly one of the following:
\[
\begin{array}{ll}
\text{\rm(YES)} & \displaystyle\bigcap_{q=1}^r S_q\neq\varnothing,\\[1ex]
\text{\rm(NO)} & S_1,\ldots,S_r\text{ are pairwise disjoint}.
\end{array}
\]
The players must distinguish between the two cases.
Each player initially knows only its own set \(S_q\), and communication consists of writing bits on a common blackboard visible to all players. The communication cost is the total number of bits written. This is known as the \emph{number-in-hand blackboard model}, and for every fixed \(r\ge2\), any randomized protocol for \(\textsc{Multi-Disjointness}\) in this model requires \(\Omega(M)\) bits of communication \cite{Gronemeier09}.\\\\
The lower bounds needed for our streaming reductions are summarized as follows.
\begin{theorem}[\cite{KremerNisanRon99, Gronemeier09,Jayram09}]\label{prop:CCbound}
    Let \(I\) be a set of size \(M\).
    \begin{itemize}
        \item Every randomized one-way protocol for \(\textsc{Index}\) on \(I\) requires \(\Omega(M)\) bits of communication.
        \item For every fixed \(r\ge2\), every randomized number-in-hand blackboard protocol for \(r\)-player \(\textsc{Multi-Disjointness}\) on \(I\) requires \(\Omega(M)\) bits of communication. 
    \end{itemize}
\end{theorem}
We briefly explain how the second bound is used in multi-pass streaming reductions. Suppose the stream is divided into \(r\) consecutive blocks, with player \(q\) generating the \(q\)-th block from \(S_q\). To simulate a \(p\)-pass streaming algorithm using \(b\) bits of memory, the players process their blocks in order and pass the memory state from one player to the next; between consecutive passes, the last player passes the state back to the first. This uses \(pr-1\) messages of \(b\) bits each. Since \(r\) is fixed in our applications, the \(\Omega(M)\) lower bound above implies
$b=\Omega(M/p)$.
\subsection{\texorpdfstring{$k$}{k}-Cores}

The $k$-core identifies the largest part of a graph in which every vertex has at least $k$ neighbors within that part.

\begin{definition}[$k$-core]\label{def:kcore}
Let $G$ be an undirected graph and let $k\geq 1$ be an integer. The \textit{$k$-core} of $G$ is the unique largest induced subgraph with minimum degree at least $k$. We denote its vertex set by $K_k(G)$.
\end{definition}

The $k$-core can be obtained by repeatedly deleting vertices of degree less than $k$ in the remaining graph, regardless of the deletion order. Each deletion is safe: such a vertex cannot belong to any induced subgraph of the remaining graph with minimum degree at least $k$. Thus, no vertex of the $k$-core is ever deleted, and when the process terminates, the remaining graph is exactly the $k$-core. The result of this \textit{peeling process} is therefore independent of the deletion order. Moreover, adding edges can only enlarge $K_k(G)$, since the existing core still has minimum degree at least $k$ after the addition.

The peeling process also bounds the number of edges not within the core. Each deleted vertex has at most $k-1$ remaining neighbors. Counting each such edge when its first endpoint is deleted shows that at most $(k-1)n$ edges have an endpoint outside the core, where $n=|V(G)|$. This bound is useful for constructing sparse certificates that preserve all edges with an endpoint outside the core.

\subsection{Color-Coding}

Color-coding, introduced by Alon, Yuster, and Zwick~\cite{AlonYZ95}, is a technique for finding small subgraphs by randomly partitioning the input vertices into color classes. The coloring helps enforce vertex-disjointness: subgraphs using disjoint sets of colors cannot share a vertex. In its standard form, each input vertex receives a color from $[q]$ independently and uniformly. A copy on $q$ vertices is called \textit{colorful} if its vertices receive distinct colors. Any fixed copy becomes colorful with probability $\Omega_q(1)$; the same bound holds for any prescribed color assignment. When copies satisfying certain color constraints can be found efficiently, random coloring yields a randomized algorithm for finding a copy in the original graph.

An advantage of color-coding is that the coloring step can be derandomized. For colorful copies, it suffices to construct a small family of colorings such that every set of $q$ vertices is colorful under at least one coloring, and run the search for each coloring in the family. Such a family is called a \textit{$q$-perfect family of hash functions}. We use the following construction.

\begin{lemma}[Schmidt--Siegel \cite{SS90}]\label{ColorCoding}
Let $q\geq 1$ be a fixed integer and let $V$ be a set with $|V|=n$. One can explicitly construct a family $\mathcal F$ of colorings $\chi:V\to[q]$ with size $|\mathcal F|=O_q(\log^2 n)$ such that for every $V'\subseteq V$ with $|V'|=q$, we have $\chi(V')=[q]$ for some $\chi\in\mathcal F$.
\end{lemma}

In some applications, we want different parts of a copy to lie in prescribed color classes. To obtain such colorings with $r$ classes, we compose each $\chi\in\mathcal F$ with every map $f:[q]\to[r]$. Since some $\chi$ assigns distinct colors to the $q$ vertices of any fixed copy, an appropriate $f$ maps these colors to the prescribed classes. For fixed $q$ and $r$, considering all $r^q$ maps increases the family size by only a constant factor.

%% file: AppendixB.tex
\section{Further Space Bounds and Multi-Pass Results}\label{appendixB}

Sections~\ref{sec:section3} and~\ref{sec:section5} establish space dichotomies, but leave open questions about finer asymptotic space bounds and the role of additional passes.
In Section~\ref{sec:bipartite-tightness-lower}, we investigate whether the Turán number gives tight space bounds for undirected bipartite patterns, giving a partial answer to the following question.
\begin{question}[Extremal-number lower bounds]
For which fixed bipartite patterns $H$ do we have $R^p_{2/3}(\textsc{Sub}(H))=\Omega(\mathrm{ex}(n,H)/p)$, or at least $R^1_{2/3}(\textsc{Sub}(H))=\Omega(\mathrm{ex}(n,H))$?
\end{question}
Section~\ref{sec:undirected-pass-space-tradeoffs} investigates pass--space tradeoffs, giving matching bounds for odd cycles and more general upper bounds in terms of the Turán number, discussing the following question.
\begin{question}[Multi-pass upper bounds]\label{q:multipass-upper}
For every fixed non-bipartite pattern $H$, can $p$-pass detection be achieved using $\widetilde O(n^2/p)$ space for $p=O(n)$? More generally, for which fixed undirected patterns $H$ can a $p$-pass streaming algorithm detect $H$ using $\widetilde O(\mathrm{ex}(n,H)/p+n)$ space?
\end{question}
Section~\ref{sec:nwo-multipass-partial} presents partial results toward a multi-pass dichotomy for oriented patterns, as mentioned in the beginning of \cref{sec:section5}.

Throughout this appendix, all $H$-gadgets are understood in the sense of \cref{transgadget}. We call these \emph{transversal gadgets} to distinguish them from the gadgets of \cref{ReductionOriented}. The definition of transversal gadgets extends directly to oriented graphs $\vec H$ by replacing edges with arcs and requiring subgraph copies to preserve orientations, and in this case we call them \emph{transversal $\vec H$-gadgets}.

We also use the following consequence of the proof of \cref{lem:subgraph-space}. If for every sufficiently large integer $n$ there exists a $\Theta(n)$-vertex transversal gadget with parameters $(k,r)$, where $r$ is fixed and $k=k(n)$, then we get a  randomized $p$-pass $\Omega(k/p)$-space lower bound for the corresponding subgraph finding problem. This applies for both undirected $H$ and oriented \nolinebreak $\vec H$. Since $\mathrm{ex}(n,H)/\binom n2$ is nonincreasing in $n$ (by averaging over the $(n-1)$-vertex induced subgraphs of an extremal graph), we have $\mathrm{ex}(cn,H)=\Theta(\mathrm{ex}(n,H))$ for every constant $c>0$, so the bounds below also hold in terms of the number of vertices of the constructed instances.

\subsection{Lower Bounds for Undirected Bipartite Patterns}\label{sec:bipartite-tightness-lower}

Our dichotomy identifies which patterns admit subquadratic-space algorithms, but leaves open their precise space complexity. For bipartite graphs, the extremal number gives a natural upper bound, motivating us to ask whether it also yields corresponding single-pass and multi-pass lower bounds.

\subsubsection{Lower-Bound Constructions via Extremal Numbers}
We have obtained lower bounds for several classes of bipartite patterns. The following theorem summarizes our current results.

\begin{theorem}\label{thmUB}
    Let $H$ be a fixed bipartite graph.
    \begin{itemize}
        \item If $H$ is an even cycle, a forest with a component containing at least two edges, or has diameter less than $4$, then $R^p_{2/3}(\textsc{Sub}(H))=\Omega(\mathrm{ex}(n,H)/p)$.
        \item If $H$ has connectivity $>2$, then $R^1_{2/3}(\textsc{Sub}(H))=\Omega(\mathrm{ex}(n,H))$.
    \end{itemize}
\end{theorem}

We adapt the construction of Drucker, Kuhn, and Oshman \cite[Lemma 18]{DruckerKO14} to $H$-gadgets for even cycles $H=C_{2\ell}$. We take two copies of a dense $C_{2\ell}$-free graph and join corresponding vertices by paths of length 
$\ell-1$. Each pair of corresponding edges forms one of the sets $E_i$ in \cref{transgadget}.
\begin{proposition}\label{ECprop}
    Fix $\ell\geq 2$. For every sufficiently large $n$, there is an $
    \ell n$-vertex transversal $C_{2\ell}$-gadget with parameters $(\text{ex}(n,C_{2\ell}),2)$. Therefore, 
    $R^p_{2/3}(\textsc{Sub}(C_{2\ell}))=\Omega(\mathrm{ex}(n,C_{2\ell})/p)$.
\end{proposition}
\begin{proof}
    Take an $n$-vertex $C_{2\ell}$-free graph $G_1$ with vertices $v_1,\ldots,v_n$ and edges $e_1,\ldots,e_k$, where $k=\text{ex}(n,C_{2\ell})$. Let $G_2$ be a vertex-disjoint copy of $G_1$, with corresponding vertices $v_j'$ and edges $e_i'$. For each $j\in[n]$, join $v_j$ to $v_j'$ by a path $Q_j$ of length $\ell-1$, using new internal vertices, with these paths pairwise vertex-disjoint. Denote the resulting graph by $G$, and put
\begin{align*}
    E_i=\{e_i,e_i'\}\quad \forall i\in[k],\quad F=E(G)\setminus \bigcup_{i\in [k]}E_i=\bigcup_{j\in[n]}E(Q_j).
\end{align*}
We have $|V(G)|=\ell n$, so to finish the proof it remains to show that $G$ is a $C_{2\ell}$-gadget with parameters $(k,2)$ by checking the two conditions in \cref{transgadget}.

If $e_i=\{v_a,v_b\}$, then $e_i,e_i'$ together with $Q_a$ and $Q_b$, form a $C_{2\ell}$. Thus $(V(G),F\cup E_i)$ contains a copy of $C_{2\ell}$ for every $i\in[k]$, and the first gadget condition holds.
As for the second condition, let $Y$ contain exactly one edge from each $E_i$, and suppose for contradiction that $(V(G),F\cup Y)$ contains a copy $C$ of $C_{2\ell}$. Since $G_1$ and $G_2$ are $C_{2\ell}$-free, $C$ must use an edge of some $Q_a$, and in that case it has to traverse the entire path $Q_a$, because each internal vertex of $Q_a$ has degree two. Thus to form a cycle, $C$ must traverse at least two connecting paths, say $Q_a$ and $Q_b$. These paths are vertex-disjoint, and they already contribute $2(\ell-1)$ to the total length of $C$. Therefore, 
$C$ consists of exactly these two connecting paths $Q_a,Q_b$ and the edges $\{v_a,v_b\}$, $\{v_a',v_b'\}$. The latter are corresponding edges $e_i,e_i'$ for some $i$, contradicting the choice of $Y$. Hence $(V(G),F\cup Y)$ is $C_{2\ell}$-free, as needed.
\end{proof}

Next we construct transversal $H$-gadgets for any forest $H$ with a component containing at least two edges. The construction is even simpler since $\text{ex}(n,H)=\Theta(n)$ (see \cref{sec:turan}).

\begin{proposition}\label{forestLB}
    If $H$ is a fixed forest with a component containing at least two edges, then
    for every sufficiently large $n$, there is an $
    N$-vertex transversal $H$-gadget with parameters $(n,2)$, where $N=\Theta(n)$ depends only on $H$. Therefore, $R^p_{2/3}(\textsc{Sub}(H))=\Omega(n/p)$.
\end{proposition}
\begin{proof}
    Let $H_1$ be a largest component of $H$, having the most edges. Since $H_1$ has at least two edges, we can fix two distinct edges $e_1,e_2\in H_1$. Take $n$ vertex-disjoint copies $H_1^{(1)},\ldots,H_1^{(n)}$ of $H_1$, and let $e_1^{(i)},e_2^{(i)}$ denote the copies of $e_1,e_2$ in $H_1^{(i)}$, respectively. We also fix one vertex-disjoint copy of $H-H_1$. Define a graph $G$ on the vertex set
    \begin{align*}
        V(G)=V(H-H_1)\sqcup \bigsqcup_{i=1}^nV(H_1^{(i)}),
    \end{align*}
    so $|V(G)|=\Theta(n)$. For each $i\in [n]$, put
    \begin{align*}
        E_i=\{e_1^{(i)},e_2^{(i)}\},\quad F=E(G)\setminus \bigsqcup_{i\in [n]}E_i=E(H-H_1)\sqcup \bigsqcup_{i=1}^n(E(H_1^{(i)})\setminus \{e_1^{(i)},e_2^{(i)}\}).
    \end{align*}
    It remains to show that $G$ is a transversal $H$-gadget with parameters $(n,2)$ by checking the two conditions in \cref{transgadget}.
    
    For every $i\in [n]$, the graph $(V(G),F\cup E_i)$ contains the copy $H_1^{(i)}$, and together with the fixed copy of $H-H_1$, they form a copy of $H$. Thus the first gadget condition is satisfied. As for the second condition, let $Y$ contain exactly one edge from each $E_i$ and consider $(V(G),F\cup Y)$. In this graph, for each $i\in [n]$ the subgraph on  $V(H_1^{(i)})$ is obtained from $H_1^{(i)}$ by deleting one edge. Thus if $c$ is the number of components in $H$ with exactly $|E(H_1)|$ edges, then the fixed copy of $H-H_1$ contains exactly $c-1$ such components, and none of the incomplete copies $H_1^{(i)}$ contains such a component. Thus in $(V(G),F\cup Y)$ there are only $c-1$ such components and hence it cannot contain a copy of $H$, as needed.
\end{proof}

The remaining forests, in which every component has at most one edge, form an exceptional case. For such $H$, we cannot choose two distinct edges $e_1,e_2$ in the same component $H_1$, as required by the preceding reduction. In fact, the following proposition gives a deterministic one-pass algorithm using $\widetilde O(1)$ space, which is much smaller than $\text{ex}(n,H) = \Omega(n)$.

\begin{proposition}
    Let $H$ be a graph with $\ell$ edges in which every component has at most one edge, then $D^1(\textsc{Sub}(H))=\tilde O(1)$.
\end{proposition}

\begin{proof}
    If some $v\in V(G)$ has degree at least $2\ell$, then removing one of its incident edges does not affect the existence of $H$ in $G$. Indeed, if there is a copy $H'$ of $H$ in $G$ and $v\in V(H')$, then $N(v)\setminus{V(H')}\ne \emptyset$, so there will still be a copy of $H$ no matter which incident edge is removed.
    
    On the other hand, a graph with maximum degree $2\ell-1$ can have at most $(\ell-1)(4\ell-3)$ edges if it is $H$-free. This is because in such a graph, every edge is incident to at most $(4\ell-4)$ other edges, so if there were more than $(\ell-1)(4\ell-3)$ edges, then we can find a copy of $H$ by repeatedly selecting an edge and marking the incident edges as non-selectable.
    
    The following algorithm works by the two facts above. We retain a graph $G'$, initialized to \nolinebreak $\emptyset$.\linebreak On each input edge $e\in E(G)$, we store $e$ in $G'$ if and only if this keeps every vertex of $G'$ at degree less than $2\ell$. During this process, if at any point $|E(G')|>(\ell-1)(4\ell-3)$, then we know $H\subseteq G'$ and can find one. Otherwise, at the end of the stream we still have $|E(G')|\leq (\ell-1)(4\ell-3)$, and then we examine whether $H\subseteq G'$, which will be equivalent to whether $H\subseteq G$. In both cases, any $\ell$ pairwise disjoint edges of $G'$ extend to a copy of $H$, since the isolated vertices of $H$ can be mapped to arbitrary unused vertices, as $n\geq |V(H)|$.
\end{proof}

Back to the $p$-pass $\Omega(\text{ex}(n,H)/p)$ lower bound, we know that it can be generalized to some other bipartite $H$, but as stated in Theorem~\ref{thmUB}, in general we either need an extra condition on the diameter of $H$, or can only derive single-pass lower bounds, where there are still uncovered cases. 

We first derive the multi-pass lower bound.
Since $H$ is bipartite, the assumption that $H$ has diameter less than $4$ ensures that
any two distinct vertices in the same bipartition class of $H$
have a common neighbor. We use this property to construct transversal $H$-gadgets.

\begin{proposition}\label{diam4LB}
    Let $H$ be a fixed bipartite graph of diameter less than $4$, then
for every sufficiently large $n$, there is an
$N$-vertex transversal $H$-gadget with parameters
$(k,|E(H)|)$, where $k=\Theta(\operatorname{ex}(n,H))$ and $N=\Theta(n)$ depend only on $H$.
Thus
$R^p_{2/3}(\textsc{Sub}(H))
=\Omega({\text{ex}(n,H)}/{p})$.
\end{proposition}

\begin{proof}
    Let $G'$ be an $H$-free bipartite graph with parts
$L=\{g_1,\ldots,g_n\}$,
$R=\{g'_1,\ldots,g'_n\}$ 
and $k=\Theta(\operatorname{ex}(n,H))$ edges.
Such a graph can be obtained by first taking a $2n$-vertex extremal
$H$-free graph and then taking a bipartite subgraph of it on balanced parts preserving at least half of the original edges. We write
\begin{align*}
    E(G')=\{f_1,\ldots,f_k\},
\quad
f_s=\{g_{i_s},g'_{j_s}\}\quad \forall 1\leq s\leq k.   
\end{align*}
Let the bipartition classes of $H$ be
$A=\{v_1,\ldots,v_a\}$, 
$B=\{u_1,\ldots,u_b\}$ 
and put $r=|E(H)|$.
For each $i\in[n]$, take copies $A_i=\{v_1^{(i)},\ldots,v_a^{(i)}\}$ and
$B_i=\{u_1^{(i)},\ldots,u_b^{(i)}\}$, 
with all these sets pairwise disjoint.
Define $G$ on the vertex set
\begin{align*}
    V(G)=\bigsqcup_{i=1}^{n}(A_i\sqcup B_i),
\end{align*}
so $|V(G)|=|V(H)|n=\Theta(n)$.
For each $s\in[k]$, put
\begin{align*}
    E_s=
\left\{
\{v_\alpha^{(i_s)},u_\beta^{(j_s)}\}:
\{v_\alpha,u_\beta\}\in E(H)
\right\},\quad  E(G)=\bigsqcup_{s=1}^{k}E_s,
\quad F=\emptyset.
\end{align*}
The sets $E_s$ are pairwise disjoint and each has size $r$,  so it remains to show that $G$ is a transversal $H$-gadget by checking the two conditions in \cref{transgadget}.  

The intuition behind $G$ is that we replace each vertex $g_i$ of the skeleton $G'$
by the set $A_i$, and each vertex $g'_j$ by the set $B_j$.
Each skeleton edge $f_s=\{g_{i_s},g'_{j_s}\}$ then supports
a copy of $H$ between $A_{i_s}$ and $B_{j_s}$, whose edges
form the set $E_s$. The first gadget condition is then easily verified since $(V(G),E_s)$ contains a copy of $H$ on
$A_{i_s}\cup B_{j_s}$.

To check the second gadget condition, let $Y$ contain exactly one edge from each $E_s$, and put
$G_Y=(V(G),Y)$.
Note that no two distinct vertices in any $A_i$ can have a common neighbor
in $G_Y$, because a common neighbor in some $B_j$ would require
two selected edges from the same set $E_s$, where
$f_s=\{g_i,g'_j\}$.
Likewise, no two distinct vertices in any $B_j$ can have a common neighbor
in $G_Y$.

Now suppose for contradiction that $G_Y$ contains a copy $H'$ of $H$.
Since $H$ is connected, each bipartition class of $H'$
lies entirely in one of the two parts
$\bigcup_{i\in [n]} A_i$ and $\bigcup_{j\in [n]} B_j$.
Any two distinct vertices in the same bipartition class
of $H'$ have a common neighbor, because
$\operatorname{diam}(H)<4$.
Therefore,
\begin{align*}
    |V(H')\cap A_i|,
|V(H')\cap B_i|\le 1
\quad \forall i\in[n].    
\end{align*}
However, consider the map
\begin{align*}
    \pi:V(G)\to V(G'),\quad v_\alpha^{(i)}\mapsto g_i,
\quad
u_\beta^{(j)}\mapsto g'_j.    
\end{align*}
By construction, $\pi$ maps every edge of $G_Y$ to an edge
of $G'$, and since $V(H')$ intersects each $A_i$ and each $B_j$ in at most one vertex, $\pi$ is
injective on $V(H')$, so its restriction to $H'$ gives
a copy of $H$ in $G'$, which contradicts the choice of $G'$.
\end{proof}
\begin{figure}[H]
        \centering
        \includesvg[width=0.5\linewidth]{diam4bipartite.svg}
        \caption{Constructing the transversal $H$-gadget in Proposition~\ref{diam4LB} from the $H$-free skeleton $G'$.}
    \end{figure}
For the single-pass lower bound, we use a reduction from
\textsc{Index} based on the same edge-completion idea
as the $\vec H$-gadgets in \cref{ReductionOriented}, which also appears in \cite[Proposition 4]{KariMRS15}, but applied to undirected bipartite graphs.
To be specific, the complete bipartite base graph $\vec K$ in \cref{ReductionOriented} is replaced by a dense
$H$-free graph $G'$, whose edges encode Alice's input.
Bob attaches a copy of $H$ with one edge removed, identifying
the endpoints of the missing edge with those of his queried
edge in $G'$.
If Alice included that edge, it completes the attached copy
of $H$.
For $H$ with connectivity greater than $2$,
the converse follows because a copy of $H$ cannot cross
between the base graph and the attachment through their two
shared vertices, and neither side alone contains $H$
when the queried edge is absent.
We implement this reduction directly below.

\begin{proposition}\label{k2LB}
    Let $H$ be a fixed bipartite graph. If $H$ has connectivity   $\kappa(H)>2$, then $R^1_{2/3}(\textsc{Sub}(H))=\Omega(\mathrm{ex}(n,H))$. 
\end{proposition}

\begin{proof}
    Let $H_0$ be a copy of $H$ on vertex set $\{v_1,\ldots,v_{|V(H)|}\}$ and fix any $\{v_i,v_j\}\in E(H_0)$. For every sufficiently large $n$, let $G'$ be an $n$-vertex $H$-free graph with $\text{ex}(n,H)$ edges and write
    \begin{align*}
        E(G')=\{e_1,\ldots,e_{\text{ex}(n,H)}\},\quad e_i=\{v_{i,1},v_{i,2}\},\quad v_{i,1},v_{i,2}\in V(G').
    \end{align*}
    Given an instance of $\textsc{Index}$, where Alice holds $S\subseteq [\text{ex}(n,H)]$ and Bob holds $k\in [\text{ex}(n,H)]$, we define the reduction graph $G_H$ by
    \begin{align*}
        &V(G_H)=
        V(G')\sqcup V(H_0)\setminus\{v_i,v_j\},\quad E(G_H)=E_1\sqcup E_2\sqcup E_3,\\
        &E_1=\{e_\ell:\ell\in S\},\quad E_2= \{e\in E(H_0):\{v_i,v_j\}\cap e=\emptyset\},\\
        & E_3= \Big\{\{v_{k,1},v_\ell\}:\{v_i,v_\ell\}\in E(H_0),\ v_\ell\ne v_j\Big\}\sqcup \Big\{\{v_{k,2},v_\ell\}:\{v_j,v_\ell\}\in E(H_0),\ v_\ell\ne v_i\Big\}.
    \end{align*}
    That is, we use $E(G')$ to encode $S$ and attach a copy of $H$ by overlapping $\{v_i,v_j\}$ with $\{v_{k,1},v_{k,2}\}$, where $k$ is the target index. 
    We show that this instance of $\textsc{Index}$ can be reduced to $\textsc{Sub}(H)$ on $G_H$, and this finishes the proof since $|V(G_H)|=\Theta(n)$.

    If $(S,k)$ is a YES instance, then
     $k\in S$ and hence  $\{v_{k,1},v_{k,2}\}\in E(G_H)$, so a copy of $H$ is formed by $\{v_{k,1},v_{k,2}\}\sqcup V(H_0)\setminus \{v_i,v_j\}$. Conversely, suppose 
     $k\notin S$ but
     there is a copy $H'$ of $H$ in $G_H$, then $\{v_{k,1},v_{k,2}\}\notin E(G_H)$, so 
        \begin{align*}
        &H'\not\subseteq G_H[V(G')],\quad   H'\not\subseteq G_H[\{v_{k,1},v_{k,2}\}\sqcup V(H_0)\setminus\{v_i,v_j\}],\\
        &V(H')\cap (V(G')\setminus \{v_{k,1},v_{k,2}\} )\ne \emptyset,\quad V(H')\cap (V(H_0)\setminus \{v_i,v_j\})\ne\emptyset,
    \end{align*}
    and $V(H')\cap \{v_{k,1},v_{k,2}\}$ is a vertex-cut of size at most 2 in $H'$, contradicting $\kappa(H)>2$.
\end{proof}

\subsubsection{Obstacles to Extending the Bounds to General Cases}\label{sec:bipartite-lower-bound-obstacles}
Our reductions in Propositions~\ref{diam4LB} and~\ref{k2LB} rely on diameter or connectivity assumptions to rule out unintended copies of $H$. The examples below illustrate the obstacles to extending these arguments beyond these cases.
\begin{figure}[H]
        \centering
        \includesvg[width=0.5\linewidth]{diam4counter.svg}
        \caption{An example where $|Y\cap E_s|\leq 1$ for each $s$ but still $H\subseteq G_Y$. See \cref{Question3}.}
        \label{proofWrong}
    \end{figure}
\begin{question}\label{Question3}
    For general bipartite graphs $H$ with diameter $\geq 4$, what is $R^p_{2/3}(\textsc{Sub}(H))$?
    Do we still have the $\Omega(\mathrm{ex}(n,H)/p)$ lower bound?
\end{question}  

If $\text{diam}(H)\ge 4$, then the construction in
\cref{diam4LB} might not give a transversal $H$-gadget.
In the proof of \cref{diam4LB}, although each set $E_s$ still forms a copy
of $H$, the argument that every $G_Y$ is $H$-free
may fail.
This is because now two vertices in the same bipartition class of $H$
might not have a common neighbor, so even if $G_Y$ contains a copy $H'$ of $H$,
we can no longer deduce that
$|V(H')\cap A_i|,
|V(H')\cap B_j|\le 1$, and
hence the projection $\pi:V(G)\to V(G')$ might not be injective on $V(H')$. Then the existence of $H'$
might not contradict the fact that $G'$ is $H$-free.
 For example, \cref{proofWrong} indicates how the proof may go wrong when $H=C_8$.

\begin{question}
    What is $R^1_{2/3}(\textsc{Sub}(H))$ for general bipartite graphs $H$ with connectivity $\kappa(H)\leq 2$? Do we still have the $\Omega(\mathrm{ex}(n,H))$ lower bound?
\end{question} 

We proved $R^1_{2/3}(\textsc{Sub}(H))=\Omega(\text{ex}(n,H))$ for bipartite $H$ with $\kappa(H)>2$ in \cref{k2LB}. In the proof, we attached a copy $H_0$ of $H$ by overlapping an $\{v_i,v_j\}\in E(H_0)$ with the target index edge $e_k$. Then, calling a copy $H'$ of $H$ in  $G_H$ ``crossing'' if it uses both vertices from $V(G')$ and from $V(H_0)\setminus \{v_i,v_j\}$,
 by $\kappa(H)>2$ we showed that if $e_k\notin E(G_H)$, then there is no crossing copies of $H$, and hence $H\nsubseteq G_H$.
However, when turning to the case $\kappa(H)=2$, we face obstacles. To  formulate them, we first define the \emph{two-vertex-cut-components} of $H$.
 \begin{definition}[Two-vertex-cut-component, TVCC]
     Let $H$ be bipartite with $\kappa(H)=2$ and $\{u,v\}$ be a two-vertex-cut of $H$, where $u,v$ are from different parts and $\{u,v\}\notin E(H)$. Let $A'$ be a component in $H-u-v$, and let $A(u,v)=H[\{u,v\}\cup V(A')]$, the subgraph formed by adding $u,v$ back to the component $A'$ in the remaining graph.
    We say $A(u,v)$ is a two-vertex-cut-component (TVCC) of $H$.
 \end{definition}
 
If $\kappa(H)=2$, then in the proof of \cref{k2LB}, when the index edge $e_k\notin E(G_H)$, although both $G_H[V(G')]$ and $G_H[V(H_0)]$ are $H$-free, a crossing $H$ copy still might exist since one of $G_H[V(G')]$, $G_H[V(H_0)]$ might contain a TVCC of $H$ and the other might contain the remaining part.

One might ask if this can be prevented by carefully choosing the $\{v_i,v_j\}\in E(H)$ to be overlapped with the index edge $e_k$. Indeed, if there is an $\{v_i,v_j\}\in E(H)$ such that no TVCC $A(u,v)$ of $H$ can be embedded into $H$ with the identification $\{u,v\}\mapsto \{v_i,v_j\}$, then we can rule out any crossing copy in the proof. However, we are obstructed by the following result.
\begin{figure}[H]
        \centering
        \includesvg[width=0.7\linewidth]{connect2counter.svg}
        \caption{An example for Proposition~\ref{YC}.}
        \label{YCgraph}
    \end{figure}
\begin{proposition}\label{YC}
    There exists a bipartite $H$ with $\kappa(H)=2$ such that for all $e\in E(H)$, there is a TVCC $A(u,v)$ of $H$ and an embedding $\varphi:A(u,v)\hookrightarrow H$ such that $\varphi:\{u,v\}\mapsto e$.
\end{proposition}
\begin{proof}    
    An example is given by \cref{YCgraph}. In the statement of \cref{YC}, if $e$ is chosen to be any red edge, then $u,w_1,w_2,v$ form the required TVCC. On the other hand, no matter which blue edge we chose to be $e$, some TVCC $A(u,v)$ can always be embedded so that $\{u,v\}\mapsto e$. For instance, if we take $e=\{u,w_3\}$, then the TVCC formed by $u,w_4,w_5,\ldots,w_{12},w_{13},v$ can be embedded by
    \begin{align*}
        &u\mapsto u,\quad v\mapsto w_3,\quad w_8\mapsto v,\quad 
        w_i\mapsto z_i\quad \forall 4\leq i\leq 13,\text{ }i\ne 8.\qedhere
    \end{align*}
\end{proof}

\subsection{Space--Pass Tradeoffs for Undirected Patterns}\label{sec:undirected-pass-space-tradeoffs}

We now turn from lower bounds to algorithms. The lower bounds discussed above give a natural target for $p$-pass algorithms, but whether this target can be achieved by a multi-pass streaming algorithm remains open in general. We obtain matching bounds for some patterns and weaker upper bounds for others.

\subsubsection{Tight Bounds for Odd Cycles}
For odd cycles, the simplest non-bipartite patterns, we use color-coding~\cite{AlonYZ95} to match the $\Omega(n^2/p)$ lower bound from Section~\ref{sec:section3} up to polylogarithmic factors.

Applying Lemma~\ref{ColorCoding} with $q=2\ell+1$ gives $L=O(\log^2 n)$ colorings of $V(G)$ such that every copy of $C_{2\ell+1}$ is colorful under at least one coloring. We search for such a copy using dynamic programming. Roughly speaking, we partition $V(G)$ into groups of vertices and process one group at a time, using its vertices as anchors over several streaming passes. For each anchor vertex and color set, we inductively record which vertices are reachable by a colorful path using exactly those colors, and use this information to determine whether a colorful cycle exists.

\begin{proposition}\label{oddcycleUB}
    Fix $\ell\geq 1$, then for all $p=O(n)$, we have $D^p(\textsc{Sub}(C_{2\ell+1}))=\tilde O(n^2/p)$.
\end{proposition}
\begin{proof}
    If $p<4\ell$, then $p$ is bounded by a constant, so we store the whole graph using $O(n^2)=O(n^2/p)$ bits in one pass and search it directly. Otherwise, let
    \[
    r=\min\left\{n,\left\lfloor\frac{p-2\ell+1}{2\ell+1}\right\rfloor\right\},
    \]
    so that $r\geq 1$, and $r=\Theta(p)$ since $p=O(n)$. Divide $V(G)=V_1\sqcup\ldots\sqcup V_r$ into almost balanced parts, every two of which differ in size by at most $1$. Take $L=O(\log^2n)$ colorings $\chi_1,\ldots,\chi_L:V(G)\to [2\ell+1]$ as above, using space $O(Ln)=\tilde O(n^2/p)$. Set $j=1$.
    \begin{enumerate}
        \item [(1)]
        Write $V_j=\{v_1^{(j)},\ldots,v_{|V_j|}^{(j)}\}$, where $|V_j|\leq\lceil n/r\rceil$. For all $i\in [|V_j|]$, $c\in [L]$, $0\leq d\leq 2\ell$ and $S\subseteq [2\ell+1]$ with $|S|=d+1$, we declare an array $D_d^{(i,c)}[S]:V(G)\to \{0,1\}$, using space
        \begin{align*}
            \tilde O(\lceil n/r\rceil L(2\ell)2^{2\ell+1}n)=\tilde O(n^2/p)
        \end{align*}
        in total, and initialize $D_d^{(i,c)}[S](x)=0$ for all $x\in V(G)$ and  $d\geq 1$. Our goal is to have
    \begin{align*}
        D_d^{(i,c)}[S](x)=1\Leftrightarrow{}&\text{under }\chi_c,\text{ }\exists\text{ a length-}d\text{ path  }v_i^{(j)}\leadsto x\\
        &\text{using exactly the colors in }S.
    \end{align*}
    To this end, for $d=0$ we need to initialize by
    \begin{align*}
        D_0^{(i,c)}[S](x)=\begin{cases}
            1,\text{ if }x=v_i^{(j)}\text{ and }S=\{\chi_c(v_i^{(j)})\},\\
            0,\text{ otherwise.}
        \end{cases}
    \end{align*}
    \item [(2)]
    Iteratively, for $d\geq 1$, suppose all $D_{d-1}^{(i,c)}[S]$ have been determined in the first $d-1$ passes, then in the $d$-th pass, whenever an edge $\{x,y\}$ arrives, we check if $D_{d-1}^{(i,c)}[S\setminus \{\chi_c(x)\}](y)=1$, setting $D_d^{(i,c)}[S](x)=1$ if this is the case. Likewise, if $D_{d-1}^{(i,c)}[S\setminus \{\chi_c(y)\}](x)=1$, then we set $D_d^{(i,c)}[S](y)=1$. This way, after the $(2\ell)$-th pass,  $D_{2\ell}^{(i,c)}[S]$ is determined for all $i,c,S$. Then in the $(2\ell+1)$-th pass we can check if any $w$ with $D_{2\ell}^{(i,c)}[S](w)=1$ is adjacent to $v_i^{(j)}$. If so, then we go to step (3), otherwise we go to step (4).\\
    \item [(3)] If we get to this step,
    then we know $w$ and $v_i^{(j)}$ are on a $C_{2\ell+1}$, and we store $v_i^{(j)},w$. In the $(2\ell+2)$-th pass, whenever an edge $\{w,y\}$ incident to $w$ arrives, we check if $D_{2\ell-1}^{(i,c)}[S'](y)=1$ for some $S'\not\ni \chi_c(w)$. If so, then $v_i^{(j)},w,y$ is a part of a $C_{2\ell+1}$ colorful under $\chi_c$, and we store them. Iteratively, after the $4\ell$-th pass we are guaranteed to find a $C_{2\ell+1}$ colorful under $\chi_c$, and we output this copy.\\
    \item [(4)]
    If we get to this step, then we know no vertex $w$ with $D_{2\ell}^{(i,c)}[S](w)=1$ is adjacent to \nolinebreak $v_i^{(j)}$. If $j=r$ then this means that $C_{2\ell+1}\nsubseteq G$. Otherwise we
    increase $j$ by $1$ and restart the process from step $(1)$. Note that all arrays $D_d^{(i,c)}[S]$ can be re-initialized and reused, so the total space is still $\tilde O(n^2/p)$.
    \end{enumerate}
    
The algorithm is correct because if there is a copy of $C_{2\ell+1}$, then it is colorful under some $\chi_c$, and hence 
it contains a vertex in $V_j$ for some $j$ and will be found. Moreover, for each group $V_j$, steps $(1)$ and $(2)$ use $2\ell+1$ passes: $2\ell$ passes to compute the arrays $D_1^{(i,c)},\ldots,D_{2\ell}^{(i,c)}$ and one pass to test adjacency to the anchors. Step $(3)$ is executed at most once and uses $2\ell-1$ further passes. Hence the total number of passes is at most $r(2\ell+1)+2\ell-1\leq p$ by the choice of $r$.
\end{proof}

\subsubsection{Upper Bounds via Vertex Cover}

The color-coding algorithm in the preceding subsection stores reachability information using $\Omega(n)$ bits per anchor vertex. Applying the same approach to even cycles or other bipartite graphs may therefore incur an $\Omega(n^2/p)$ space cost, which can be much larger than the trivial one-pass bound of $\widetilde O(\mathrm{ex}(n,H))$. In this subsection, we investigate an alternative upper bound in terms of the extremal number.
Ideally, we would like to obtain a $p$-pass space upper bound of $\widetilde O(\mathrm{ex}(n,H)/p+n)$ for every fixed bipartite graph $H$, which is the bipartite case of the second part of Question~\ref{q:multipass-upper}.

We do not yet know whether this tradeoff can be achieved for every bipartite pattern or whether it is too strong in general. The following proposition gives a weaker deterministic upper bound in terms of the extremal number, where the space decreases with the number of passes at a rate determined by $\tau(H)$, the vertex cover size of $H$.

\begin{proposition}\label{bipartite_upperbound}
Let $H$ be a fixed undirected graph with at least one edge, and let $\tau(H)$ be its minimum vertex-cover size. For every integer $p\geq 2$, we have
\[
    D^p(\textsc{Sub}(H))
    =\widetilde O\!\left(\frac{\mathrm{ex}(n,H)}{p^{1/\tau(H)}}+n\right).
\]
\end{proposition}

\begin{proof}
Let $S\subseteq V(H)$ be a minimum vertex cover of $H$, so that $|S|=\tau(H)$. Every edge of $H$ has at least one endpoint in $S$.

As in the na\"ive algorithm, the algorithm may depend on $n$ and $H$, so the value $\mathrm{ex}(n,H)$ can be hard-wired into it. We first restrict the stream to its first $\mathrm{ex}(n,H)+1$ edges, or the entire stream if it is shorter. This preserves whether the input contains $H$, so we may assume that
\[
    |E(G)|\leq \mathrm{ex}(n,H)+1.
\]
Let $1\leq r\leq n$ be an integer to be chosen later. In the first pass, record the degree of every vertex. Then partition $V(G)$ into $r$ possibly empty buckets $B_1,\ldots,B_r$, assigning each vertex to a bucket whose current degree sum is smallest, with ties broken by index. Before each assignment, the smallest degree sum is at most $2|E(G)|/r$, and assigning one vertex increases it by at most $n-1$. Thus, for every $i\in[r]$,
\[
    \sum_{v\in B_i}\deg_G(v)
    \leq \frac{2|E(G)|}{r}+n
    \leq \frac{2(\mathrm{ex}(n,H)+1)}{r}+n.
\]

We now enumerate all tuples $(i_1,\ldots,i_{\tau(H)})\in[r]^{\tau(H)}$, allowing repeated indices. For each tuple, let $U=B_{i_1}\cup\cdots\cup B_{i_{\tau(H)}}$.
Use one pass to store in $E$ all edges with at least one endpoint in $U$. The number of stored edges is at most
\[
    |E|
    \leq \sum_{v\in U}\deg_G(v)
    \leq \tau(H)\left(\frac{2(\mathrm{ex}(n,H)+1)}{r}+n\right).
\]
At the end of each pass, search $(V(G),E)$ for a copy of $H$, reusing the storage for $E$ between passes. Output a copy if one is found, or report $H\nsubseteq G$ if all tuples fail.

For correctness, fix an embedding $\varphi:H\to G$. Since $|S|=\tau(H)$, the vertices of $\varphi(S)$ lie in at most $\tau(H)$ buckets. Our enumeration therefore includes a tuple whose union $U$ contains $\varphi(S)$. Every edge of the copy has an endpoint in $\varphi(S)$, so all its edges are stored in $E$ during that pass. Conversely, every returned copy uses only edges of $G$.

The algorithm uses $1+r^{\tau(H)}$ passes. Choosing
\[
    r=\min\left\{n,\left\lfloor(p-1)^{1/\tau(H)}\right\rfloor\right\}
\]
keeps this number at most $p$. The degree counters and bucket assignments use $\widetilde O(n)$ space. Since $H$ is fixed, the total space is
\[
    \widetilde O\!\left(\frac{\mathrm{ex}(n,H)+1}{r}+n\right)
    =\widetilde O\!\left(\frac{\mathrm{ex}(n,H)}{p^{1/\tau(H)}}+n\right),
\]
where the case $r=n$ uses $\mathrm{ex}(n,H)=O(n^2)$.
\end{proof}
In particular, for an even cycle $H=C_{2\ell}$, we have $\tau(H)=\ell$. Choosing
\[
    p=\left\lceil\left(\frac{\mathrm{ex}(n,C_{2\ell})}{n}\right)^\ell\right\rceil
\]
in Proposition~\ref{bipartite_upperbound} yields a deterministic algorithm using $\widetilde O(n)$ space and $O(1+(\mathrm{ex}(n,C_{2\ell})/n)^\ell)$ passes. By the Bondy--Simonovits bound~\cite{BondySimonovits}, $\mathrm{ex}(n,C_{2\ell})=O(n^{1+1/\ell})$, this gives an $O(n)$-pass, $\widetilde O(n)$-space algorithm.

This neat result motivates us to ask whether the bound in Proposition~\ref{bipartite_upperbound} is optimal up to polylogarithmic factors for even cycles. A positive answer would rule out, for even cycles, the stronger tradeoff proposed in the second part of Question~\ref{q:multipass-upper}.

\begin{question}\label{even_cycle_tightness_p_version}
For every fixed $\ell\geq 2$ and every integer $2\leq p\leq (\mathrm{ex}(n,C_{2\ell})/n)^\ell$, do we have
\[
    R^p_{2/3}(\textsc{Sub}(C_{2\ell}))
    =\widetilde\Omega\!\left(\frac{\mathrm{ex}(n,C_{2\ell})}{p^{1/\ell}}\right)?
\]
\end{question}

The case of $\widetilde O(n)$ space connects this question to a classical conjecture of Erd\H{o}s that $\mathrm{ex}(n,C_{2\ell})=\Theta(n^{1+1/\ell})$ for every fixed $\ell\geq 2$. This is known for $\ell=2,3,5$ and remains open for all other values of $\ell$, see the survey of Füredi and Simonovits \cite{FurediS13}. This motivates the following question.

\begin{question}\label{even_cycle_tightness_n_version}
For every fixed $\ell\geq 2$, must every randomized streaming algorithm for $\textsc{Sub}(C_{2\ell})$ using $\widetilde O(n)$ space and succeeding with probability at least $2/3$ use $\widetilde\Omega(n)$ passes?
\end{question}

By Proposition~\ref{bipartite_upperbound}, a $\widetilde\Omega(n)$ pass lower bound for algorithms using $\widetilde O(n)$ space would imply
\[
    \mathrm{ex}(n,C_{2\ell})=\widetilde\Omega(n^{1+1/\ell}),
\]
establishing the conjectured order of magnitude up to polylogarithmic factors.
Also, if Question~\ref{even_cycle_tightness_p_version} has a positive answer, then Question~\ref{even_cycle_tightness_n_version} is equivalent, up to polylogarithmic factors, to the Erd\H{o}s conjecture.

\input{AppendixC}

%% file: AppendixC.tex
\subsection{Partial Multi-pass Classification for Oriented Patterns}\label{sec:nwo-multipass-partial}\label{AppendixC}
Recall that Section~\ref{sec:section5} classifies which oriented patterns require quadratic space in a single pass. The corresponding multi-pass classification remains open for NWO bipartite patterns, as noted at the beginning of Section~\ref{sec:section5}. The following theorem summarizes our partial results.

\begin{theorem}\label{thmOMP}
    Let $\vec H $ be an NWO bipartite graph.
    \begin{itemize}
        \item If $H$ is an even cycle, or $H=K_{s,t}$ with $2\leq s\leq t$, then $R^p_{2/3}(\textsc{Sub}(\vec H))=\Omega(n^2/p)$.
        \item If $H$ is a forest and $r$ is the largest radius of its components, then $D^{2r}(\textsc{Sub}(\vec H))=\tilde O(n)$.
    \end{itemize}
\end{theorem}

We first construct transversal $\vec H$-gadgets for every NWO even cycle $\vec H$,
giving the $p$-pass $\Omega(n^2/p)$-space lower bound,
just like non-bipartite graphs.
The construction resembles that of Proposition~\ref{ECprop},
but uses two WO orientations of $K_{n,n}$ in place of the
two undirected extremal graphs.
These provide $n^2$ distinguished pairs of arcs, while the
NWO assumption prevents a copy of $\vec H$ from lying
entirely within either complete bipartite graph.

\begin{proposition}\label{propNWOcycle}
Let $\vec H$ be a fixed NWO orientation of $C_{2\ell}$,
where $\ell\geq 2$.
For every sufficiently large integer $n$, there exists a
$2\ell n$-vertex transversal $\vec H$-gadget with parameters
$(n^2,2)$. Therefore,
$
R^p_{2/3}(\textsc{Sub}(\vec H))
=\Omega(n^2/p).
$
\end{proposition}

\begin{proof}
    Label the vertices of $\vec H$ in cyclic order as
$v_1,\ldots,v_{2\ell}$.
Removing the two opposite edges
$\{v_1,v_2\}$ and $\{v_{\ell+1},v_{\ell+2}\}$
from the underlying cycle leaves two oriented paths,
\[
\vec P=\vec H[\{v_2,\ldots,v_{\ell+1}\}],
\quad
\vec Q=\vec H[\{v_{\ell+2},\ldots,v_{2\ell},v_1\}],
\]
each of length $\ell-1$.
Take $n$ copies $\vec P^{(1)},\ldots,\vec P^{(n)}$ of $\vec P$
and $n$ copies $\vec Q^{(1)},\ldots,\vec Q^{(n)}$ of $\vec Q$,
all pairwise vertex-disjoint.
In $\vec P^{(j)}$, denote the vertices corresponding to
$v_2$ and $v_{\ell+1}$ by $y_j$ and $w_j$, respectively.
In $\vec Q^{(i)}$, denote the vertices corresponding to
$v_{\ell+2}$ and $v_1$ by $z_i$ and $x_i$, respectively.
Let
\[
F=
\bigsqcup_{i=1}^n (E(\vec P^{(i)})
\sqcup E(\vec Q^{(i)})).
\]
For every $(i,j)\in[n]\times[n]$, define the arcs
\[
\vec e_{i,j}=
\begin{cases}
(x_i,y_j), & \text{if }(v_1,v_2)\in E(\vec H),\\
(y_j,x_i), & \text{if }(v_2,v_1)\in E(\vec H),
\end{cases},\quad \vec f_{i,j}=
\begin{cases}
(w_j,z_i), &
\text{if }(v_{\ell+1},v_{\ell+2})\in E(\vec H),\\
(z_i,w_j), &
\text{if }(v_{\ell+2},v_{\ell+1})\in E(\vec H).
\end{cases}
\]
Put $
E_{i,j}=\{\vec e_{i,j},\vec f_{i,j}\}
$
and define $\vec G$ by
\[
V(\vec G)=\bigsqcup_{i=1}^n(V(\vec P^{(i)})\sqcup V(\vec Q^{(i)})),\quad E(\vec G)
=
F\sqcup
\bigsqcup_{(i,j)\in[n]\times[n]}E_{i,j}.
\]
We have $|V(\vec G)|=2\ell n$, while the sets $E_{i,j}$ are pairwise disjoint and each has size
two. Moreover, the arcs $\vec e_{i,j}$ form a WO orientation of $K_{n,n}$
between $\{x_1,\ldots,x_n\},\{y_1,\ldots,y_n\}$,
and the arcs $\vec f_{i,j}$ form another between
$\{w_1,\ldots,w_n\}, \{z_1,\ldots,z_n\}$. Therefore, it remains to show that $\vec G$ is a transversal $\vec H$-gadget by checking the two conditions in \cref{transgadget}.

For every $(i,j)$, the paths $\vec P^{(j)}$ and $\vec Q^{(i)}$,
together with $\vec e_{i,j}$ and $\vec f_{i,j}$, form a copy of
$\vec H$. Thus $(V(\vec G),F\cup E_{i,j})$ contains
$\vec H$, verifying the first gadget condition.
As for the second condition, let $Y$ contain exactly one arc
from each $E_{i,j}$, and put
\[
\vec G_Y=(V(\vec G),F\cup Y).
\]
Suppose for contradiction that $\vec G_Y$ contains a copy $\vec H'$ of
$\vec H$.
The subgraph consisting only of the arcs in $Y$ is WO,
but $\vec H'$ is NWO, so $\vec H'$ must use at least one arc from the connecting paths $\vec P^{(j)},\vec Q^{(i)}$.
Whenever $\vec H'$ uses an arc of a connecting path, it
traverses the entire path, since all internal vertices
of that path have degree two in the underlying graph.
Thus to form a cycle, $\vec H'$ must traverse at least two connecting paths. But these paths are vertex-disjoint, and they already contribute $2(\ell-1)$ to the total length of $\vec H'$. Therefore, $\vec H'$ consists of exactly two connecting paths and two more arcs joining their endpoints. By the construction of $\vec G$, these two connecting paths must be $\vec P^{(j)}, \vec Q^{(i)}$ for some $i,j$ and the two more arcs are $\vec e_{i,j},\vec f_{i,j}$, but this contradicts the choice of $Y$ since $\vec e_{i,j},\vec f_{i,j}$ are in the same $E_{i,j}$.
\end{proof}

Next we show that NWO complete bipartite graphs $\vec H$ also admit
transversal $\vec H$-gadgets, giving the $p$-pass $\Omega(n^2/p)$-space lower bound.
The construction starts with $n$ vertex-disjoint copies of
$\vec H-v$ and $n$ additional vertices representing $v$.
For each pair $(i,j)$, two arcs between the $i$-th copy
and the $j$-th additional vertex form the edge set $E_{i,j}$,
while all remaining arcs needed to complete a copy of $\vec H$
are fixed.
We choose $v$ and these two incident arcs so that any copy
of $\vec H$ must contain both arcs of some $E_{i,j}$.
The choice depends on which bipartition classes of $\vec H$
contain WO vertices, where a vertex is called \emph{WO} if it
is not internal, i.e., all its incident arcs are directed
towards it or all are directed away from it.

\begin{proposition}\label{kstLB}
Fix $2\leq s\leq t$, and let $\vec H$ be a fixed NWO
orientation of $K_{s,t}$.
For every sufficiently large integer $n$, there exists an
$N$-vertex transversal $\vec H$-gadget with parameters
$(n^2,2)$, where $N=\Theta(n)$ depends only on $\vec H$. Therefore,$
R^p_{2/3}(\textsc{Sub}(\vec H))
=\Omega(n^2/p)$.
\end{proposition}
\begin{proof}
    Let the two parts of $\vec H$ be $A$ and $B$, where
$|A|=s$ and $|B|=t$.
We first describe a construction of the transversal $\vec H$-gadget given a vertex $v$ with
neighbors $u_1,\ldots,u_d$. The correct choice of $v,u_1,u_2$
will be specified in the three cases below.

Take $n$ vertex-disjoint copies
$\vec H^{(1)},\ldots,\vec H^{(n)}$ of $\vec H-v$.
For each $u\in V(\vec H)\setminus\{v\}$, denote its copy
in $\vec H^{(i)}$ by $u^{(i)}$.
Also take a vertex set
$
V_v=\{v^{(1)},\ldots,v^{(n)}\}
$
disjoint from each $H^{(i)}$.
For every $h\in[d]$ and $i,j\in[n]$, define
\[
\vec e_{h,i,j}=
\begin{cases}
(u_h^{(i)},v^{(j)}),
    & \text{if }(u_h,v)\in E(\vec H),\\
(v^{(j)},u_h^{(i)}),
    & \text{if }(v,u_h)\in E(\vec H),
\end{cases}\quad 
E_{i,j}=\{\vec e_{1,i,j},\vec e_{2,i,j}\},
\]
and
\[
F=
\left(\bigsqcup_{i=1}^n E(\vec H^{(i)})\right)
\sqcup
\{\vec e_{h,i,j}:3\leq h\leq d,\ i,j\in[n]\}.
\]
Define $\vec G$ by
\[
V(\vec G)
=
V_v\sqcup\bigsqcup_{i=1}^n V(\vec H^{(i)}),
\qquad
E(\vec G)
=
F\sqcup\bigsqcup_{(i,j)\in[n]\times[n]}E_{i,j}.
\]
Thus the directions of all arcs incident to $v^{(j)}$
agree with the corresponding arcs incident to $v$ in
$\vec H$.
Each $E_{i,j}$ supplies two of these arcs,
and the remaining ones belong to $F$. We have $|V(\vec G)|=|V(\vec H)|n=\Theta(n)$, while the sets $E_{i,j}$ are pairwise disjoint and each has
size two. Therefore, it remains to show that $G$ is a transversal $\vec H$-gadget by checking the two conditions in \cref{transgadget}.

For every $(i,j)$, the graph $(V(\vec G),F\cup E_{i,j})$
contains a copy of $\vec H$ on
$V(\vec H^{(i)})\cup\{v^{(j)}\}$, and
this verifies the first gadget condition. As for the second condition, let $Y$ contain exactly one arc
from each $E_{i,j}$, and put
$\vec G_Y=(V(\vec G),F\cup Y)$.
Writing
\[
U_i=\{u_1^{(i)},\ldots,u_d^{(i)}\},
\quad
W_i=V(\vec H^{(i)})\setminus U_i,
\]
then the underlying graph $G_Y$ of $\vec G_Y$ is bipartite with parts
\[
\bigsqcup_{i=1}^n U_i
\quad\text{and}\quad
V_v\sqcup\bigsqcup_{i=1}^n W_i.
\]
Every vertex in $W_i$ has neighborhood exactly $U_i$,
and every vertex of $U_i$ has all its neighbors in
$W_i\cup V_v$.
Here and below, neighborhoods and degrees always refer to
the underlying undirected graph. 

Suppose for contradiction that $\vec G_Y$ contains a copy $\vec H'$ of
$\vec H$.
Let $A'$ and $B'$ be the images of $A$ and $B$ under
an isomorphism from $\vec H$ to $\vec H'$.
In each of the following cases, we prove that a correct choice of $v,u_1,u_2$ in the construction of $G$, can let $\vec H'\subseteq \vec G_Y$ lead to a contradiction for every such $Y$. The choice of $v,u_1,u_2$ is independent of $Y$, so the second gadget condition is also verified.
\begin{figure}[H]
        \centering
        \includesvg[width=0.65\linewidth]{directedKST.svg}
        \caption{The transversal $\vec H$-gadget for \cref{kstLB} in Case 1. We only present and label the full copy $\vec H^{(1)}$ of $\vec H-v$ in the figure, but the other copies $\vec H^{(i)}$ also exist for all $i\in [n]$. The blue arcs are the $\vec e_{h,i,j}$. }
\end{figure}
\noindent \textbf{Case 1: Some vertex of $A$ is WO in $\vec H$,
and no vertex of $B$ is WO.}\\\\
Choose a WO vertex $v\in A$ and any two distinct
neighbors $u_1,u_2\in B$.
Then $|U_i|=t$ and $|W_i|=s-1$ for each $i\in [n]$.
After deleting all vertices in the sets $W_i$,
every remaining arc is directed towards $V_v$ if $v$
is a sink, or away from $V_v$ if $v$ is a source.
Thus this remaining subgraph is WO.
Since $\vec H'$ is NWO, it contains a vertex
$w\in W_i$ for some $i\in[n]$.

Suppose first that $w\in B'$.
Since $N_{\vec G_Y}(w)=U_i$, we have $A'\subseteq U_i$.
Moreover, every vertex of $B'$ is adjacent to every vertex of $A'$,
so $B'\subseteq W_i\cup V_v$.
As
$|B'|=t\geq s>|W_i|,
$
some $v^{(j)}$ has to belong to $B'$.
However, $v^{(j)}$ is WO in $\vec G_Y$, and hence also
in $\vec H'$, contradicting the assumption that no
vertex of $B$ is WO.

Therefore $w\in A'$.
Since $w$ has exactly $t$ neighbors, we obtain
$B'=U_i$, and consequently $A'\subseteq W_i\cup V_v$.
Because $|A'|=s>|W_i|$, some $v^{(j)}$ belongs to $A'$.
It must be adjacent to both $u_1^{(i)}$ and $u_2^{(i)}$.
Thus both arcs of $E_{i,j}$ belong to $Y$,
contradicting the choice of $Y$.\\\\
\textbf{Case 2: Some vertex of $B$ is WO in $\vec H$.}\\\\
Choose a WO vertex $v\in B$ and any two distinct
neighbors $u_1,u_2\in A$.
Then $|U_i|=s$ and $|W_i|=t-1$ for each $i\in [n]$.
As in Case 1, deleting all vertices in the sets $W_i$
leaves a WO subgraph.
Thus $\vec H'$ contains some vertex $w\in W_i$.
This vertex has exactly $s$ neighbors in $\vec G_Y$,
namely $U_i$.
Since both bipartition classes of $\vec H'$ have at least
$s$ vertices, the class opposite $w$ must be exactly $U_i$.
The class containing $w$ therefore has $t$ vertices,
all lying in $W_i\cup V_v$.
Since $|W_i|=t-1$, it must contain some $v^{(j)}$.
This vertex is adjacent to both $u_1^{(i)}$ and
$u_2^{(i)}$, forcing both arcs of $E_{i,j}$ to belong
to $Y$, which is a contradiction.\\\\
\textbf{Case 3: No vertex of $\vec H$ is WO.}\\\\
Choose any $v\in B$.
Since $v$ is not WO, we may choose $u_1,u_2\in A$ with
\[
(u_1,v),(v,u_2)\in E(\vec H).
\]
Again, $|U_i|=s$ and $|W_i|=t-1$ for each $i\in [n]$.
After deleting all vertices in the sets $W_i$,
each remaining vertex $u_h^{(i)}$ is WO:
all its incident arcs are directed towards $V_v$ or
all are directed away from $V_v$, according to the
orientation of the edge between $u_h$ and $v$ in $\vec H$.
A copy of $\vec H$ in this remaining subgraph would
contain a vertex $u_h^{(i)}$ and hence a WO vertex,
contrary to the assumption on $\vec H$.
Therefore $\vec H'$ contains some vertex $w\in W_i$.
The same neighborhood argument as in Case 2 now applies.
Since $w$ has exactly $s$ neighbors, the bipartition
class opposite $w$ in $\vec H'$ is $U_i$.
The other class has $t$ vertices in $W_i\cup V_v$,
so it contains some $v^{(j)}$.
Both arcs of $E_{i,j}$ must then belong to $Y$,
again a contradiction.
\end{proof}

In the preceding examples, the space lower bounds scale with the extremal number. For forests, however, although $\mathrm{ex}(n,\vec H)=\Theta(n^2)$ whenever $\vec H$ is NWO, every orientation of a fixed forest admits a constant-pass algorithm using $\widetilde O(n)$ space.
The algorithm again uses color-coding and dynamic programming. For each coloring, we record whether each input vertex can serve as the root of a colorful copy of a subtree of $\vec H$ using a prescribed set of colors.

\begin{proposition}\label{forest}
    Let $H$ be a fixed forest and $\vec H$ be any orientation of $H$. Let $r$ be the largest radius of the components of $H$, then $D^{2r}(\textsc{Sub}(\vec H))=\tilde O(n)$.
\end{proposition}

\begin{proof}
    Let $\vec H_1,\ldots,\vec H_k$ be the components of $\vec H$, having radii $r_1,\ldots,r_k$, respectively. For each \nolinebreak $\vec H_i$, choose $v_i$ from the center of $H_i$ to be the root. Then for all $x\in V(\vec H_i)$, define $h(x)=r_i-\text{dist}_{H_i}(v_i,x)$. Also, denote by $\vec T_x$ the subtree of $\vec H_i$ rooted at $x$.
    
    By Lemma~\ref{ColorCoding}, we can take $L=O(\log^2 n)$ many colorings $\chi_1,\ldots,\chi_L:V(\vec G)\to [|V(\vec H)|]$ so that if $\vec G$ contains a copy of $\vec H$, then that copy is colorful under one of the colorings. This uses space $O(Ln)=\tilde O(n)$. Now for all $x\in V(\vec H)$, $c\in [L]$ and $ S\subseteq [|V(\vec H)|]$ with $|S|=|\vec T_x|$, declare an array $D_x^c[S]:V(\vec G)\to \{0,1\}$, initialized to $D_x^c[S](v)=0$ for all $v\in V(\vec G)$. This uses space $\tilde O( |V(\vec H)|L2^{|V(\vec H)|}n )=\tilde O(n)$, and our goal is to have
    \begin{align*}
        D_x^c[S](v)=1\iff \exists \text{ embedding }\varphi:\vec T_x\hookrightarrow \vec G,\text{ }x\mapsto v,\text{ }\varphi(\vec T_x)\text{ uses color }S\text{ under }\chi_c.
    \end{align*}
    
    On the other hand, for all $x,c,S$ as before, if $x$ has parent $y$, then for all $s\in [|V(\vec H)|]\setminus S$ we also declare an array $W_{y,x}^c[S,s]:V(\vec G)\to \{0,1\}$, initialized to $W_{y,x}^c[S,s](v)=0$ for all $v\in V(\vec G)$. This uses space $\tilde O(n)$, and our goal is to have $W_{y,x}^c[S,s](u)=1$ if and only if
    \begin{align*}
        &\chi_c(u)=s,\text{ and }\\& \exists v\in V(\vec G)\text{ such that } D_x^c[S](v)=1 \text{ with }\begin{cases}
            (v,u)\in E(\vec G),\text{ if }(x,y)\in E(\vec H),\\
            (u,v)\in 
            E(\vec G),\text{ if }(y,x)\in E(\vec H).
        \end{cases}
    \end{align*}
    Then we begin the dynamic programming, starting from $j=1$.\\
    \begin{enumerate}
        \item [(1)] For each leaf  $x\in V(\vec H)$, $c\in [L]$ and $S\subseteq [|V(\vec H)|]$ with $|S|=1$, we set $D_x^c[S](u)=1$ if and only if $S=\{\chi_c(u)\}$. This fixes $D_x^c[S]$ for all leaves $x$, including all $x$ with $h(x)=0$.\\
        \item [(2)]
        Before the $j$-th pass, $D_x^c[S]$ would have been fixed for all $x$ with $h(x)\leq j-1$. Then in the $j$-th pass, whenever an arc $\vec e=(u,v)$ arrives, we examine all $x\in V(\vec H)$ with $h(x)=j-1$ that have a parent, $c\in [L]$, $S\subseteq [|V(\vec H)|]$ with $|S|=|\vec T_x|$, and $s\in [|V(\vec H)|]\setminus S$.  Letting $y$ be the parent of $x$, we check if 
        \begin{align*}
            D_x^c[S](u)=1,\quad \chi_c(v)=s,\quad
        (x,y)\in E(\vec H).
        \end{align*}
        If this is the case, then we set $W_{y,x}^c[S,s](v)=1$. Likewise, we also check if
        \begin{align*}
            D_x^c[S](v)=1,\quad  \chi_c(u)=s,\quad (y,x)\in E(\vec H).    
        \end{align*}
        If this is the case, then we set $W_{y,x}^c[S,s](u)=1$.\\
        \item [(3)]
        After the $j$-th pass, $W_{y,x}^c[S,s]$ have been fixed for all $x$ with $h(x)=j-1$ and $y$ being the parent of $x$. Now for each non-leaf $y\in V(\vec H)$ with $h(y)=j$, let $y$ have children $x_1,\ldots,x_i$; leaves $y$ with $h(y)=j$ keep the arrays fixed in step (1). For all $u\in V(\vec G)$, if there exist $c\in [L]$, $s\in [|V(\vec H)|]$ and disjoint $S_1,\ldots,S_i\subseteq [|V(\vec H)|]\setminus \{s\}$ such that
        \begin{align*}
            W_{y,x_1}^c[S_1,s](u)=\ldots=W_{y,x_i}^c[S_i,s](u)=1,
        \end{align*}
        then we set $D_y^c[S_1\cup \ldots\cup S_i\cup \{s\}](u)=1$. This fixes $D_x^c[S]$ for all $x$ with $h(x)=j$. Now if $j=r$ then we move on to step (4). Otherwise we increase $j$ by $1$ and repeat from step (2).\\
        \item [(4)]
        Now $D_x^c[S]$ is known for all $x\in V(\vec H)$. In particular we know $D_{v_1}^c[S],\ldots,D_{v_k}^c[S]$. If $\vec H\subseteq \vec G$, then there will be a colorful copy of $\vec H$ under some $\chi_c$. Thus if there are some $c\in [L]$ and disjoint $S_1,\ldots,S_k\subseteq [|V(\vec H)|]$ such that 
        \begin{align*}
            D_{v_1}^c[S_1](u_1)=\ldots=D_{v_k}^c[S_k](u_k)=1
        \end{align*}
        for some $u_1,\ldots,u_k\in V(\vec G)$, then we know $\vec H\subseteq \vec G$ and go to step (5). Otherwise $\vec H\nsubseteq G$.\\
        \item [(5)]
        If $\vec H\subseteq \vec G$, then we can find a copy of $\vec H$ using $r$ extra passes and $\tilde O(|V(H)|n)=\tilde O(n)$ extra space, which is similar to step (3) in the proof of Proposition~\ref{oddcycleUB} so we only sketch the proof as follows. Now that we know there is a copy $\vec H\subseteq \vec G$ given by
        \begin{align*}
            \vec H\hookrightarrow \vec G,\quad v_1\mapsto u_1,\ldots,v_k\mapsto u_k,    
        \end{align*}
        in the first extra pass we store all the arcs incident to $u_1,\ldots,u_k$. For all $i\in [k]$, let the children of $v_i$ be $v_{i,1},\ldots,v_{i,k_i}$, then we can find some disjoint
        \begin{align*}
            Z_{i,1},\ldots,Z_{i,k_i}\subseteq [|V(\vec H)|],\text{ } Z_{i,1}\sqcup\ldots\sqcup Z_{i,k_i}=S_i\setminus \{\chi_c(u_i)\},\text{ }  |Z_{i,j}|=|\vec T_{v_{i,j}}|\text{ }  \forall j\in [k_i]
        \end{align*}
        and some $u_{i,1},\ldots,u_{i,k_i}\in V(\vec G)$ such that
        \begin{align*}
            D_{v_{i,1}}^c[Z_{i,1}](u_{i,1})=\ldots=D_{v_{i,k_i}}^c[Z_{i,k_i}](u_{i,k_i})=1\quad \forall i\in [k],
        \end{align*}
        so we know there is a copy $\vec H\subseteq \vec G$ given by
        \begin{align*}
            \vec H\hookrightarrow \vec G,\quad v_1\mapsto u_1,\ldots,v_k\mapsto u_k, \quad v_{i,1}\mapsto u_{i,1},\ldots,v_{i,k_i}\mapsto u_{i,k_i}\quad \forall i\in [k].   
        \end{align*}
        Then in the next extra pass we store all the arcs incident to $u_{i,1},\ldots,u_{i,k_i}$ for all $i\in [k]$, and so on. This recovers a copy of $\vec H$ within $r$ extra passes. \qedhere
    \end{enumerate}
\end{proof}

It remains open whether every bipartite pattern with an NWO component containing a cycle requires $\Omega(n^2/p)$ space in $p$ passes.

%% file: AppendixD.tex
\section{Omitted Proofs for Induced Subgraph Finding}\label{AppendixD}

This appendix contains the proofs omitted from \cref{sec:induced}.
Appendix~\ref{sec:induced-undirected-proofs} proves the induced-$C_4$ lower bound in \cref{prop:induced-c4} and its SM-BCAST extension.
Appendix~\ref{sec:induced-p4-proof} reviews the structure of cographs and proves the balanced decomposition in \cref{prop:balanced-cograph-phase}.
Appendix~\ref{sec:induced-oriented-proofs} proves the lower bounds for orientations of $P_3$ and $P_4$ in \cref{prop:oriented-P3-hard,prop:oriented-P4-hard}.

\subsection{Induced \texorpdfstring{$C_4$}{C4} Hardness and the SM-BCAST Extension}\label{sec:induced-undirected-proofs}

We first prove the $p$-pass streaming lower bound for induced-$C_4$ detection
in \cref{prop:induced-c4}. We then use the same construction to obtain tight
message complexity bounds for induced detection of $C_4$ and its complement
$2K_2$ in SM-BCAST.

\begin{proof}[Proof of Proposition~\ref{prop:induced-c4}]
We reduce $\textsc{Set-Disjointness}$ on $[n^2]$ to $\textsc{IndSub}(C_4)$.

Construct $G$ on four pairwise disjoint vertex sets
$U_a=\{u_{a,i}:i\in[n]\}$ and $V_a=\{v_{a,j}:j\in[n]\}$, for $a\in\{1,2\}$.
Make $U_1$ and $V_2$ cliques and $U_2$ and $V_1$ independent sets,
and add the matching edges $\{u_{1,i},u_{2,i}\}$ and $\{v_{1,j},v_{2,j}\}$ for all $i,j\in[n]$.
Given inputs $S_1,S_2\subseteq[n^2]$, add
\[
E_a=\bigl\{\{u_{a,i},v_{a,j}\}:(i-1)n+j\in S_a\bigr\},
\quad a\in\{1,2\}.
\]
There are no other edges in $G$.

Neither $G[U_1\cup V_1]$ nor $G[U_2\cup V_2]$ contains an induced $C_4$,
since each is partitioned into a clique and an independent set.
Any induced $C_4$ must therefore cross between $U_1\cup V_1$ and $U_2\cup V_2$.
Since the edges joining these two sets form a matching, the cycle contains exactly two crossing edges.
The crossing edges must consist of one $U_1$--$U_2$ edge and one $V_1$--$V_2$ edge:
otherwise, completing the cycle would require an edge inside the independent set $U_2$ or $V_1$.
The matching edges therefore force the cycle to have the form
\[
u_{1,i}-v_{1,j}-v_{2,j}-u_{2,i}-u_{1,i}
\]
for some $i,j\in[n]$.
This cycle exists only if $(i-1)n+j\in S_1\cap S_2$; conversely, any such common input coordinate
gives an induced $C_4$ on these four vertices.
Hence $G$ contains an induced $C_4$ if and only if $S_1\cap S_2\neq\emptyset$.

Alice can stream $E_1$ together with the fixed clique and matching edges, and Bob can stream $E_2$.
A $p$-pass streaming algorithm using $b$ bits of space therefore gives an $O(pb)$-communication protocol
for $\textsc{Set-Disjointness}$ on $[n^2]$ by transferring the memory state between the two parties.
Since $G$ has $4n$ vertices, the $\Omega(n^2)$ randomized communication lower bound
yields the claimed $\Omega(n^2/p)$ space lower bound.
\end{proof}

The input-dependent edges lie within the two sides, and the crossing edges are fixed.
Each party therefore knows the full neighborhoods of the vertices on its side and
can simulate their messages in the SM-BCAST model studied by Kari et al.~\cite{KariMRS15}.
Their \textsc{Index}-based construction gives an $\Omega(\sqrt n)$ lower bound for
induced-$C_4$ recognition~\cite[Proposition~4]{KariMRS15} and yields a one-pass
streaming lower bound under a direct adaptation. The \textsc{Set-Disjointness}
reduction gives the following tight bounds.

\begin{proposition}\label{prop:c4-sm-bcast}
For $H\in\{C_4,2K_2\}$, randomized induced-$H$ detection in SM-BCAST has message complexity
$\Theta(n)$ bits per vertex.
\end{proposition}

\begin{proof}
For $C_4$, we use the same graph without modification.
Alice simulates the players in $U_1\cup V_1$, while Bob simulates those in $U_2\cup V_2$.
With $b$ bits per vertex, this gives an $O(nb)$-communication protocol for
Set-Disjointness on $[n^2]$ ~\cite[proofs of Propositions~3 and~4]{KariMRS15}.
Its randomized $\Omega(n^2)$ communication lower bound gives $b=\Omega(n)$.

Sending full neighborhoods as bit vectors gives the matching $O(n)$ upper bound.
The bound for $2K_2=\overline{C_4}$ follows from the complementation equivalence
in Lemma~1 of Kari et al.~\cite{KariMRS15}.
\end{proof}

Together with Corollary~6 of Kari et al.~\cite{KariMRS15},
Proposition~\ref{prop:c4-sm-bcast} establishes $\Theta(n)$ randomized message complexity
for induced-$H$ detection in SM-BCAST for every fixed undirected pattern $H$
on at least three vertices outside $\{P_3,\mathrm{co}\text{-}P_3,P_4\}$.

\subsection{Cograph Structure and Balanced Decomposition}\label{sec:induced-p4-proof}

We prove \cref{prop:balanced-cograph-phase} using structural properties of cographs
from the parallel recognition algorithm of Nikolopoulos and Palios~\cite{cographs}.
We use their root-decomposition characterization (Lemma~2.1 and Corollary~2.1)
and their bound on connected low-degree subgraphs (Lemma~2.3), with notation
adapted to our setting. \Cref{lem:large_Bi_high} adapts their Lemma~2.4 to
handle oversized blocks when there are no middle-degree vertices.
These properties let us compute a balanced decomposition or recover an induced
$P_4$ in the streaming model.

\begin{lemma}[\cite{cographs}, Lemma~2.1 and Corollary~2.1]\label{lem:cograph_structure}
Let $G$ be a connected graph on at least two vertices. Then $G$ is induced-$P_4$-free if and only if, for every $v\in V(G)$, there exist $k\in\mathbb{N}$ and a partition of $V(G)$ into
\[
A_1,\ldots,A_k, \; B_1,\ldots,B_k
\]
such that:
\begin{itemize}
    \item $N(v)=A_1\sqcup\cdots\sqcup A_k$ and $V(G)\setminus N(v)=B_1\sqcup\cdots\sqcup B_k$;
    \item for every $i\in[k]$, every vertex of $A_i$ is adjacent to every vertex of $B_i$;
    \item for all $1\le i<j\le k$, every vertex in $A_i\cup B_i\cup B_j$ is adjacent to every vertex in $A_j$;
    \item for all $1\le i<j\le k$, there are no edges between $A_i\cup B_i$ and $B_j$; and
    \item for every $i\in[k]$, the induced subgraphs $G[A_i]$ and $G[B_i]$ are induced-$P_4$-free.
\end{itemize}
A partition satisfying the first four conditions is called a \emph{root decomposition} of $G$ at $v$, and its sets $A_i$ and $B_i$ are called \emph{blocks}.
In every root decomposition, $v\in B_1$.
\end{lemma}

\begin{figure}[H]
        \centering
        \includesvg[width=0.5\linewidth]{P_4_free_graphs.svg}
        \caption{The structure of an induced-$P_4$-free graph. A line between two displayed blocks denotes a complete join between them.}
\end{figure}

\begin{proof}
First suppose that the claimed partition exists.
Each block is a module: all of its vertices have the same neighbors outside the block.
The quotient graph obtained by contracting these blocks has all the $A_i$ pairwise adjacent, all the $B_i$ pairwise nonadjacent, and nested neighborhoods from the $B_i$ into the $A_i$.
This quotient contains no induced $P_4$.
Moreover, an induced $P_4$ meeting more than one block cannot contain two vertices from the same block, because those vertices have identical adjacency to all vertices outside the block (as $P_4$ has no nontrivial module).
It would therefore induce a $P_4$ in the quotient.
Hence every induced $P_4$ must lie within a single block, which is excluded by the last condition.

Conversely, assume that $G$ is induced-$P_4$-free. Fix a root vertex $v$.
We first observe that $G$ has diameter at most two, since any shortest path of length at least three contains an induced $P_4$.

We then partition $V(G)$ into $A=N(v)$ and $B=V(G)\setminus A$.
Every vertex in $B$ has at least one neighbor in $A$.
The absence of induced $P_4$'s imposes the following relations between the neighborhoods of vertices in $A$ and $B$.

First, if $u,w\in A$ are nonadjacent, then $N(u)\cap B=N(w)\cap B$.
Otherwise, after exchanging $u,w$ if necessary, a vertex $x\in (N(u)\setminus N(w))\cap B$ gives an induced path $x-u-v-w$.
Similarly, if $u,w\in B$ are adjacent, then $N(u)\cap A=N(w)\cap A$.
Otherwise, after exchanging $u,w$ if necessary, a vertex $x\in (N(u)\setminus N(w))\cap A$ gives an induced path $v-x-u-w$.

Finally, the neighborhoods $N(u)\cap B$, for $u\in A$, are nested; that is, any two are comparable by inclusion.
By the first neighborhood observation, it suffices to consider adjacent $u,w\in A$.
If their neighborhoods in $B$ are incomparable, choose
\[
x\in (N(u)\setminus N(w))\cap B,
\qquad
y\in (N(w)\setminus N(u))\cap B.
\]
The preceding observation forces $x\not\sim y$, since they have different neighborhoods in $A$.
Hence $x-u-w-y$ is an induced $P_4$, a contradiction.

We can now construct the required partition. Group the vertices of $A$ into classes $A_1,\ldots,A_k$ with common neighborhoods $S_i=N(A_i)\cap B$ satisfying
\[
S_1\subsetneq S_2\subsetneq\cdots\subsetneq S_k=B,
\]
where the final equality follows from the diameter bound.
Set $S_0=\emptyset$ and $B_i=S_i\setminus S_{i-1}$.
By construction, $B_i$ is complete to $A_j$ exactly when $i\le j$.
The neighborhood relations above imply that distinct $A_i$ are completely joined and distinct $B_i$ are anticomplete.
Each block induces a $P_4$-free graph, and $v\in S_1=B_1$, completing the proof.
\end{proof}

The block structure in Lemma~\ref{lem:cograph_structure} can be recovered in a constant number of streaming passes: neighborhood sizes determine a candidate partition, and edge counts verify the required adjacency relations between blocks.

\begin{lemma}\label{lem:root-decomposition}
Let $G$ be a connected graph on $m\ge2$ vertices, and let $v\in V(G)$.
An $O(1)$-pass deterministic streaming routine using $\widetilde O(m)$ space either certifies that $G$ is not a cograph by returning an induced copy of $P_4$, or returns a root decomposition $\mathcal Q_v$ of $G$ at $v$.

For a disconnected $m$-vertex graph with component vertex sets $C_1,\ldots,C_k$ and a chosen root $v_i\in C_i$ for each $i$, the routine can be run on all $G[C_i]$ in parallel, using $O(1)$ passes and $\widetilde O(m)$ total space. Isolated vertices are retained as singleton blocks.
\end{lemma}

\begin{proof}
Write $A=N(v)$ and $B=V(G)\setminus A$.
By Lemma~\ref{lem:cograph_structure}, if $G$ is a cograph, the blocks within $A$ are determined by the values $|N(u)\cap B|$.
We use one pass to store $A$ and a spanning tree of $G$, and a second pass to compute these counts. In the same pass, we count each vertex's neighbors in $A$ to group the vertices of $B$.
Order the groups on $A$ by increasing counts and those on $B$ by decreasing counts, giving $A_1,\ldots,A_k$ and $B_1,\ldots,B_\ell$.
By Lemma~\ref{lem:cograph_structure}, if $G$ is a cograph, then $k=\ell$ and this partition determines all adjacencies between distinct blocks.

If $k=\ell$, check that each vertex in $A$ and $B$ has the number of neighbors on the opposite side prescribed by the root decomposition; by induction on the block index, these counts force exactly the prescribed adjacencies between $A$ and $B$.
In one further pass, check that no edge joins distinct $B_i$ and that the number of edges joining distinct $A_i$ is $\sum_{i<j}|A_i||A_j|$.
If all checks pass, return the candidate decomposition. Otherwise, we recover an induced $P_4$ by checking the neighborhood relations used in the proof of Lemma~\ref{lem:cograph_structure}.

First, if $Z=\{b\in B:N(b)\cap A=\emptyset\}$ is nonempty, scan the stored spanning tree for an edge $xy$ with $x\in B\setminus Z$ and $y\in Z$; such an edge exists by connectedness.
In a further pass, choose any $a\in N(x)\cap A$ and return the induced path $v-a-x-y$.

Next, test whether the neighborhoods of vertices in $A$ within $B$ are nested.
Order $A=\{a_1,\ldots,a_r\}$ by nondecreasing numbers of neighbors in $B$.
Nestedness requires each vertex in $B$ to be adjacent precisely to a suffix of this order.
In one pass, scan for an edge $a_i b$ with $b\in B$ and $i\le r-|N(b)\cap A|$.
If one occurs, store $N(b)$ to find a nonneighbor $a_j$ with $j>i$.
The degree order implies that $N(a_i)\cap B$ and $N(a_j)\cap B$ are incomparable.
Store $N(a_i)$ and $N(a_j)$ to find $c\in (N(a_j)\setminus N(a_i))\cap B$.
By the argument in the proof of Lemma~\ref{lem:cograph_structure}, in every case the five vertices $v,a_i,a_j,b,c$ contain an induced $P_4$, which we return.

If neither step returns a witness, the adjacencies between $A$ and $B$ match those prescribed by the root decomposition.
The failed verification must therefore involve an edge between distinct $B$-blocks or a missing edge between distinct $A$-blocks.
We retain such an edge when encountered, or locate a missing edge by counting each vertex's neighbors in other $A$-blocks and storing one neighborhood.
The corresponding induced $P_4$ can then be recovered using the argument in the proof of Lemma~\ref{lem:cograph_structure}.

The routine uses $O(1)$ passes and stores a spanning tree, vertex sets, block labels, $O(m)$ counters, and a constant number of selected neighborhoods, using $\widetilde O(m)$ space.
For disconnected graphs, first identify the components in one pass and then process them in parallel; their disjointness preserves the total space bound and the $O(1)$ pass bound.
\end{proof}

When applying a root decomposition, we will sometimes \emph{refine} each block $Q$ into the vertex sets of the connected components of $G[Q]$, or alternatively of $\overline G[Q]$, calling the resulting sets \emph{refined blocks}.
Both operations preserve containment of every induced $P_4$ within a single block.
For an $m$-vertex graph, all blocks can be refined in parallel over the same passes using $\widetilde O(m)$ total space.
Refinement into the components of $G[Q]$ takes one pass; refinement into those of $\overline G[Q]$ takes $O(\log m)$ passes by Lemma~\ref{lem:components-or-copath} with $k=2$, unless the routine instead outputs an induced copy of $P_4$, in which case the algorithm may stop and return that copy.

Our next goal is to choose a decomposition and refine it so that every block has at most three quarters of the vertices of the graph being decomposed.
Applying this procedure to all current blocks in parallel reduces the maximum block size by a constant factor in each round, giving a recursion depth of $O(\log n)$.

The degree of the root provides a starting point.
In the root decomposition at $v$, each resulting block lies either in $N(v)$ or in $V(G)\setminus N(v)$.
Thus, choosing a root with degree between $m/4$ and $3m/4$ already gives the desired size bound.
To handle the remaining cases, we distinguish three degree classes:
\begin{align*}
\mathrm{Low}(G)  &:= \{v\in V(G):\deg_G(v)<m/4\},\\
\mathrm{Mid}(G)  &:= \{v\in V(G):m/4\le\deg_G(v)\le3m/4\},\\
\mathrm{High}(G) &:= \{v\in V(G):\deg_G(v)>3m/4\}.
\end{align*}
When no middle-degree root exists, we turn to component refinement.
To analyze this refinement, we need to understand how low and high degrees constrain component sizes.

\begin{lemma}[\cite{cographs}, Lemma~2.3]\label{lem:no_big_subgraph}
Suppose $G$ is an $m$-vertex cograph with $m\ge2$. If $F$ is a connected subgraph of $G$ with $V(F)\subseteq\mathrm{Low}(G)$, then $|V(F)|\le m/2$.
\end{lemma}

\begin{proof}
The key is that, in the nontrivial case, $V(F)$ can be covered by the neighborhoods of two vertices of $F$.
The claim is immediate when $|V(F)|=1$.
Otherwise, choose any $x\in V(F)$ and apply Lemma~\ref{lem:cograph_structure} to obtain a root decomposition at $x$ of the connected component of $G$ containing $F$.
Let $\ell$ be the largest index with $A_\ell\cap V(F)\ne\emptyset$, and choose $y$ in this intersection.
Connectedness and the block adjacency relations give
\[
V(F)\subseteq \bigcup_{i=1}^{\ell}A_i\ \cup\ \bigcup_{i=1}^{\ell}B_i.
\]
The first union lies in $N(x)$ and the second in $N(y)$. Hence
\[
|V(F)|\le \deg_G(x)+\deg_G(y)<m/2.
\]
\end{proof}

For a cograph $G$, Lemma~\ref{lem:no_big_subgraph} gives size bounds for components of $G[\mathrm{Low}(G)]$ and, by complementation, $\overline G[\mathrm{High}(G)]$, since $G$ is a cograph if and only if $\overline G$ is a cograph.
In particular, a connected low-degree subgraph violating this size bound admits no valid root decomposition, so the recognition algorithm can recover an induced $P_4$ using Lemma~\ref{lem:root-decomposition} in $O(1)$ additional passes.
For the high-degree case, we can similarly use Lemma~\ref{lem:components-or-copath} with $k=2$ to detect oversized components $F$ of $\overline G[\mathrm{High}(G)]$ and apply Lemma~\ref{lem:root-decomposition} on $G[V(F)]$ to recover an induced $P_4$ from any such component.

The main remaining case is when $\mathrm{Mid}(G)=\emptyset$ but both high-degree and low-degree vertices are present.
The following lemma helps us find a suitable decomposition in this case.

\begin{lemma}[cf.~\cite{cographs}, Lemma~2.4]\label{lem:large_Bi_high}
Suppose that $G$ is a connected $m$-vertex cograph, $\mathrm{Mid}(G)=\emptyset$, and both $\mathrm{Low}(G)$ and $\mathrm{High}(G)$ are nonempty.
Choose
\[
u\in\arg\max_{x\in\mathrm{Low}(G)}|N(x)\cap\mathrm{High}(G)|.
\]
Construct the root decomposition at $u$, and refine each block $Q$ into the vertex sets of the connected components of $G[Q]$.
If a resulting block $F$ has $|F|>3m/4$, then $F\subseteq\mathrm{High}(G)$.
\end{lemma}

\begin{proof}
We show that if $F$ contains a low-degree vertex, then some low-degree vertex in $F$ has more high-degree neighbors than $u$, contradicting its choice.

First, since each $A_i\subseteq N(u)$ and $\deg(u)<m/4$, the block $F$ must be the vertex set of a connected component of some $G[B_i]$.
Thus $N(u)\cap F=\emptyset$, so all high-degree neighbors of $u$ lie outside $F$.

We next show that every vertex of $F$ is adjacent to all high-degree vertices outside $F$.
For any $x\in\mathrm{High}(G)\setminus F$, the inequalities $\deg(x)>3m/4$ and $|F|>3m/4$ force $x$ to have a neighbor in $F$.
The block adjacency relations and the fact that $G[F]$ is a connected component of $G[B_i]$ then imply that $x\in A_j$ for some $j\ge i$.
Hence $x$ is adjacent to all of $F$.

It remains to find a low-degree vertex in $F$ with an additional high-degree neighbor inside $F$.
By Lemma~\ref{lem:no_big_subgraph}, $F$ contains a high-degree vertex.
If it also contains a low-degree vertex, the connectedness of $G[F]$ and the absence of middle-degree vertices give adjacent vertices $s\in F\cap\mathrm{Low}(G)$ and $t\in F\cap\mathrm{High}(G)$.
Therefore,
\[
|N(s)\cap\mathrm{High}(G)|
\ge |\mathrm{High}(G)\setminus F|+1
> |N(u)\cap\mathrm{High}(G)|,
\]
contradicting the choice of $u$.
\end{proof}

Together with \cref{lem:components-or-copath}, these structural properties yield a balanced decomposition algorithm with the size guarantees required for the recursion.

\begin{proof}[Proof of Proposition~\ref{prop:balanced-cograph-phase}]
First compute the connected components in one pass.
Whenever a subroutine returns an induced $P_4$, return that copy and stop.
Every induced $P_4$ lies within a single component, so if all components have size at most $3m/4$, their vertex sets already form the required decomposition.
Otherwise, only the largest component needs further decomposition; retain the others as blocks.
We then focus on this largest component and may assume that $G$ is connected.

Compute the three degree classes.
If $\mathrm{Mid}(G)\ne\emptyset$, choose $u\in\mathrm{Mid}(G)$ and return its root decomposition using Lemma~\ref{lem:root-decomposition}.
Each resulting block has size at most $3m/4$.

Now suppose that $\mathrm{Mid}(G)=\emptyset$.
If all vertices of $G$ are low-degree, connectedness and Lemma~\ref{lem:no_big_subgraph} imply that $G$ is not a cograph; apply Lemma~\ref{lem:root-decomposition} at any root to return a witness, as discussed after Lemma~\ref{lem:no_big_subgraph}.

If $V(G)=\mathrm{High}(G)$, then $\mathrm{High}(G)\subseteq\mathrm{Low}(\overline G)$.
When $G$ is a cograph, $\overline G$ is also a cograph, so applying Lemma~\ref{lem:no_big_subgraph} to $\overline G$ shows that every connected component has size at most $m/2$.
Apply Lemma~\ref{lem:components-or-copath} with $k=2$.
If a returned component has vertex set $C$ with $|C|>m/2$, apply Lemma~\ref{lem:root-decomposition} to $G[C]$ to return a witness; otherwise, return the connected-component partition of $\overline G$.
In the former case, the call to Lemma~\ref{lem:root-decomposition} indeed returns an induced $P_4$.
Every $x\in C$ satisfies $\deg_G(x)>3m/4$, so
$\deg_{G[C]}(x)\ge\deg_G(x)-(m-|C|)>|C|-m/4>|C|/2$,
and hence $G[C]$ is connected and Lemma~\ref{lem:root-decomposition} applies.
If it returned a root decomposition of $G[C]$ at some root $v$, then the last block $A_k$
would be nonempty and completely joined to $C\setminus A_k$, which contains $v$;
thus $A_k$ would be disconnected from $C\setminus A_k$ in $\overline G[C]$,
contradicting that $C$ induces a connected subgraph of $\overline G$.

Finally, suppose both $\mathrm{Low}(G)$ and $\mathrm{High}(G)$ are nonempty.
Choose $u$ as in Lemma~\ref{lem:large_Bi_high}, construct its root decomposition, and refine each block $Q$ into the connected components of $G[Q]$.
Return the partition if all blocks have size at most $3m/4$.
Otherwise, there is a unique oversized block $F$.
If $F\not\subseteq\mathrm{High}(G)$, the argument in the proof of Lemma~\ref{lem:large_Bi_high} forces $F\subseteq\mathrm{Low}(G)$: a low--high edge within $F$ would contradict the choice of $u$.
In this case, apply Lemma~\ref{lem:root-decomposition} to $G[F]$ to return a witness; otherwise, $F\subseteq\mathrm{High}(G)$.

Apply Lemma~\ref{lem:components-or-copath} again with $k=2$ to $G[F]$.
Return any induced $P_4$ it finds; otherwise, replace $F$ by the connected components of $\overline G[F]$.
Since $F\subseteq\mathrm{High}(G)\subseteq\mathrm{Low}(\overline G)$, Lemma~\ref{lem:no_big_subgraph} bounds the size of each such component by $m/2$ whenever $G$ is a cograph.
If any component has vertex set $C$ with $|C|>m/2$, apply Lemma~\ref{lem:root-decomposition} to $G[C]$ to return a witness, which is an induced $P_4$ by the same argument as above since $C\subseteq\mathrm{High}(G)$; otherwise, return the resulting partition.

Every returned block now has size at most $3m/4$, and every induced $P_4$ remains contained in a single block.
Witness recovery takes only $O(1)$ additional passes.
Processing disjoint blocks in parallel, the routine uses $O(\log m)$ passes and $\widetilde O(m)$ total space.
\end{proof}

\begin{algorithm}[!t]
\caption{Cograph recognition with induced-$P_4$ recovery in the streaming model}
\label{alg:cograph}
\small
\algtext*{EndIf}
\algtext*{EndProcedure}
\begin{algorithmic}[1]
\Statex Any induced $P_4$ found by a subroutine halts all executions and is returned with \textsc{False}.
\Statex \textsc{Recover}$(S)$ applies Lemma~\ref{lem:root-decomposition} to $G[S]$ at any root and returns its witness.
\Statex
\Procedure{IsCograph}{$G$}
    \If{$|V(G)|\le3$} \Return \textsc{True} \EndIf
    \State $\mathcal Q\gets$ \Call{BalancedDecompose}{$G$}
    \State \Return $\bigwedge_{Q\in\mathcal Q}\operatorname{IsCograph}(G[Q])$ \algcom{In parallel over shared stream scans}
\EndProcedure
\Statex
\Procedure{BalancedDecompose}{$G$} \algcom{Proposition~\ref{prop:balanced-cograph-phase}}
    \State $m\gets |V(G)|$; compute the connected-component partition $\mathcal C$ of $G$
    \If{every component has size at most $3m/4$} \Return $\mathcal C$ \EndIf
    \State Retain the other components and restrict $G$ to its largest component
    \State $m\gets |V(G)|$; compute $\mathrm{Low}(G)$, $\mathrm{Mid}(G)$, and $\mathrm{High}(G)$
    \If{$\mathrm{Mid}(G)\ne\emptyset$}
        \State Take a root decomposition $\mathcal Q$ at any $u\in\mathrm{Mid}(G)$ \algcom{Lemma~\ref{lem:root-decomposition}}
    \ElsIf{$V(G)=\mathrm{Low}(G)$} \Call{Recover}{$V(G)$}
    \ElsIf{$V(G)=\mathrm{High}(G)$}
        \State $\mathcal Q\gets$ the component partition of $\overline G$ \algcom{Lemma~\ref{lem:components-or-copath}, $k=2$}
        \State If some $C\in\mathcal Q$ has $|C|>m/2$, \Call{Recover}{$C$}
    \Else
        \State Choose $u$ as in Lemma~\ref{lem:large_Bi_high} and take its root decomposition $\mathcal Q$
        \State Refine each $Q\in\mathcal Q$ into the connected components of $G[Q]$
        \If{some $F\in\mathcal Q$ has $|F|>3m/4$}
            \State If $F\not\subseteq\mathrm{High}(G)$, \Call{Recover}{$F$}
            \State Refine $F$ into the components of $\overline G[F]$ \algcom{Lemma~\ref{lem:components-or-copath}, $k=2$}
            \State If a new component $C$ has $|C|>m/2$, \Call{Recover}{$C$}
        \EndIf
    \EndIf
    \State \Return $\mathcal Q$ together with the retained components
\EndProcedure
\end{algorithmic}
\end{algorithm}

\begin{remark}
Although Lemma~\ref{lem:components-or-copath} is effective for finding induced $\mathrm{co}\text{-}P_3$ and $\mathrm{co}\text{-}P_4$, it does not by itself decide whether $G$ contains an induced $\mathrm{co}\text{-}P_k$ for $k\ge5$.
Even when $G$ is induced-$\mathrm{co}\text{-}P_k$-free, the routine may return a connected component of $\overline G$ containing nearly all vertices, so the returned partition need not yield useful progress.
Indeed, $\mathrm{co}\text{-}P_k$ is non-bipartite for $k\ge5$, and Proposition~\ref{prop:induced-nonbipartite} gives an $\Omega(n^2/p)$ space lower bound for $p$ passes.
\end{remark}

\subsection{Hardness Reductions for Oriented \texorpdfstring{$P_3$ and $P_4$}{P3 and P4}}\label{sec:induced-oriented-proofs}

We prove the lower bounds in \cref{prop:oriented-P3-hard,prop:oriented-P4-hard}.
We first give the reductions for orientations of $P_3$, then extend them by
attaching matching leaves to handle the three orientations of $P_4$ up to
isomorphism and reversal.

\begin{proof}[Proof of Proposition~\ref{prop:oriented-P3-hard}]
Reversing all arcs in both the input and the pattern preserves induced copies and streaming resources. It therefore suffices to consider
\[
    a\to b\to c
    \qquad\text{and}\qquad
    a\to b\leftarrow c.
\]
For both reductions from $\textsc{Set-Disjointness}$ on $[n]\times [n]$, let Alice and Bob hold sets $S_1,S_2\subseteq[n]\times [n]$. We use three vertex layers
\[
    V(\vec G)=X_1\sqcup X_2\sqcup X_3,
    \qquad
    X_t=\{x^{(t)}_1,\ldots,x^{(t)}_n\}.
\]

\smallskip
\noindent\emph{Directed path $a\to b\to c$.}
The input-dependent arcs encode the two sets along $X_1\to X_2\to X_3$. Fixed arcs from $X_1$ to $X_3$ create a chord whenever the endpoint indices differ. Specifically, let $E(\vec G)=E_1\sqcup E_2\sqcup E_3$, where
\[
\begin{aligned}
E_1&=\{(x^{(1)}_i,x^{(2)}_j): (i,j)\in S_1\},\\
E_2&=\{(x^{(2)}_j,x^{(3)}_i): (i,j)\in S_2\},\\
E_3&=\{(x^{(1)}_i,x^{(3)}_k): i\neq k\}.
\end{aligned}
\]
If a directed path on three vertices exists, it must have the form $x^{(1)}_i\to x^{(2)}_j\to x^{(3)}_k$. It is induced exactly when $i=k$, in which case its two arcs encode the same element $(i,j)\in S_1\cap S_2$. Conversely, every common element gives such an induced path.

\smallskip
\noindent\emph{In-star $a\to b\leftarrow c$.}
Here the two input sets provide arcs from $X_1$ and $X_2$ into $X_3$:
\[
\begin{aligned}
E_1&=\{(x^{(1)}_i,x^{(3)}_j): (i,j)\in S_1\},\\
E_2&=\{(x^{(2)}_i,x^{(3)}_j): (i,j)\in S_2\}.
\end{aligned}
\]
To force matched sources and a center in $X_3$, we make $X_1$ and $X_2$ transitive tournaments with arcs toward larger indices, while arcs between the layers point toward smaller indices and leave matched pairs nonadjacent. The fixed arcs are
\[
\begin{aligned}
E_3&=\{(x^{(1)}_i,x^{(2)}_j): i>j\},
&E_4&=\{(x^{(2)}_i,x^{(1)}_j): i>j\},\\
E_5&=\{(x^{(1)}_i,x^{(1)}_j): i<j\},
&E_6&=\{(x^{(2)}_i,x^{(2)}_j): i<j\}.
\end{aligned}
\]

If an induced in-star exists, its two sources must lie in $X_1\cup X_2$, since vertices in $X_3$ have no out-neighbors. The sources are nonadjacent, so they must be a matched pair $x^{(1)}_i,x^{(2)}_i$. Such a pair has no common out-neighbor in $X_1\cup X_2$: the index of such a neighbor would have to be both larger than $i$ within one layer and smaller than $i$ across layers. The center therefore lies in $X_3$, and the in-star is exactly $x^{(1)}_i\to x^{(3)}_j\leftarrow x^{(2)}_i$ for some $(i,j)\in S_1\cap S_2$. Every common element also yields this in-star.

In both constructions, Alice streams $E_1$ and the fixed arcs, while Bob streams $E_2$. The usual simulation gives an $O(ps)$-communication protocol for $\textsc{Set-Disjointness}$ on $[n]\times [n]$ from a $p$-pass, $s$-space algorithm. Since the graph has $3n$ vertices, the $\Omega(n^2)$ communication lower bound yields $s=\Omega(n^2/p)$.
\end{proof}

\begin{proof}[Proof of Proposition~\ref{prop:oriented-P4-hard}]
Up to isomorphism and reversal, $P_4$ has three orientations:
\[
\vec H_1: a\to b\to c\to d,
\qquad
\vec H_2: a\to b\to c\leftarrow d,
\qquad
\vec H_3: a\to b\leftarrow c\to d.
\]
We extend the constructions of Proposition~\ref{prop:oriented-P3-hard} by attaching matching leaves. Each common input element then yields the required orientation of $P_4$; conversely, if such a copy exists, it contains an induced directed path or in-star in the original three layers.

\smallskip
\noindent\emph{Extending the directed path.}
For $\vec H_1$ and $\vec H_2$, attach a new leaf to each vertex of $X_3$ in the directed-path construction, oriented away from $X_3$ for $\vec H_1$ and toward $X_3$ for $\vec H_2$. These leaves form $X_4=\{x^{(4)}_1,\ldots,x^{(4)}_n\}$, with matching arcs
\[
\begin{aligned}
\{(x^{(3)}_i,x^{(4)}_i):i\in[n]\}
&\qquad\text{for }\vec H_1,\\
\{(x^{(4)}_i,x^{(3)}_i):i\in[n]\}
&\qquad\text{for }\vec H_2.
\end{aligned}
\]
A common element $(i,j)$ gives the required copy on $x^{(1)}_i,x^{(2)}_j,x^{(3)}_i,x^{(4)}_i$.
Conversely, if an induced copy of $\vec H_1$ exists, it must follow the four layers in order. If an induced copy of $\vec H_2$ exists, its vertex playing $b$ must lie in $X_2$, the only layer whose vertices can have both incoming and outgoing arcs; its predecessor and successor lie in $X_1$ and $X_3$. In either case, the vertices playing $a,b,c$ form an induced directed path in the original construction, implying $S_1\cap S_2\neq\varnothing$.

\smallskip
\noindent\emph{Extending the in-star.}
For $\vec H_3$, attach a new out-neighbor to each vertex of $X_2$ in the in-star construction. These leaves form $X_4$, with matching arcs
\[
    \{(x^{(2)}_i,x^{(4)}_i):i\in[n]\}.
\]
A common element $(i,j)$ yields
\[
    x^{(1)}_i\to x^{(3)}_j\leftarrow x^{(2)}_i\to x^{(4)}_i.
\]
Conversely, suppose an induced copy of $\vec H_3$ exists. No vertex of $X_4$ can play $a,b$, or $c$: it has no out-neighbors and only one in-neighbor. The vertices playing these three roles therefore form an induced in-star in the original construction, again implying $S_1\cap S_2\neq\varnothing$.

The added matching is input-independent, so the same communication simulation applies. Each graph has $4n$ vertices, giving the claimed $\Omega(n^2/p)$ lower bound.
\end{proof}